\documentclass{llncs}
\makeatletter
\let\proof\@undefined
\let\endproof\@undefined
\makeatother

\usepackage{enumerate}
\usepackage{url}
\usepackage{ulem}
\usepackage{subfigure}
\usepackage{a4wide}
\usepackage{color}
\usepackage{marginnote}

\usepackage{amsfonts}
\usepackage{amsthm}

   \usepackage[pdftex]{graphicx}
\usepackage{amsmath}

\newcommand{\cX}{{\mathcal X}}

\newcommand{\cC}{{\mathcal C}}
\newcommand{\cT}{{\mathcal T}}

\newcommand{\cS}{{\mathcal S}}

\newtheorem{observation}{Observation}

\newcommand{\cass}{\textsc{Cass}}

\newcommand{\cassinden}{\textsc{Cass$^{\textsc{DC}}$}}

\begin{document}
\title{On the elusiveness of clusters}
\author{Steven Kelk\inst{1}\thanks{Steven Kelk was 
partly funded by a Computational Life Sciences grant of The Netherlands Organisation for Scientic Research (NWO)}, Celine Scornavacca\inst{2}, Leo van Iersel\inst{3}\thanks{Leo van Iersel was funded by the Allan Wilson Centre for Molecular Ecology and Evolution.}}
\institute{Department of Knowledge Engineering (DKE), Maastricht University,\\ P.O. Box 616, 6200 MD Maastricht, The Netherlands. steven.kelk@maastrichtuniversity.nl \\ \and 
Center for Bioinformatics (ZBIT),
T\"ubingen University, 
Sand 14,\\ 72076 T\"ubingen, Germany. scornava@informatik.uni-tuebingen.de\\ \and University of Canterbury, Department of Mathematics and Statistics,\\Private Bag 4800, Christchurch, New Zealand. l.j.j.v.iersel@gmail.com}
\maketitle

\begin{abstract}
Rooted phylogenetic networks are often used to represent conflicting phylogenetic signals. 
Given a set of clusters, a network is said to represent these clusters in the \emph{softwired} sense if, for each cluster in the input set, at least one tree embedded in the network contains that cluster. 
Motivated by parsimony we might wish to construct such a network using as few reticulations as possible, or minimizing the \emph{level} of the network, i.e. the maximum number of reticulations used in any ``tangled'' region of the network.
Although these are NP-hard problems, here we prove that, for every fixed $k \geq 0$, it is polynomial-time solvable to construct a phylogenetic network with level equal to $k$ representing a cluster set, or to determine that no such network exists. However, this algorithm does not lend itself to a practical implementation.  
We also prove that the comparatively efficient \textsc{Cass} algorithm correctly solves this problem (and also minimizes the reticulation number) when input clusters are obtained from two not necessarily binary gene trees on the same set of taxa  but  does not always  minimize level for general cluster sets. Finally, we describe a new algorithm which generates in polynomial-time all binary phylogenetic networks with exactly $r$ reticulations representing a set of input clusters (for every fixed $r \geq  0$).
\end{abstract}



\section{Introduction}
\label{sec:introall}

The  
traditional abstraction for modeling evolution is the phylogenetic tree. The underlying principle of such a tree is that
the observed diversity in a set of species (or, more abstractly, a set of \emph{taxa}) can be explained by branching
events that cause lineages to split into two or more sublineages \cite{SempleSteel2003,MathEvPhyl,reconstructingevolution}. However, there is increasing attention for the situation
when observed data cannot satisfactorily be modeled by a tree. The field of phylogenetic networks has arisen with this
challenge in mind. Phylogenetic networks generalize phylogenetic trees, but within this very general characterization there
are many different definitions and models \cite{HusonRuppScornavacca10}.
In this article we
are concerned with
\emph{rooted}  phylogenetic networks.  Such networks assume that the observed data evolves from a unique starting point (the root) and that evolution is directed away from this root. The main way these networks differ from rooted phylogenetic trees is the presence of \emph{reticulation} nodes: nodes with indegree 2 or higher. 
For the remainder of this article we will use the term phylogenetic network, or just network, to refer to rooted phylogenetic networks. We refer the reader to \cite{HusonRuppScornavacca10,Nakhleh2009ProbSolv,Semple2007,husonetalgalled2009,twotrees} for detailed background
information.

Constructing a phylogenetic network that ``explains'' the observed data is a trivial problem if no optimality criteria are
imposed upon the constructed network. One simple optimality criterion that has attracted a great deal of attention in the
literature, \emph{reticulation  minimization}, is to compute a phylogenetic network that explains the observed data but using as few reticulation events (essentially, reticulation nodes) as possible.
This is an algorithmically hard problem, irrespective of the exact construction technique that is being applied \cite{twotrees}. A related optimality criterion, \emph{level minimization} \cite{JanssonSung2006,JanssonEtAl2006,lev2TCBB,reflections,tohabib2009} is motivated
by the observation that a phylogenetic network can be regarded as some kind of tree backbone decorated with \emph{tangles} of reticulate activity \cite{HusonRuppScornavacca10}. Here the challenge is to construct
a phylogenetic network that explains the observed data but such that the maximum number of reticulation events inside any biconnected component, the \emph{level}, is as low as possible. Reticulation  minimization is
thus a global optimality criterion, and level minimization is in some sense a local optimality criterion; see Figure
\ref{fig:bcc}. Both criteria will have an important role in this article.

\begin{figure}[h]
  \centering
  \includegraphics[scale=.2]{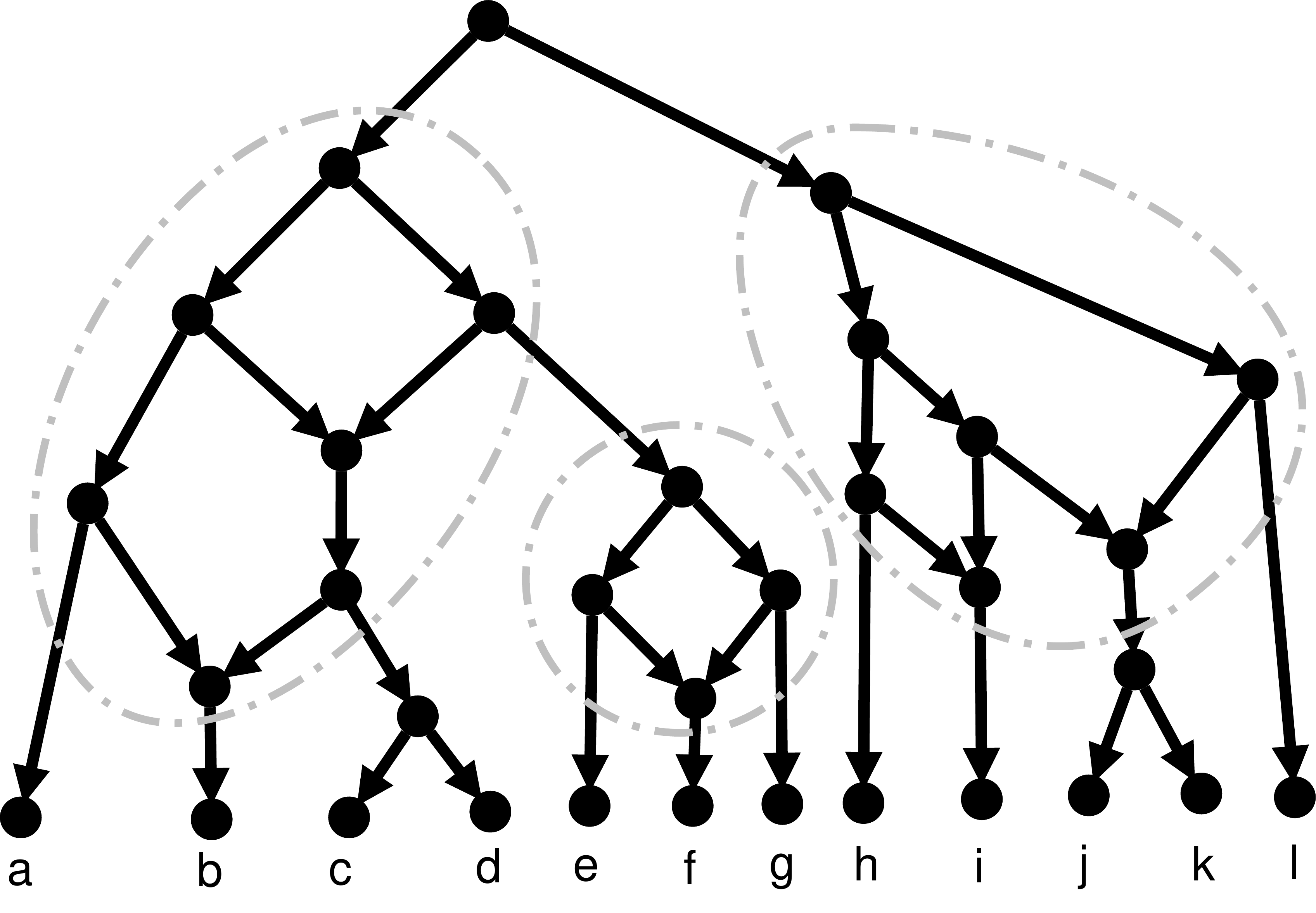}
  \caption{Example of a phylogenetic network with five reticulations. The encircled subgraphs form its biconnected components, also known as its ``tangles''. This binary network has level equal to 2 since each biconnected 
component contains at most two reticulations.}
  \label{fig:bcc}
\end{figure}

\begin{figure}[t]
  \centering
  \begin{subfigure}[]
  {
    \centering
    \includegraphics[scale=0.16]{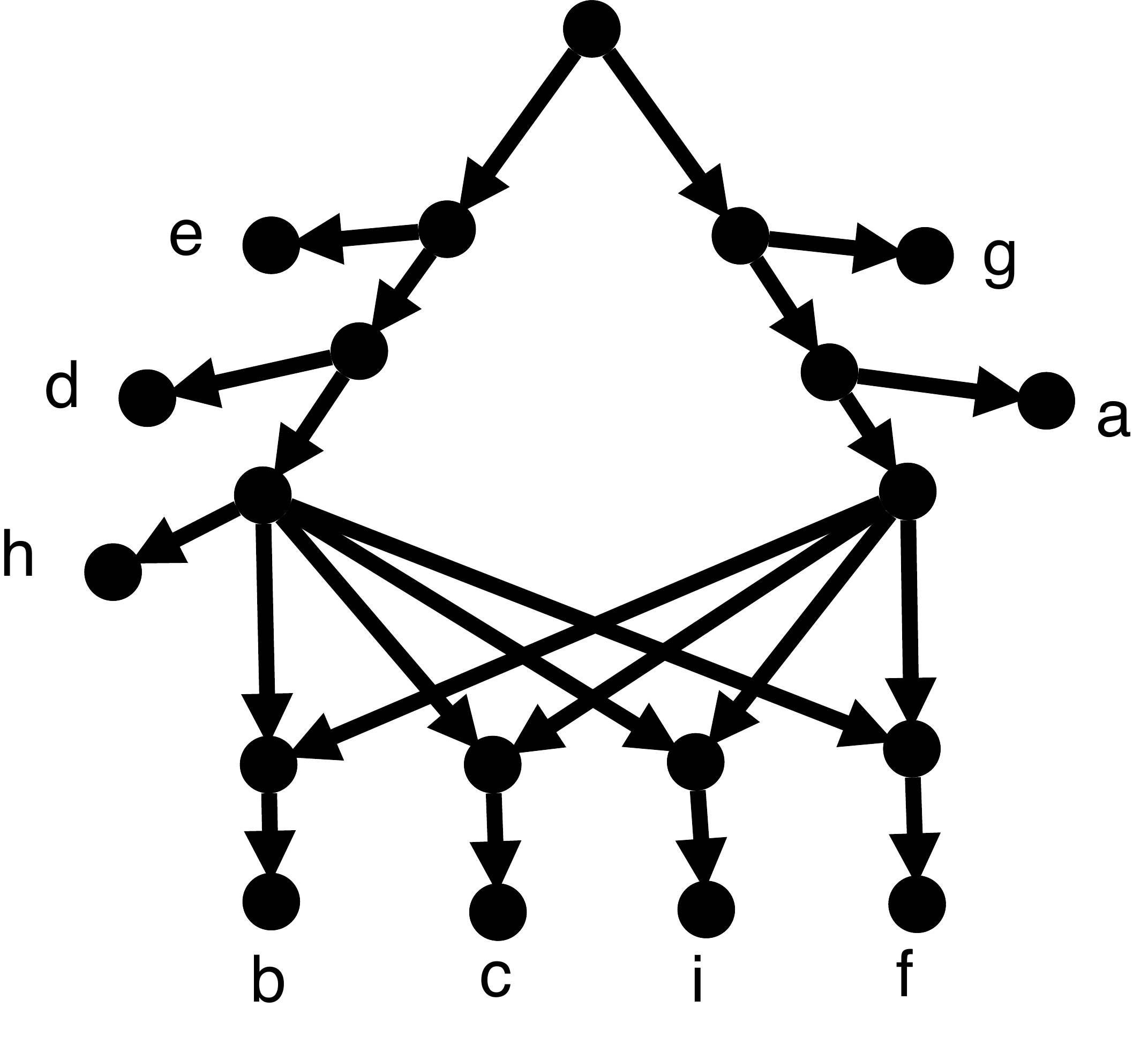}
    \label{fig:gallednetwork}
    \vspace{.5cm}
  }
  \end{subfigure}
  \begin{subfigure}[]
  {
    \centering
    \includegraphics[scale=0.16]{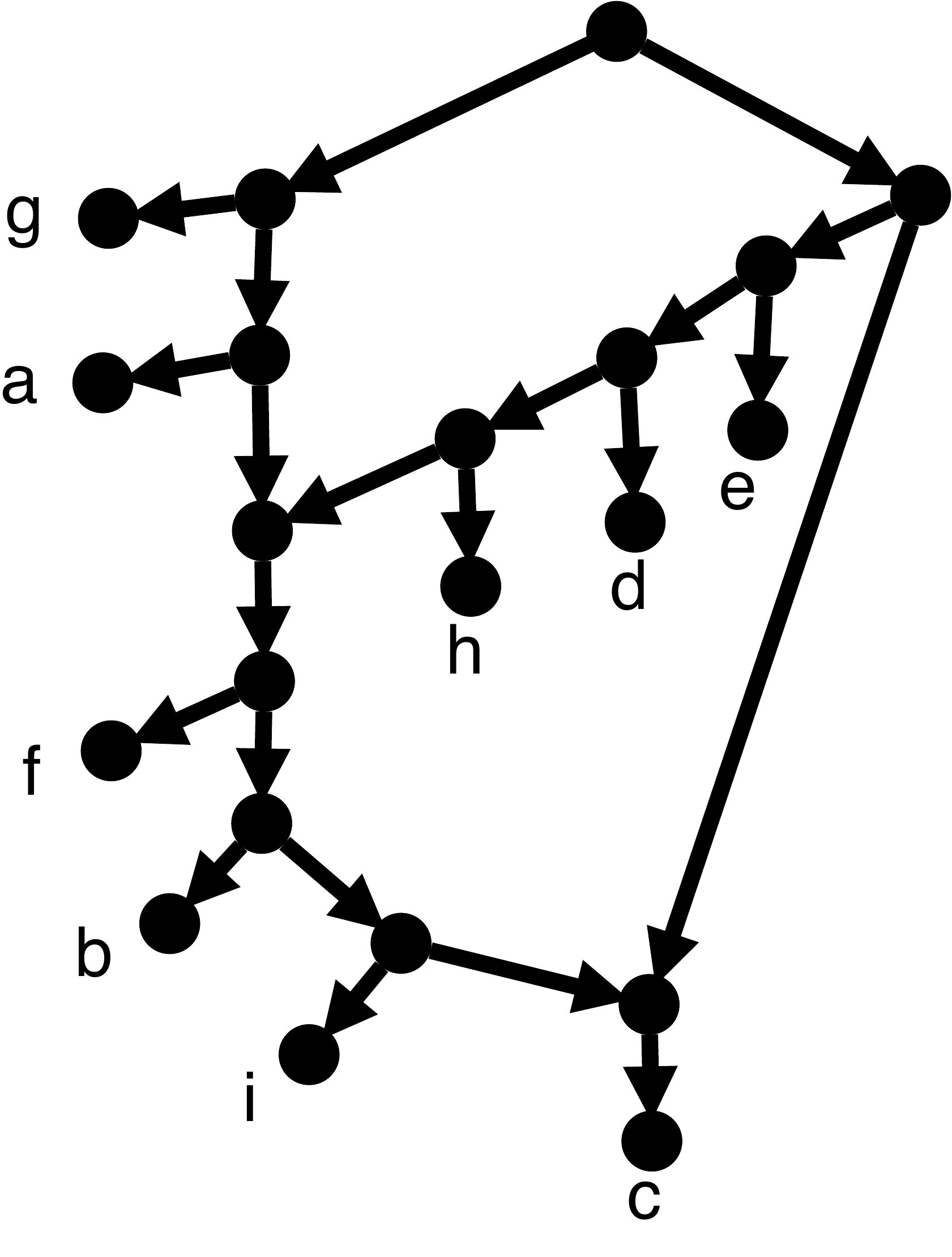}
    \label{fig:level2example}
  }
  \end{subfigure}
  \caption{(a) The output of the galled network algorithm~\cite{husonetalgalled2009} for $\mathcal{C}=\{\{a,b,f,g,i\}$, $\{a,b,c,f,g,i\}$, $\{a,b,f,i\}$, $\{b,c,f,i\}$, $\{c,d,e,h\}$, 
$\{d,e,h\}$, $\{b,c,f,h,i\}$, $\{b,c,d,f,h,i\}$, $\{b,c,i\}$, $\{a,g\}$, $\{b,i\}$, $\{c,i\}$, $\{d,h\}\}$ and (b) a (simple) network with two fewer reticulations that also represents this set of 
clusters.}
  \label{fig:example}
\end{figure}

The question remains, when does a phylogenetic network explain the observed data? This depends very much on the exact construction technique being applied. The classical problem is
motivated by the biological observation that, although the evolution of a set of organisms might best be explained by a phylogenetic network, the individual genes of the organisms
will generally undergo treelike evolution \cite{Nakhleh2009ProbSolv}. In such a case one can think of the gene trees as being \emph{displayed by} (i.e. topologically embedded within) the species network. The reticulation nodes then have an explicit biological interpretation as  (for example) hybridization, recombination or horizontal gene transfer events. Hence the following problem: given a set of rooted phylogenetic trees, all on the same set of taxa, compute a phylogenetic network with a minimum number of reticulations that displays all the input trees. The problem is already NP-hard (and APX-hard) for the case of two input trees \cite{bordewich}. However, extensive research by different authors has shown that, by exploiting the fixed parameter tractability of the problem \cite{bordewich2,sempbordfpt2007},
the two-tree problem can be solved to satisfaction for many instances \cite{quantifyingreticulation,wuISBRA2010}. The case of more than two input trees, or input trees that are not all binary (i.e. some nodes have outdegree three or higher), has been considerably less well studied \cite{pirnISMB2010,huynh}.

A parallel, and related, line of research concerns ``piecewise'' assembly of phylogenetic networks. Whereas the tree problem described above concerns the combination of
a small number of large hypotheses (e.g. gene trees) into a phylogenetic network, an alternative strategy is to combine a large number of small hypotheses into a phylogenetic
network. Examples of such small hypotheses include \emph{rooted triplets} (phylogenetic trees defined on size-3 subsets of the taxa) \cite{simplicityAlgorithmica,reflections}, (binary) characters (e.g. whether or
not the taxon is vertebrate) \cite{gusfielddecomp2007,gusfield2,WuG08,myers2003}, and \emph{clusters} (clades)
\cite{HusonKloepper2007,husonetalgalled2009,cass}. Proponents of such piecewise assembly techniques argue that in this way it is easier (than with trees) to discard
parts of the input that are not well-supported. In this article we focus specifically on clusters, although the classical tree problem and all the other piecewise construction techniques
do play a secondary role. This secondary role is linked to the fact, as observed in \cite{twotrees}, that under certain circumstances all these different models behave in a unified way. We
shall return to this point later.

Let us then say more about the cluster model. A cluster $C$ is a subset of the taxa and we say that a phylogenetic network \emph{represents} the cluster in the softwired sense if \emph{some} tree embedded in the network contains a clade equal to that cluster \cite{husonetalgalled2009}. In other words: \emph{some} tree $T$ embedded in the network has an edge such that $C$ is exactly the set of all taxa reachable from the head of that edge by directed paths. The general problem is, given a set of clusters, to construct an optimal phylogenetic network that represents all the input clusters. The set of input clusters can be constructed in an ad-hoc fashion, but often the set of clusters is generated by extracting the set of clusters induced by a set of rooted phylogenetic trees and then possibly excluding weakly supported clusters. This is the technique applied in the program \textsc{Dendroscope} \cite{Dendroscope3}. A disadvantage of this technique is that some of the topology of the original trees can be lost \cite{twotrees}, but on the plus side it permits a focus on only well-supported clades, which is a major concern of practicing phylogeneticists.

In \cite{husonetalgalled2009} it was shown, given a set of clusters, how to construct a \emph{galled network} with a small number of reticulations that represents the clusters. However, given that galled networks are a restricted subclass of phylogenetic networks, it was unclear how far that algorithm actually minimizes the number of reticulations (or the level)
when ranging over the entire space of phylogenetic networks, see for example Figure \ref{fig:example}. This is the context in which the \textsc{Cass} algorithm was developed \cite{cass}. The \textsc{Cass} algorithm, in some sense a natural follow-up to the algorithm of \cite{husonetalgalled2009}, was formally shown to produce solutions of minimum \emph{level} whenever the minimum level is at most two. However, the optimality of the \textsc{Cass} algorithm for ``higher-level'' inputs, and the performance of the algorithm in terms of minimizing number of reticulations, remained unclear. On the practical side the good news was that
\textsc{Cass} produced solutions with fewer reticulations and lower level than the algorithm from \cite{husonetalgalled2009}. Intriguingly it was also observed in \cite{cass} that, for several sets of input clusters induced by two binary trees, the networks produced by \textsc{Cass} had an identical number of reticulations to networks generated by algorithms that aim to display the trees themselves. This observation was the inspiration behind \cite{twotrees} in which it was proven that, in the case of two binary trees, the choice of construction technique (tree, triplets, characters, clusters) does not affect the number of reticulations required. However, this unification was shown to break down for data obtained from three or more binary trees.

After \cite{twotrees} several important questions about clusters remained open. Does \textsc{Cass} always minimize level? If not, can we find a different algorithm that efficiently minimizes level? Under which circumstances does \textsc{Cass} also minimize the number of reticulations? In how far do the unification results of \cite{twotrees} hold for non-binary trees?

The results in this article settle many of these open questions.
Firstly, we show that in the case of clusters obtained from two not necessarily binary trees, 
a divide and conquer algorithm using  \textsc{Cass}  as   subroutine, called here \cassinden and implemented in \textsc{Dendroscope} \cite{Dendroscope3}, 
does 
minimize both the number of reticulations, and the level. 
Spin-off results from this include non-binary versions of several unification results from \cite{twotrees}, culminating in the observation that, in the case of clusters obtained from two  trees,  {\cassinden} also
computes the minimum number of reticulations required to display the trees themselves (in the sense of \cite{linzsemple2009}), rather
than just the clusters from the trees. We also obtain deeper insights into why the two-tree case is so special and not representative for the problem on three or more trees. In particular, the two-tree case seems to be best understood as the only point at which a very natural lower bound is guaranteed to be tight.

Secondly we show that \textsc{Cass} does not, unfortunately, always minimize level when the input data requires solutions of level 3 or higher. We give an explicit 
counterexample and explain what goes wrong with the \textsc{Cass} algorithm in this case.

To offset this negative result we describe a polynomial-time algorithm that shows, for every fixed natural number $k$, how to determine whether a set of clusters can be represented by a 
network with level $k$. This algorithm, which is very different to \textsc{Cass}, is purely theoretical but does give important insights into the underlying structure of the cluster model. Also on the positive side we show that, for sets of clusters induced by arbitrarily large sets of \emph{binary} trees, a simple polynomial-time algorithm can construct \emph{all} binary phylogenetic networks with $r$ reticulations that represent
the clusters, for every fixed $r \geq 0$.  To demonstrate an important design principle first observed in \cite{twotrees} we give a practical implementation of this algorithm, \textsc{Clustistic}, which elegantly ``bootstraps'' an existing software package for merging rooted triplets into a phylogenetic network.

To summarize, the results in this article help advance our understanding of the cluster model considerably. Nevertheless, many questions remain, and in the final section of this
article we discuss a number of them. Perhaps the biggest question, which is the motivation for the title of this article, concerns the fact that the cluster model so far has not enjoyed the same kind of steady algorithmic improvements witnessed in the tree literature. Why does it seem harder to work with the clusters inside the trees than the trees themselves? To what, exactly, can the elusiveness  of clusters be attributed?

%
%

\section{Preliminaries}\label{sec:prelim}

Consider a set~$\mathcal{X}$ of taxa. A \emph{rooted phylogenetic network} (on~$\mathcal{X}$), henceforth \emph{network}, is a directed acyclic graph with a single node with indegree zero 
(the \emph{root}), no nodes with both indegree and outdegree equal to 1, and leaves bijectively labeled by~$\mathcal{X}$. In this article we identify the leaves with $\cX$. The indegree of a 
node~$v$ is denoted~$\delta^-(v)$ and~$v$ is called a \emph{reticulation} if~$\delta^-(v)\geq 2$. An edge~$(u,v)$ is called a 
\emph{reticulation edge} if its target node
$v$ is a reticulation and is called a \emph{tree edge} otherwise.
 When counting reticulations in a network, we count reticulations with more than two incoming 
edges more than once because, biologically, these reticulations represent several reticulate evolutionary events. Therefore, we formally define the \emph{reticulation number} of a 
network~$N=(V,E)$ as \[r(N) = \sum_{\substack{v\in V: \delta^-(v)>0}}(\delta^-(v)-1) = |E| - |V| + 1 \enspace.\]

A \emph{rooted phylogenetic tree} on $\mathcal{X}$, henceforth \emph{tree}, is simply a network that has reticulation number zero. 
We say that a network $N$ on $\mathcal{X}$ \emph{displays}
a tree $T$ if~$T$ can be obtained from $N$ by performing a series of node  and edge deletions and eventually by suppressing nodes with both indegree and outdegree equal to 1. 
We assume without loss of generality that each reticulation has outdegree at least one. Consequently, each leaf has indegree one. We say that a network is \emph{binary} if
every reticulation node has indegree 2 and outdegree 1 and every tree node that is not a leaf has outdegree 2.

 Proper subsets of~$\mathcal{X}$ are called \emph{clusters}, and a cluster $C$ is a \emph{singleton} if $|C|=1$. We say that an edge $(u,v)$ of 
a tree \emph{represents} a cluster $C \subset \cX$ if $C$ is the set of leaf descendants of $v$. A tree $T$ represents a cluster $C$ if it contains an 
edge that represents $C$. It is well-known that the set of clusters represented by a tree is a 
laminar set, and uniquely defines that tree. We say that a network $N$ represents a cluster $C \subset \mathcal{X}$ ``in the hardwired sense'' if there exists a tree edge  $(u,v)$ of $N$ such that $C$ is the set of leaf descendants of $v$. Alternatively, we say that  $N$ represents $C$ ``in the softwired sense'' if $N$ displays some tree $T$ on $\cX$ such that $T$ 
represents $C$. In this article we only consider the softwired notion of cluster representation and henceforth assume this implicitly.  A network represents a set of clusters $\mathcal{C}$ if it 
represents every cluster in $\mathcal{C}$ (and possibly more). 
The set of softwired  clusters of a network can be obtained as follows. 
For a network $N$, we say that a \emph{switching} of $N$ is obtained by, for each reticulation node, deleting all but one of its
incoming edges. Given a network $N$ and a switching $T_N$ of $N$, we say that an edge $(u,v)$ of $N$ represents a cluster $C$ w.r.t. $T_N$ if $(u,v)$ is an edge of $T_N$ and $C$ is the set of leaf 
descendants of $v$ in $T_N$. The set of softwired  clusters of $N$ is the set of clusters represented by all edges of $N$ w.r.t. $T_N$, where $T_N$ ranges over all possible switchings \cite{HusonRuppScornavacca10}. It is
also natural to define that an edge $(u,v)$ of $N$ represents a cluster $C$ if there exists some switching $T_N$ of $N$ such that $(u,v)$ represents $C$ w.r.t $T_N$.
Note that, in general, an edge of $N$ might represent multiple clusters, and a cluster might be represented by multiple edges of $N$.

Given a set of clusters $\cC$ on $\cX$, throughout the article  we assume that, for any taxon $x \in \cX$, $\cC$ contains at least one cluster $C$
containing $x$. 
For a set $\mathcal{C}$ of clusters on $\mathcal{X}$ we define $r(\mathcal{C})$ as $\min \{ r(N) | N \text{ represents } \mathcal{C} \}$, we sometimes refer to this as the \emph{reticulation number} 
of $\mathcal{C}$. 
The related concept of {\em level} 
 requires some more background. A directed acyclic graph is \emph{connected} (also called ``weakly connected'') if there is an undirected path 
(ignoring edge orientations) between each pair of nodes. A node (edge) of a directed graph is called a \emph{cut-node} (\emph{cut-edge}) if its removal disconnects the graph. A directed graph 
is \emph{biconnected} if it contains no cut-nodes. A biconnected subgraph~$B$ of a directed graph~$G$ is said to be a \emph{biconnected component} if there is no biconnected subgraph~$B' \neq 
B$ of $G$ that contains~$B$. A phylogenetic network is said to be a 
\emph{${\mbox{level-}\leq k}$ network}  
if each biconnected component has reticulation number less than or equal to $k$.\footnote{Note that to 
determine the reticulation number of a biconnected component, the indegree of each node is computed using only edges belonging to  this biconnected component.}
A ${\mbox{level-}\leq k}$ 
network is called a \emph{simple level-$\leq k$ network} if the
removal of a cut-node or a cut-edge creates two or more connected components of which  
at most one is non-trivial (i.e. contains at least one edge). 
A (simple) level-$\leq k$ network $N$ is called a (simple) \emph{ level-$k$ network} if  the maximum reticulation number among the biconnected components of $N$  is precisely~$k$. 
For example, the network in Figure \ref{fig:bcc} is a  level-2 network while the one in  Figure 
\ref{fig:gallednetwork} is a simple level-4 network. Note that a tree is a level-0 network. For a set $\mathcal{C}$ of clusters on $\mathcal{X}$ we define $\ell(\mathcal{C})$, the
\emph{level} of $\mathcal{C}$, as the smallest $k \geq 0$ such that there exists a level-$k$ network that represents $\mathcal{C}$.  It is immediate that for every cluster set
$\mathcal{C}$ $r(\mathcal{C}) \geq \ell(\mathcal{C})$, because a level-$k$ network always contains at least one biconnected component containing $k$ reticulations. 


We say that two clusters~$C_1,C_2\subset\mathcal{X}$ are \emph{compatible} 
if either~$C_1\cap C_2=\emptyset$ or~$C_1\subseteq C_2$ or~$C_2\subseteq C_1$. Consider a set of clusters~$\mathcal{C}$. 
The \emph{incompatibility graph}~$IG(\mathcal{C})$ 
of~$\mathcal{C}$ is the undirected graph~$(V,E)$ that has node set~$V=\mathcal{C}$
and edge set
$E=\{\{C_1,C_2\}\enspace |\enspace C_1$ \mbox{ and } $C_2$\mbox{ are incompatible\newline} \mbox{ clusters in } $\mathcal{C}$\}.
We say that a set of taxa~$\cX' \subseteq\mathcal{X}$ is 
\emph{separated} (by~$\mathcal{C}$) if there exists a cluster~$C\in\mathcal{C}$ that is incompatible 
with~$\cX'$, and \emph{unseparated} otherwise.

We say that a set of clusters $\mathcal{C}$ on $\mathcal{X}$ is  {\em separating} if it separates all sets of taxa $\cX'$ such that $\cX' \subset \cX$ and $|\cX'| \geq 2$.
We say that a set of clusters $\mathcal{C}$ on $\mathcal{X}$ is  {\em tangled} \cite{HusonRuppScornavacca10} if:
\begin{enumerate}
\item  $IG(\mathcal{C})$ is connected and has more than one node;
\item every pair of taxa $x$ and $y$ in $\mathcal{X}$ is separated by $\mathcal{C}$.
\end{enumerate}

Remember that here we assume that any taxon of $\cX$ is contained in at least one cluster $C\in \cC$.
Then, it can easily be verified that, given a tangled set of clusters $\mathcal{C}$ on $\cX$,  $\mathcal{C}$ is separating. 

The incompatibility graph and the concept of tangled clusters
are important because they highlight an important difference between (the computation of) $r(\mathcal{C})$ and $\ell(\mathcal{C})$.
In \cite{cass} the authors show that, if $\ell(\mathcal{C})=k$, then a level-$k$ network that represents $\mathcal{C}$ can be constructed  by  combining in polynomial time simple level-$\leq k$ 
networks constructed independently for each connected component
of $IG(\mathcal{C})$. The actual procedure is slightly more involved but it
shows in any case that a polynomial-time algorithm for constructing simple level-$\leq k$ networks can easily be extended to a polynomial-time algorithm for constructing level-$\leq k$
networks. We will make use of this fact in Section \ref{sec:theory}.

Unfortunately, as has been observed by several authors, the same procedure does \emph{not} necessarily lead to networks that
have reticulation number $r(\mathcal{C})$. In other words, computation of $r(\mathcal{C})$ requires something more complicated than independently optimizing each connected component
of $IG(\mathcal{C})$. An important special case, however, is when 
$\cC$ is separating;
in this case any network $N$ that represents $\mathcal{C}$ is simple (or can be trivially modified to become simple) and $r(\mathcal{C}) = \ell(\mathcal{C})$. We will formalize this in due course.
%

To conclude the preliminaries we note that, throughout the article,
we often write that an algorithm is ``polynomial time'' without formally specifying what the input size is. Unless
otherwise specified the input is a set of clusters $\cC$ on taxa set $\cX$. It is sufficient to take
$|\cC| + |\cX|$ as a lower bound on the size of the input. In some cases (such as
Lemma \ref{lem:core} in Section \ref{sec:theory}) $|\cC|$ is at most a constant factor larger than $|\cX|$ and then it is sufficient to prove a running time polynomial in $|\cX|$. In other cases a running time of the form $O( |\cC|^{a}|\cX|^{b} )$ is obtained,
for constants $a$ and $b$, and this is clearly polynomial in $|\cC| + |\cX|$ because $(|\cC| + |\cX|)^2 \geq |\cC||\cX|$.

\subsection{Structure of the article}
\label{subsec:struc}

To facilitate the mathematical exposition we build the results of this article up in a specific order, which differs from the order presented in the introduction.
We begin with Section \ref{sec:theory}: \emph{A theoretical polynomial-time algorithm for constructing level-$k$ networks}, where we prove that, for every fixed $k \geq 0$,  the problem of
determining whether a level-$k$ network that represents $\cC$ exists (and if so to construct such a
network) is solvable in polynomial time. 
 This section is, compared to the rest of the article, comparatively self-contained. In Section
\ref{sec:stsets}: \emph{From theory to practice: the importance of ST-sets} we describe several fundamental properties of the ST-set, a special structure that plays a central role throughout the rest of the article. In
Section \ref{sec:bitrees}: \emph{Clusters obtained from sets of \emph{binary} trees on $\cX$} we show how,
given a set $\cT$ of binary trees on $\cX$, and for each fixed $r \geq 0$, it is possible to construct in polynomial time
all binary phylogenetic networks with reticulation number $r$ that represent all the clusters in the input trees. We also describe \textsc{Clustistic},
which is our implementation of this algorithm built on top of already-existing software. In Section
\ref{sec:wit}: \emph{Witnesses and a natural lower bound} we further develop the theory surrounding the
ST-set, explicitly relating it to the computation of reticulation number. This is  used extensively in Section \ref{sec:casseverything}: \emph{The optimality and non-optimality of \textsc{Cass}}
where we give both positive and negative results for the {\cass} algorithm, and in the process develop a number of powerful generalizations
of unification results from \cite{twotrees}.  

\section{A theoretical polynomial-time algorithm for constructing level-$k$ networks}
\label{sec:theory} In this section we prove that, for every fixed $k \geq 0$,  the problem of
determining whether a level-$k$ network exists that represents $\cC$, and if so to construct such a
network, can be solved in polynomial time. 

We first require some auxiliary lemmas and definitions. The proofs are rather technical so we defer them to the appendix. For a node $v$ let $\cX(v) \subseteq \cX$ be the set of all taxa reachable from $v$ by directed paths. For an edge $e=(u,v)$ we define $\cX(e)$ to be equal to $\cX(v)$.

\begin{observation} \label{obs:nocutfreedom} Let $\cC$ be a 
separating 
set of clusters on $\cX$. Let
$N$ be any network that represents $\cC$. Then each node of $N$ has at most one leaf child and for
each cut-edge $(u,v)$ in $N$, $|\cX(v)|=1$ or $\cX(v)=\cX$. \end{observation}
\begin{proof}
Deferred to the appendix.
\end{proof}

\begin{lemma} \label{lem:simpleexists} Let $\cC$ be a 
separating  set of clusters on $\cX$. Let $N$ be
any network that represents $\cC$. Then there exists a simple network $N^{*}$ with at most one
leaf-child per node such that $\ell(N^{*}) \leq \ell(N)$. \end{lemma}
\begin{proof}
Deferred to the appendix.
\end{proof}

Observation \ref{obs:nocutfreedom} and Lemma \ref{lem:simpleexists} formalize the idea that any network that represents a 
separating  set of clusters
is simple or can easily be made simple by deleting certain redundant parts of it. The following lemma shows that, in terms of minimizing reticulation number or level, we can assume without loss of generality that networks are binary.

\begin{lemma} \label{lem:transfBinary} Let $N$ be a phylogenetic network on $\cX$. Then we can
transform N into a binary phylogenetic network $N'$ such that $N'$ has the same reticulation number
and level as $N$ and all clusters represented by $N$ are also represented by $N'$. \end{lemma}
\begin{proof}
Deferred to the appendix.
\end{proof}

\noindent
Note that, given a simple binary network $N$,  each node of $N$ has at most one leaf child. 
Armed with these technical results we are ready to prove the main result of this section.

\begin{lemma}
\label{lem:core}
Let $\cC$ be 
separating  set of clusters
on $\cX$. Then, for every fixed
$k \geq 0$, it is possible to determine in polynomial time whether a level-$k$ network
exists that represents $\cC$, and if so to construct such a network.
\end{lemma}
\begin{proof}
From Lemmas \ref{lem:simpleexists} and \ref{lem:transfBinary} it is sufficient to focus on simple binary networks. We assume then that, for fixed $k$, there exists a binary simple level-$k$ network
$N$ that represents $\cC$. Let $|\cX|=n$. Then $\cC$ will contain at most $2^{k+1}(n-1)$ clusters, 
because there are at most $2^{k}$ trees displayed by a 
simple 
level-$k$ network, and each tree represents at
most $2(n-1)$ clusters. Thus, for  fixed $k$, the size of the input is 
polynomial in $n$. It follows from these observations that
 whether a set of clusters is
represented by a 
given simple level-$k$ network can be 
checked in polynomial time. 

\begin{figure}[t]
  \centering
  \includegraphics[scale=.15]{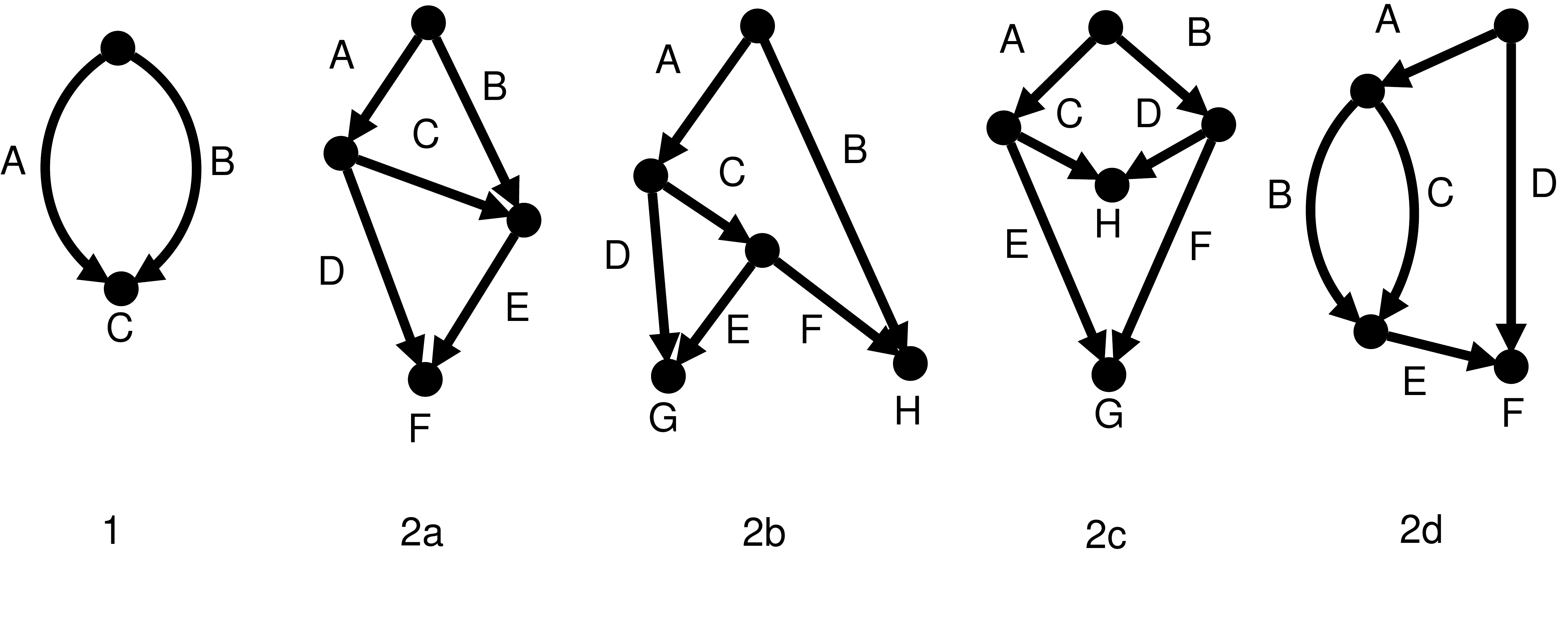}
  \caption{The single level-1 generator and the four level-2 generators. Here the sides have been labelled with capital letters.}
  \label{fig:gen}
\end{figure}

\begin{figure*}[t]
  \centering
  \includegraphics[scale=.2]{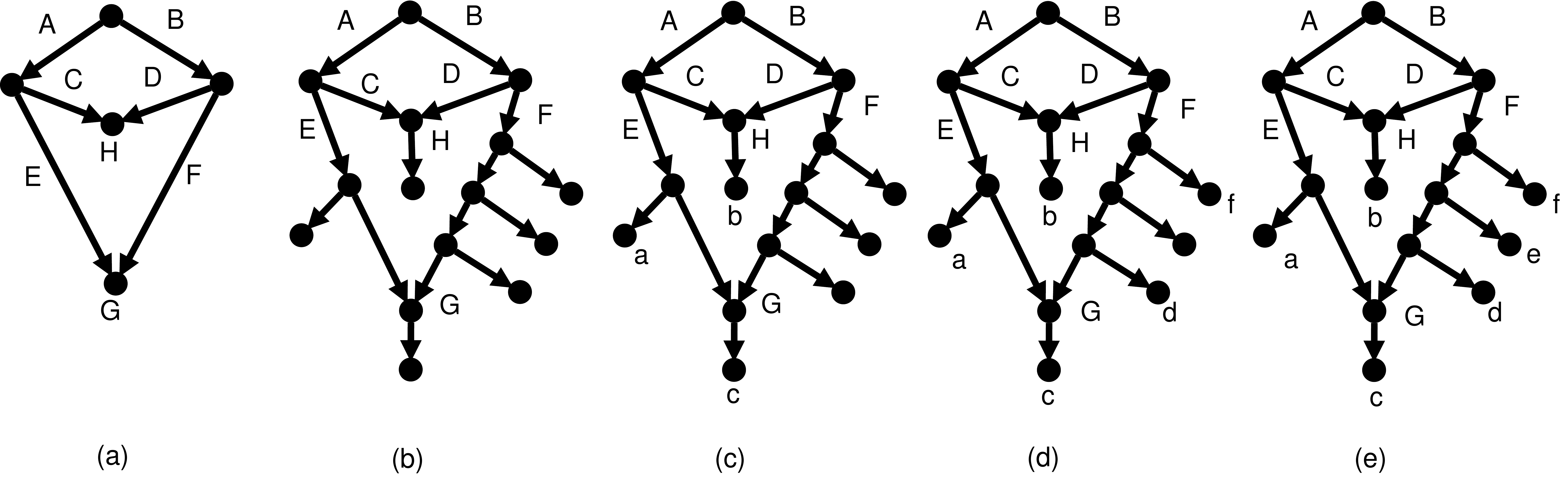}
\caption{An example of the execution of the algorithm outlined in Lemma \ref{lem:core} for the 
separating set of clusters
$\mathcal{C}=\{$$\{a,b\}$, $\{a,c\}$, $\{c,d\}$, $\{d,e\}$, $\{a,b,c\}$, $\{c,d,e\}$, $\{d,e,f\}$, $\{c,d,e,f\}$, $\{b,d,e,f\}$, $\{b,c,d,e,f\}, \{a\},\{b\},\{c\},\{d\},\{e\},\{f\}\}$ on $\cX=\{a,b,c,d,e,f\}$. The value of $k$ is fixed to 2. (a) The chosen generator $g$ is the  generator $2c$ in Figure \ref{fig:gen}. (b) We guess that the sides $E$, $H$ and $G$ contain one leaf while the side $F$ contains more than 2 leaves and all other sides contain zero leaves.  (c) We guess the single leaves on sides $E$, $H$ and $G$. (d)  We guess the leaves $s^{+}$ and $s^{-}$ on side $F$. (e) We deduce the last leaf in $F$ and we obtain a simple level-2 network on $\cX$. By a stroke of luck, our first guess represents $\cC$.}
  \label{fig:exLemma1} 
\end{figure*}

It is known that, if the leaves of $N$ are removed and all nodes with both 
indegree and outdegree
equal to 1 are suppressed, the resulting structure will
be a level-$k$ generator, defined in \cite{lev2TCBB}. See also Figure \ref{fig:gen}. For fixed $k$, there are only 
a constant
number of level-$k$ generators \cite[Proposition 2.5]{Gambette2009structure}. Recall that the \emph{sides} of a 
level-$k$ generator are
defined as the union of its edges and its nodes of indegree-2 and 
outdegree-0. For fixed $k$ the maximum number of sides ranging
over all level-$k$ generators, is a constant.

For a cluster set $\cC$ on $\cX$, we write $x \rightarrow y$ if and only if every 
non-singleton cluster in $\cC$ that contains $x$, also contains $y$. For example, for the cluster set of Figure \ref{fig:exLemma1},  we have $e \rightarrow d$ and $f \rightarrow e$. 

In the remainder of the proof, we illustrate a simple algorithm for determining whether a binary simple level-$k$ network that represents $\cC$ exists
 by attempting to reconstruct such a network. 
Let $g$ be the generator underlying $N$. 
We only require polynomially many tries
to compute $g$, because there are only a constant number of generators. So 
assume we know $g$. For each side of $g$, we guess 
whether there are 0, 1, 2 or more than 2 leaves on that side. For each side containing 
exactly one leaf, we guess what that is. 
For each side $s$
of $g$ containing 2 or more leaves, we guess the leaf $s^{+}$ that
is nearest to the root on that side, and the leaf $s^{-}$ that is
furthest from the root on that side. For an example see Figure \ref{fig:exLemma1}. Note that, since for each side we have only four options (0, 1, 2 or more than 2 leaves), and in the latter case only $s^{+}$ and $s^{-}$ have to be chosen, if follows that we have a polynomial number of guesses to try. 

We will now show how to add the remaining leaves. 
Note that we may fail to insert all leaves in the network. This means that we made the   wrong guess and that another set of guesses has to be checked. 
We say that a side $s$ is \emph{lowest} if it does not yet 
have all its
leaves, and there is no other such side $s'$ reachable from $s$. By 
reachable we mean that in the underlying generator $g$, there is a directed 
path from the head of side $s$ to the tail of side $s'$. 
Since $N$ is a directed acyclic graph, 
until all  leaves in $\cX$ have been added, there will always be
a lowest side. For example, the side $F$ in Figure \ref{fig:exLemma1}(d) is lowest.
The idea is to add leaves to the lowest side $s$, until all its leaves 
have been added. We then continue with remaining lowest
sides until we have reconstructed $N$.

Given a lowest side $s$, with $s^{+}$ and $s^{-}$ fixed, it is possible to tell in polynomial time what the correct remaining  leaves for $s$ are, as follows. 
Observe that a leaf $x$ that is on side
$s$ in $N$ and which has not yet been added has the property $s^{+} \rightarrow x \rightarrow s^{-}$.
Furthermore, there is at least one cluster $C \in \cC$ such that
$\{ x, s^{+}, s^{-} \} \cap C = \{ x, s^{-} \}$. There exists at least
one such cluster because otherwise $\{ x, s^{+} \}$ would be not separated 
in $\cC$, a contradiction since $\cC$ is 
separating. We call such a
cluster a \emph{split cluster for side $s$}.  Now, observe that for every split cluster $C$
for side $s$, and for every side $t \neq s$ that contains 2 or more leaves 
in $N$, either $\{t^{+}, t^{-}\} \cap C = \{t^{+}, t^{-}\}$ or
$\{t^{+}, t^{-}\} \cap C = \emptyset$. This follows because the
only edges in $N$ that represent $C$ lie on side $s$. 
If this is not the case, our set of guesses was incorrect and a new one has to be checked.                       

Now, consider
any leaf $y$ that has not yet been added to the network.
Assume that this leaf 
 belongs to  
side $t$ for some $t$. We want to have a simple test to avoid wrongly placing it on side $s$, with $s\neq t$. Side $t$ will 
contain three or more leaves in $N$, so we can assume that $t^{+}$
and $t^{-}$ exist. If  $s^{+} \rightarrow y \rightarrow s^{-}$   does not hold
then it is immediately
clear that $y$ cannot be put on side $s$. So assume (conversely) that 
this condition does hold, and for the same reason assume there is a split 
cluster $C$ for side $s$ that contains $y$. In other words,
there is a cluster $C$ such that $\{ y, s^{+}, s^{-} \} \cap C = \{ y, 
s^{-} \}$. 
Since $t^{+} \rightarrow y \rightarrow t^{-}$  holds, it follows that $C$ also contains $t^{+}$ and $t^{-}$,
because any cluster that contains $y$ also
contains $t^{-}$, and we know that $C$ contains either both of
$t^{+}$ and $t^{-}$, or neither of them. However, there is no edge
in $N$ that can represent $C$: the only edges that represent $C$
lie on side $s$, but the fact that $s$ is the lowest side means that
no cluster beginning on side $s$ can contain any leaves on side $t$.
To summarize, we have a simple test for determining whether
a leaf should be placed on side $s$. Once we have determined the set of 
leaves that should be placed on side $s$, it is easy to determine
the correct order of those leaves by inspecting $\rightarrow $ relationships.
Indeed, if $q$ and $p$ are two leaves that belong to side $s$, and $q$ is nearer to $s^{-}$ in the network we aim to reconstruct,  then obviously $p \rightarrow q$. Since $\cC$ is 
separating  there will be (by separation) some cluster that contains $q$ but not $p$, so $q \not \rightarrow p$. 
{If all leaves can be added in such a way, we obtain a simple level-$k$ network on $\cX$ and we can check in polynomial time if it represents $\cC$.  If it is not the case or we fail to insert at least one leaf in the network, another set of guesses can be checked until a simple level-$k$ network representing $\cC$ (if any exist) is found}.
This concludes the proof. 
\end{proof}

\noindent
The following corollary follows automatically from the generality of the proof of Lemma \ref{lem:core}.

\begin{figure*}[t]
  \centering
  \includegraphics[scale=.165]{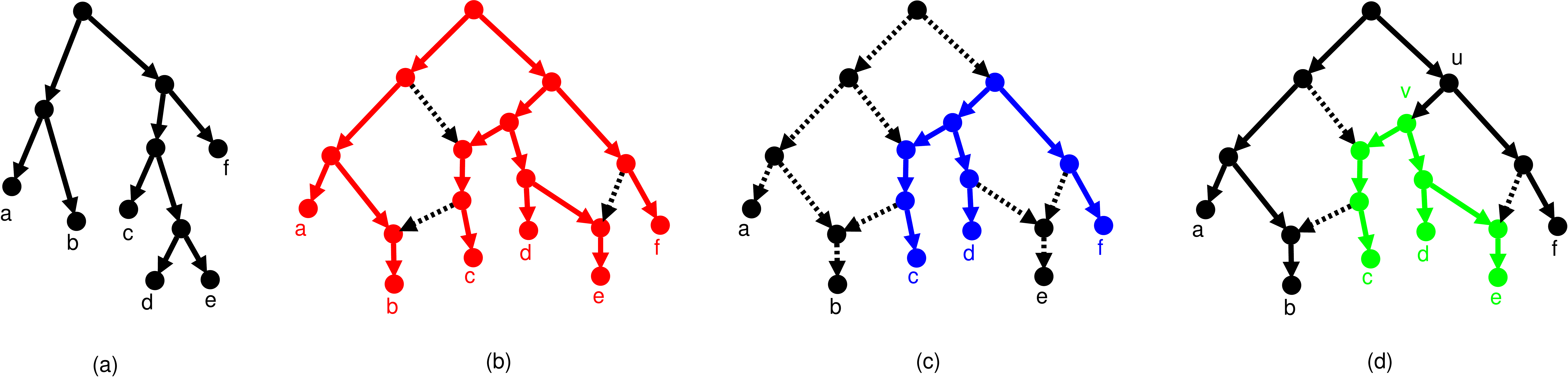}
  \caption{A phylogenetic tree~$T$ (a) and a phylogenetic network~$N$ (b,c,d); (b) illustrates in red that~$N$ displays~$T$  (deleted edges are dashed); 
(c) illustrates that~$N$ is consistent with (amongst others) the triplet $cd|f$ 
(deleted edges 
are again dashed); (d)
illustrates that~$N$ represents (amongst others) cluster $\{c,d,e\}$ in the softwired sense (dashed reticulation edges are ``switched off'').}
  \label{fig:treestripletsclusters}
\end{figure*}

\begin{corollary}
Let $\cC$ be 
a separating set of clusters 
on $\cX$. Then, for every fixed
$k \geq 0$, it is possible to construct in polynomial time \emph{all} binary simple level-$k$ networks
that represent $\cC$.
\end{corollary}

Using Lemma \ref{lem:core} we can prove the following result:
\begin{theorem}
\label{thm:core}
Let $\cC$ be a  (not necessarily 
separating)  set of clusters on $\cX$. Then, for every
fixed $k \geq 0$, it is possible to determine in polynomial time whether a level-$k$
network exists that represents $\cC$, and if so to construct such a network.
\end{theorem}
\begin{proof}
Recall that all tangled cluster sets are separating.  
It was shown in \cite{cass} that the existence
of a polynomial-time algorithm for constructing a level-$\leq k$ network from a tangled 
cluster set, is sufficient to give a polynomial-time algorithm for 
constructing level-$k$ networks from general cluster sets. (Specifically, several tangled 
cluster sets are obtained by processing each non-trivial connected component of the 
incompatibility graph of the original cluster set \cite{cass,HusonRuppScornavacca10}).
Hence we can assume
without loss of generality that $\cC$ is tangled.
 Lemma
\ref{lem:core} is thus sufficient, and we are done.
 
\end{proof}

\subsection{Rooted triplets}

It is interesting to note that the proof technique used in Lemma \ref{lem:core} leads to a simplified
proof, presented in the following corollary, of a complexity result that was first
proven in \cite{simplicityAlgorithmica}. (The algorithm in \cite{simplicityAlgorithmica} yielded a much faster running time, however). Let us first recall several definitions 
related to rooted triplets. A \emph{(rooted) triplet} on~$\cX$ is a binary phylogenetic 
tree on a size-3 subset of~$\cX$. We use $xy|z$ to denote the 
triplet with taxa~$x,y$ on one side of the root and $z$ on the other side of the root. For triplets, the notion of ``represent'' can be formalized by 
the notion of ``display'' introduced above. However, for triplets ``consistent with'' is often used instead of ``displayed by''. A triplet~$xy|z$ is 
\emph{consistent} with a phylogenetic network~$N$ (and~$N$ is \emph{consistent} with~$xy|z$) if~$xy|z$ is displayed by~$N$. See 
Figure~\ref{fig:treestripletsclusters} for an
example. Given a phylogenetic tree~$T$ on~$\cX$, we let
~$Tr(T)$ 
denote the set of all rooted triplets on~$\cX$ that are consistent with~$T$. For a set
of phylogenetic trees~$\mathcal{T}$, we let~
$Tr(\mathcal{T})$
 denote the set of all rooted triplets that are consistent with some tree
in~$\mathcal{T}$, i.e. 
$Tr
(\mathcal{T})=\bigcup_{T\in\mathcal{T}}
Tr
(T)$. A set of triplets on $\cX$ is \emph{dense} if, for every size-3 subset
$\{x,y,z\} \subseteq \cX$, at least one of $xy|z$, $xz|y$, $yz|x$ is in the triplet set.

\begin{corollary}
Let $R$ be a dense set of triplets on $\cX$. Then, for every fixed
$k \geq 0$, it is possible to determine in polynomial time whether a binary simple level-$k$ network
exists that is consistent with $R$, and if so to construct such a network.
\end{corollary}
\begin{proof}
As pointed out in \cite{simplicityAlgorithmica} it is possible to
determine in polynomial time whether a given network is indeed consistent with a set of input triplets. Then, the proof of Lemma \ref{lem:core} holds here almost entirely. 
The only
significant difference concerns the adding of leaves to the lowest side:
 a not yet allocated leaf $x$ belongs on 
lowest side $s$ if and only if the triplet $s^{-}x|s^{+}$ is in the input. 
\end{proof}

We shall return to rooted triplets again later in the article. 

\section{From theory to practice: the importance of ST-sets}
\label{sec:stsets}

The algorithm described in Section \ref{sec:theory} is polynomial time but only of theoretical interest because its running time is too high to be useful in practice. In the
rest of this article we will focus on practical polynomial-time algorithms. In all these algorithms the ST-set, which can informally be thought of as treelike subsets of $\cX$, has a central role. We begin  by formally defining ST-sets and describing their basic properties. We will expand upon these basic properties in subsequent sections of the article.

\subsection{Definition and basic properties of ST-sets}
\label{subsec:stdefs}


Given a set~$S\subseteq\mathcal{X}$ of taxa,  
 we use~$\mathcal{C}\setminus S$ to denote the result of removing all elements of~$S$ from each cluster
in~$\mathcal{C}$ and we use~$\mathcal{C}|S$ to denote~${\mathcal{C}\setminus (\mathcal{X}\setminus S)}$ (the restriction of~$\mathcal{C}$ to~$S$). We say that a set~$S \subseteq \mathcal{X}$ 
is an \emph{ST-set} with respect to $\mathcal{C}$, if~$S$ is not separated by~$\mathcal{C}$ and any two clusters~$C_1,C_2\in\mathcal{C}|S$ are 
compatible. Note that, unlike in \cite{cass}, we allow the possibility that $S = \emptyset$  
or $S = \cX$. 
(We say that an ST-set $S$ is \emph{trivial} if $S= \emptyset $ or $S= \cX$).
 An ST-set~$S$ is 
\emph{maximal} if there is no ST-set~$T$ with~$S \subset T$.

Informally, the maximal ST-sets are the result of repeatedly collapsing pairs of unseparated taxa for as long as possible; we can think
of them as ``islands of laminarity'' within the cluster set.
ST-sets first explicitly appeared in \cite{cass} but, as we shall see in due course, they implicitly arose 
earlier in the \emph{recombination network} literature. An important feature of ST-sets is that there can in general be very many of them. For example, 
suppose $\mathcal{C}$ contains only $|\mathcal{X}|=n$ singleton
clusters; then $\cC$ has $2^n$ ST-sets. However, as the following technical results show, $\cC$ will have at most $n$ 
\emph{maximal} ST-sets, and they will partition $\mathcal{X}$ i.e. they
are mutually disjoint and entirely cover $\mathcal{X}$.  Several of the proofs  have been deferred to the appendix. 

\begin{lemma} \label{lem:stcompute} 
Let $\mathcal{C}$ be a set of clusters on $\cX$ and let $S_1 \neq S_2$ be two ST-sets of $\mathcal{C}$. If $S_1 \cap S_2 \neq \emptyset$ then $S_1 \cup S_2$ is an ST-set. \end{lemma} 
\begin{proof}
Deferred to the appendix. 
\end{proof}

\begin{corollary}
\label{cor:maxstdisjoint}
Let $\mathcal{C}$ be a set of clusters on $\cX$ and let $S_1 \neq S_2$ be two maximal ST-sets of $\mathcal{C}$. Then $S_1 \cap S_2 = \emptyset$.
\end{corollary}

\begin{corollary}
\label{cor:mostn}
Let $\mathcal{C}$ be a set of clusters on $\cX$. Then there are at most $n$ maximal ST-sets with respect to $\cC$, they are uniquely
defined and they partition $\cX$.
\end{corollary}
\begin{proof}
The disjointness of maximal ST-sets guarantees that there at most $n$ of them. Consider the set $\cS$ of all (necessarily disjoint) maximal ST-sets of $\cX$. Suppose the maximal ST-sets in 
$\cS$ do not entirely cover $\cX$. Then there is some $x \in \cX$ which is disjoint from all maximal ST-sets in $\cS$. Let $S$ be the ST-set of largest 
cardinality that contains $x$; such an ST-set must exist because $\{ x \}$ is an ST-set. $S \not \in \cS$ so there exists some ST-set $S'$ such that $S \subset S'$, but this contradicts
the assumption on the cardinality of $S$. Hence $\cS$ partitions $\cX$, and by extension $\cS$ is unique. 
\end{proof}

\begin{lemma}
\label{lem:polymax}
The maximal ST-sets of a set of clusters $\cC$ on $\cX$ can be computed in polynomial time.
\end{lemma}
\begin{proof}
Deferred to the appendix. 
\end{proof}

Corollary \ref{cor:mostn} and Lemma \ref{lem:polymax} are perhaps not so surprising, but we have nevertheless proven them rigorously to highlight the fact that \emph{computing} maximal ST-sets is not a complexity bottleneck. In later sections we shall see that there is a link between NP-hardness and maximal ST-sets, but that
the hardness lies in selecting certain maximal ST-sets with special properties, not in the computation of the maximal ST-sets \emph{per se}.

Let $\mathcal{T}$ be a set of trees, where each $T \in \mathcal{T}$ is a tree on $\mathcal{X}$. For a tree $T$ we write $Cl(T)$ to denote the set of clusters induced
by edges
of $T$ i.e. $C \in Cl(T)$ if and only if some edge of $T$ represents $C$. We let $Cl(\mathcal{T}) = \cup_{T \in \mathcal{T}} Cl(T)$. Whenever we (reasonably) assume 
that all singleton clusters are present in the input\footnote{The presence or absence of the singleton clusters in the input does not change the complexity of
the problems we study because it is trivial to modify a network without raising its reticulation number or level such that it also represents all the singleton clusters.}, it is easy to see that every cluster
set $\mathcal{C}$ on $\cX$ can be written as $Cl(\mathcal{T})$ for some $\mathcal{T}$ as follows. We 
take any proper coloring of $IG(\mathcal{C})$ (i.e. map the nodes of $IG(\cC)$ to colors such that no
two adjacent nodes have the same color) and use the resulting colors to partition
$\mathcal{C}$. Clusters that have been colored the same are all mutually compatible, so can be represented by a single tree corresponding to that color. Finally, whenever a subset of clusters pertaining to a color
does not cover all elements of $\mathcal{X}$, the missing taxa can be attached to the root. An obvious corollary of this is that the chromatic number of $IG(\mathcal{C})$ is
a lower bound on the cardinality of $\mathcal{T}$.

Whenever $\cC=Cl(\mathcal{T})$ there is an important relationship between the nodes and edges of trees in $\mathcal{T}$, and the (maximal) ST-sets of $\mathcal{C}$. Let $T$ be a (not necessarily binary) tree on $\cX$. In Section \ref{sec:prelim} 
we defined when an edge of a tree
represents a cluster. Here we extend this definition to \emph{nodes} of trees. We say that a node $v$ of $T$
represents $C$ if $C$ is equal to the union of the clusters represented by some (not necessarily strict) subset of its 
outgoing edges. Note that if an edge $(u,v)$ represents a cluster $C$ then so does $v$. 


\begin{lemma}
\label{lem:nonbinaryunseptree}
Let $\cT = \{T_1, \ldots, T_{m}\}$ be a set of trees on $\cX$. Let $\emptyset \subset {\cX}'
\subset \cX$ be an unseparated set with respect to $Cl(\cT)$. Then for each $T_i \in \cT$ there exists an edge $e_i$ or a node $v_i$ in
$T_i$ such that $e_i$ or $v_i$ represents ${\cX}'$.
\end{lemma}
\begin{proof}
Consider an arbitrary tree $T_i \in \cT$. Every cluster in $Cl(T_i)$ is either disjoint from ${\cX}'$,
 a superset of it, or a subset of it, otherwise ${\cX}'$ would
be separated. If some edge of $T_i$ represents ${\cX}'$ then we are done. Otherwise consider some node $v$ furthest from the root which 
represents a cluster $C$ such that ${\cX'} \subset C$. Such a $v$ must exist because if necessary we can take the root as $v$. $C$ is equal to 
the union of the clusters represented by some subset of the edges outgoing from $v$. Each cluster represented
by an outgoing edge of $v$ is either disjoint from ${\cX}'$ or a subset of it, because of the assumption on the distance of $v$ from the root. For the same reason, 
${\cX}'$ intersects with at least two such outgoing edge clusters of $v$. But when ${\cX}'$ intersects with such a cluster it
must contain it entirely, so  ${\cX}'$ is equal to the union of some subset of the clusters represented by edges outgoing from $v$. Hence $v$ represents
${\cX}'$. 
\end{proof}

The following corollary is automatic.

\begin{corollary}
\label{cor:unseptree}
Let $\cT = \{T_1, \ldots, T_{m}\}$ be a set of binary trees on $\cX$. Let $\emptyset \subset {\cX}'
\subset \cX$ be an unseparated set with respect to $Cl(\cT)$. Then for each $T_i \in \cT$ there exists an edge $e_i$  such that $e_i$ represents ${\cX}'$.
\end{corollary}

Note that Lemma \ref{lem:nonbinaryunseptree} and Corollary \ref{cor:unseptree} hold in particular for (maximal) ST-sets, because all ST-sets are unseparated.
Hence the two following straightforward extensions to ST-sets, which we will use extensively in the next section. Recall that an ST-set $S$ is trivial if $S=\emptyset$ or $S=\cX$.

\begin{corollary}
\label{cor:stsetbinary}
Let $\cT = \{T_1, \ldots, T_{m}\}$ be a set of binary trees on $\cX$. Let $S$ be a non-trivial ST-set  
with respect to $Cl(\cT)$. Then for each $T_i \in \cT$ there exists an edge $e_i$ in
$T_i$ such that $e_i$ represents $S$.
\end{corollary}

\begin{corollary}
\label{cor:stedges}
Let $\cT = \{T_1, \ldots, T_{m}\}$ be a set of binary trees on $\cX$. Then $Cl(\cT)$ contains at most
$2(n-1)$ non-trivial ST-sets, and for every such ST-set $S$ of $Cl(\cT)$ and every tree $T_i \in \cT$ there exists a unique edge $e_i$ of
$T_i$ such that $\cX(e_i)=S$ and such that the subtree rooted at the head of $e_i$ is the unique tree that represents exactly the cluster set 
$Cl(\cT)|S$.  
\end{corollary}
\begin{proof}
Given any tree $T_i  \in \cT$, $Cl(T_i)$ will contain exactly $2(n-1)$ edges and consequently $2(n-1)$ clusters. Since an ST-set is by definition unseparated, each ST-set of $Cl(\cT)$ is a subset of the cluster set $Cl(T_i)$
(for any $i$) and there are thus at most
$2(n-1)$ ST-sets. Moreover, by definition, for any ST-set $S$ of $Cl(\cT)$ we have that the set $Cl(\cT)|S$ is compatible. Since two non-isomorphic binary trees on the same taxa set induce at least two incompatible clusters, this concludes the proof.
\end{proof}

Informally Corollary \ref{cor:stedges} states that when $\cT$ consists solely of binary trees,  each non-trivial ST-set corresponds to some subtree that is common to all trees in $\cT$.


\section{Clusters obtained from sets of \emph{binary} trees on $\cX$}
\label{sec:bitrees}

Let $\cT$ be a set of binary trees on $\cX$. In this section we prove prove that, for fixed $r \geq 0$, it is possible to construct in polynomial time
all binary phylogenetic networks with reticulation number $r$ that represent $Cl(\cT)$.  We also describe our new program \textsc{Clustistic}
which implements this algorithm.  \textsc{Clustistic} is in itself an important ``proof of concept'': it has been rapidly prototyped by, building on insights from \cite{twotrees}, slightly modifying
existing software that was originally conceived to reconstruct binary level-$k$ networks not  from clusters but rooted triplets. 




Let $N$ be a network on $\cX$ and let $T'$ be some tree on $\cX' \subset \cX$. We say that $T'$ is a \emph{Subtree Below a 
Reticulation} (SBR) of $N$ if there is a reticulation node $v$ in $N$ such that no reticulation nodes $v' \neq v$ 
are reachable from $v$ by a directed path, and that the subnetwork rooted at $v$ 
(or, when $v$ has outdegree exactly 1, the child of $v$) 
is exactly equal to $T'$. It is easy to show that (by virtue of its acyclicity) every 
network contains at least one SBR 
\cite{simplicityAlgorithmica}. A simple though critical observation is:
\begin{observation}
\label{obs:SBR_ST}
If $T'$ is an SBR of a network $N$ on $\cX$ that 
represents a cluster set $\cC$, and $T'$ has taxa set $\cX'$, then $\cX'$ is an ST-set with respect to $\cC$. 
\end{observation}
We will 
make repeated use of this throughout the rest of the article. (Note however that in general an ST-set will not specify
a unique SBR).

Given a network $N$ with an SBR $T$
we denote by $N\setminus_{T}$ the network obtained from $N$ by deleting $T$ and for as long as necessary
applying the following tidying-up operations until they are no longer needed: deleting any node with outdegree zero that is not labelled by an element of $\cX$; suppressing all nodes with indegree and outdegree both equal to 1;
replacing multi-edges with single edges; deleting nodes with indegree-0 and outdegree-1. 

We call $(S_1, S_2, ..., S_p)$  $(p \geq 0)$
a \emph{(maximal) ST-set sequence of $\cC$} if $S_1$ is a (maximal) ST-set of $\cC$,
$S_2$ is a (maximal) ST-set of $\cC \setminus S_1$, $S_3$ is a (maximal) ST-set of $\cC \setminus S_1 \setminus S_2$
and so on. Such a sequence is additionally a \emph{tree} sequence if all the clusters in $\cC \setminus S_1 \setminus \ldots \setminus S_p$
are mutually compatible i.e. can be represented by a tree. We denote by $p$ be the \emph{length} of the sequence; $p=0$ denotes the empty tree sequence.

\begin{lemma}
\label{lem:treeseq}
Let $N$ be a network that represents some cluster set $\cC$ on $\cX$. Then there exists 
a sequence of SBRs that need to be removed to prune $N$ into a tree, and this corresponds to an ST-set tree sequence of $\cC$ 
of length $r(N)$.
\end{lemma}
\begin{proof}
If $N$ is a tree then there is an empty tree sequence. Otherwise, $N$ has an SBR $T'$ with taxa set 
$\cX'$  
and  the network $N\setminus_{T'}$ represents the cluster set $\cC \setminus \cX'$. Clearly 
$r(N') < r(N)$. By observation \ref{obs:SBR_ST}, we let $S_1$ equal $\cX'$ and 
$d = r(N)-r(N')$. If $d>1$ then we let $S_2, \ldots, S_d$ all equal $\emptyset$, 
the empty ST-set. We use these empty ST-sets to model the situation when removing the SBR would cause multiple reticulation nodes to disappear simultaneously. 
Now, $N'$ also has at least one SBR, so we can iterate the whole process. 
We can repeat this until we obtain a tree: at this point we will have an ST-set tree sequence of length exactly 
$r(N)$.   
\end{proof}

In the following observation we make use of the fact that the operation $N \setminus_{T}$  is also meaningfully defined when the tree $T$ is  exactly the subtree rooted at some node of $N$. 

\begin{observation}
\label{obs:staysbinary}
Let $\cT = \{T_1, \ldots, T_{m}\}$ be a set of binary trees on $\cX$. Let $S$ 
be a non-trivial
ST-set of $Cl(\cT)$. Then $Cl(\cT)\setminus S = Cl(\cT')$ where $\cT'$ is a set of at most 
$m$ binary trees $\{T'_1, \ldots, T'_{m}\}$ on $\cX \setminus S$ with $T'_i= T_i\setminus T_v$, where $e_i = (u,v)$ is the edge of $T_i$ that represents $S$ (which exists by Corollary \ref{cor:stedges}) and $T_v$ is the subtree rooted at $v$. 
\end{observation}



\begin{figure}[t]
  \centering
  \includegraphics[scale=.15]{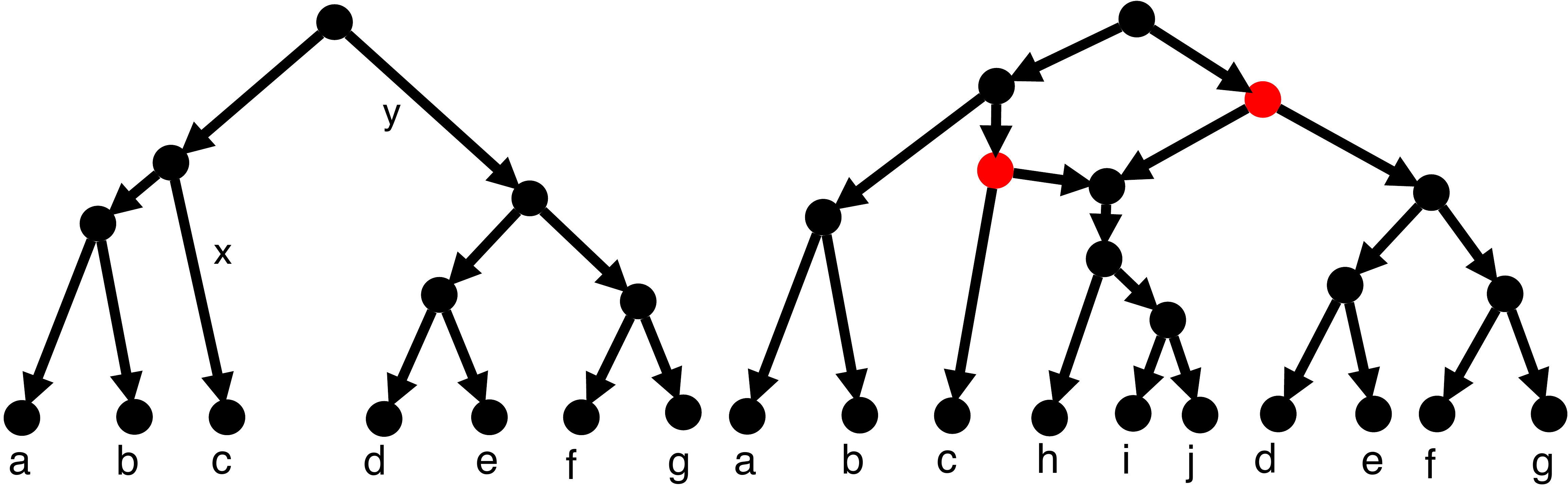}
  \caption{Let $\cC$ be some set of clusters represented by the network on the right. If we guess that ST-set $S=\{h,i,j\}$ corresponds to an SBR and remove it, we obtain
the tree on the left. To reverse the process we add (a tree corresponding to) $S$ below a reticulation node whose incoming edges subdivide edges $x$ and $y$.
}
  \label{fig:hangback}
\end{figure}

\begin{figure*}
 \centering
\begin{tabular}{c}
\begin{tabular}{ccccc}
\includegraphics[scale=0.16]{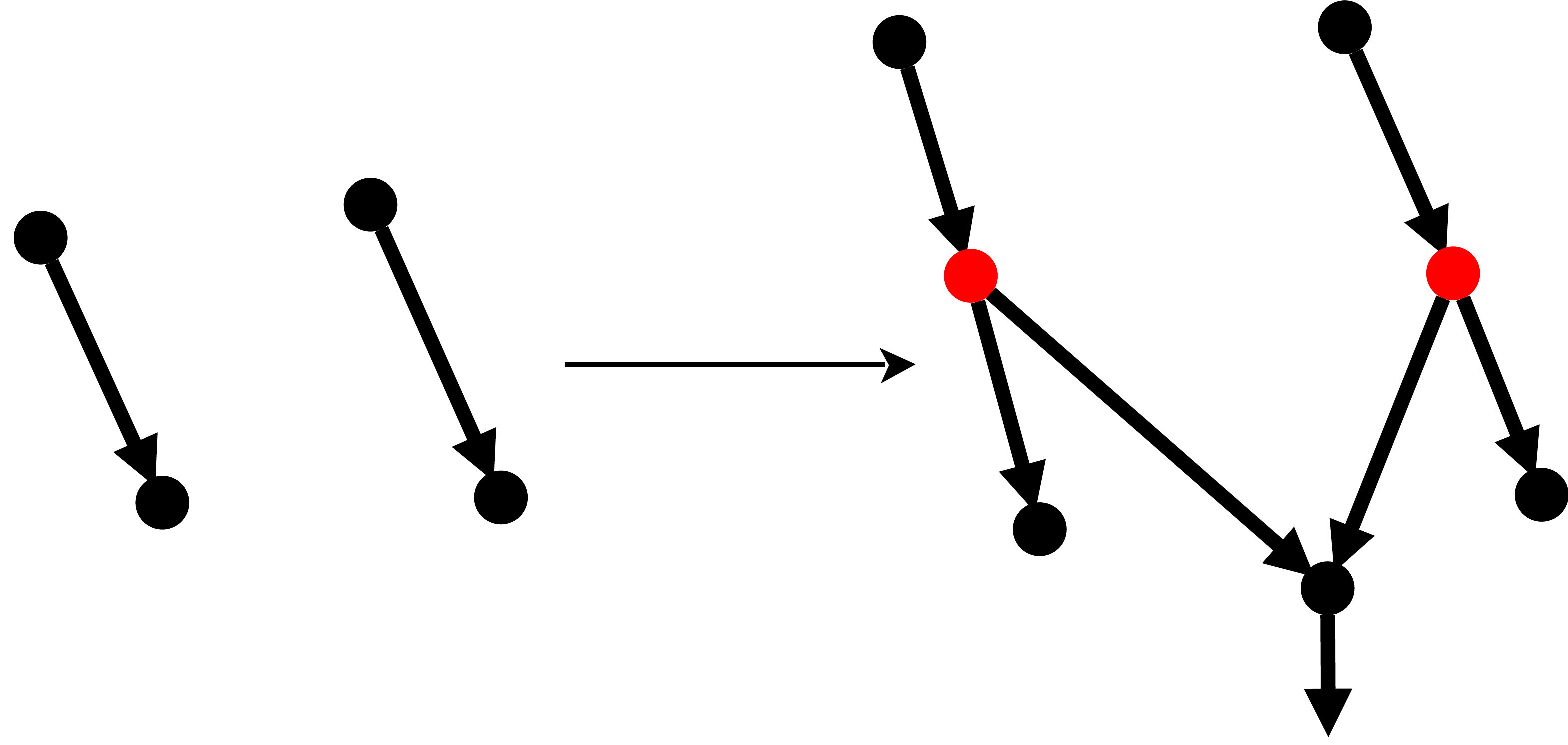} &
& 
\includegraphics[scale=0.16]{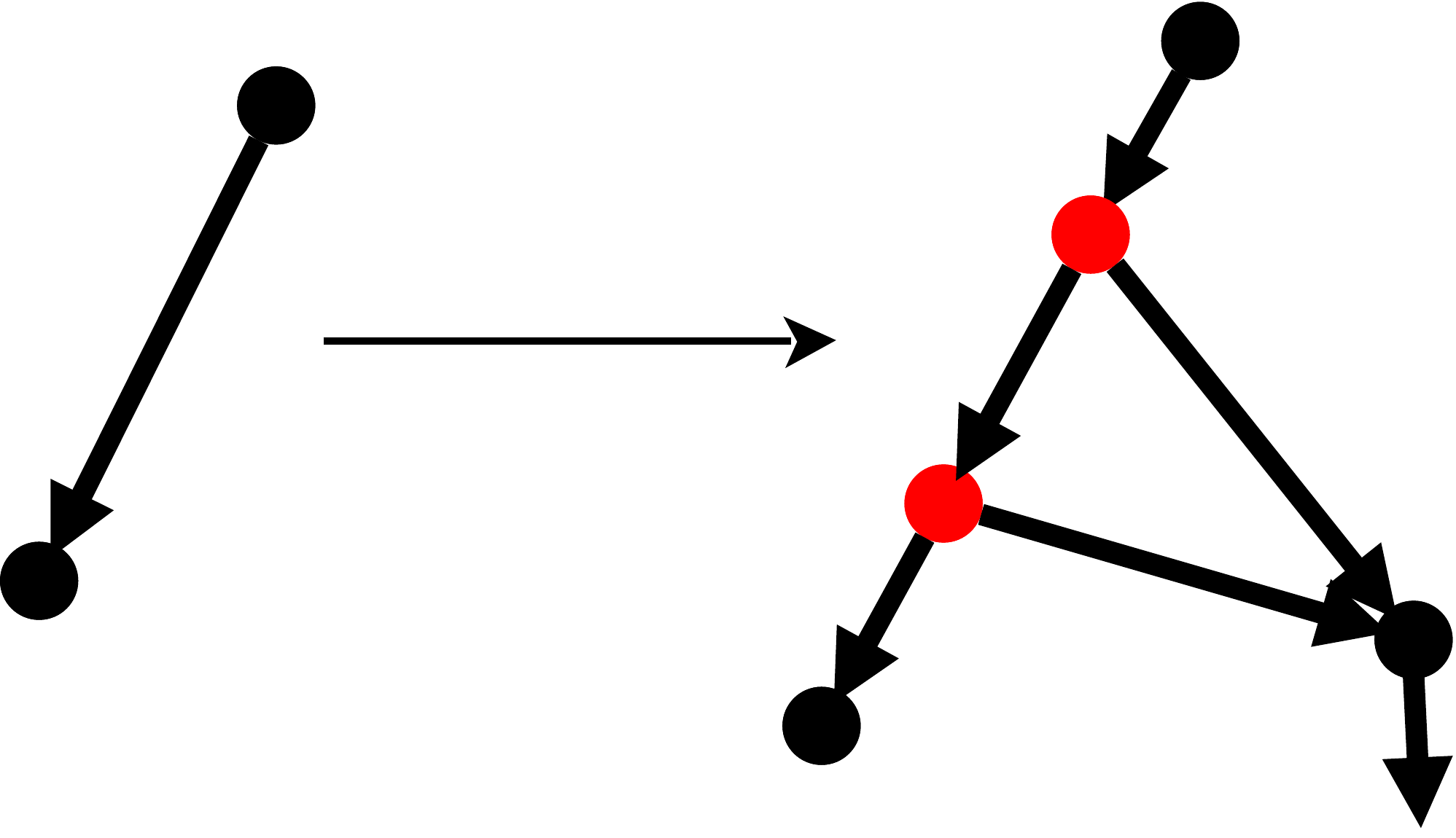} &
& 
\includegraphics[scale=0.16]{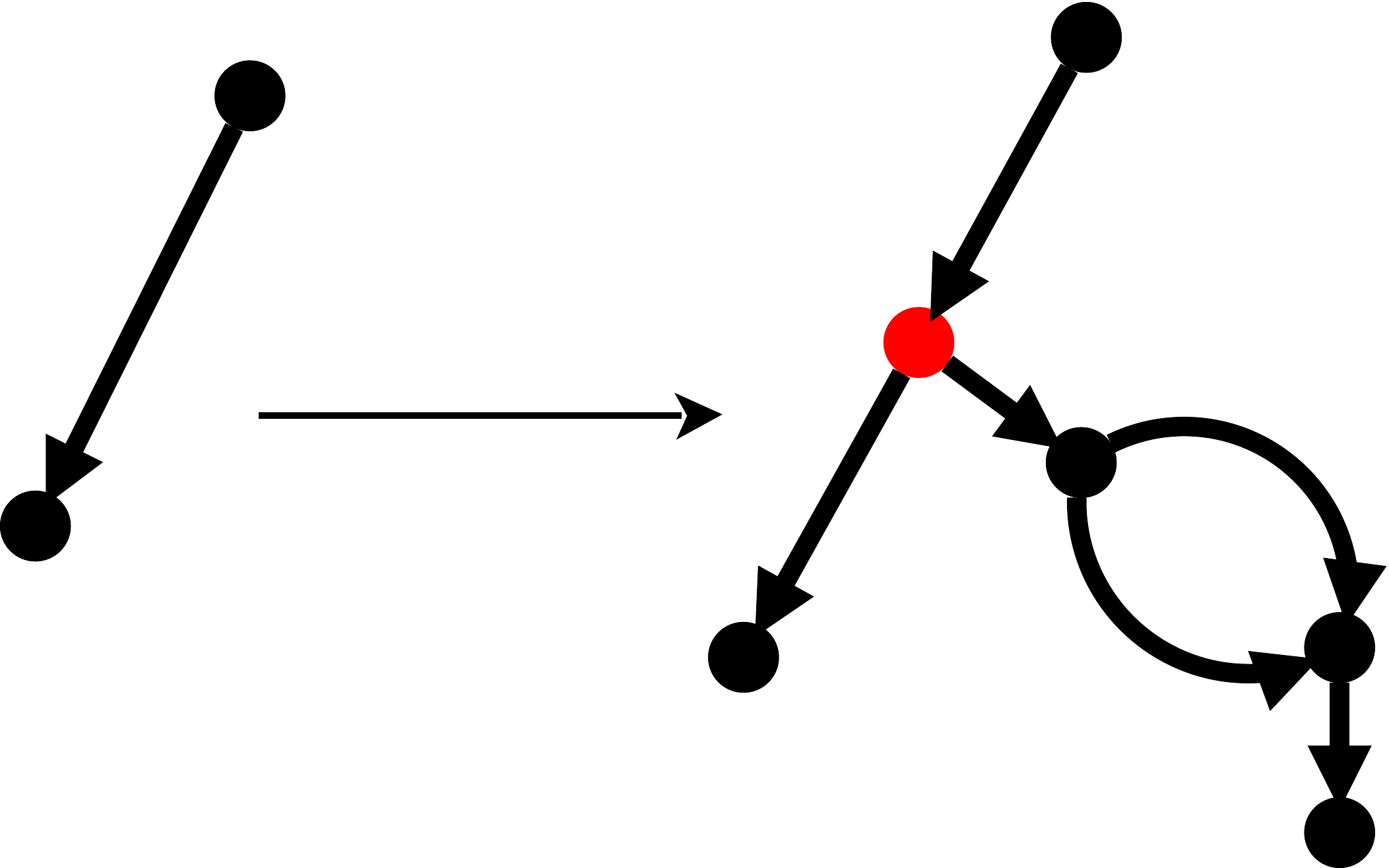} \\
(a) &   & (b) &   & (c)\\
\end{tabular}
\end{tabular}
\hfil
\caption{
The different ways of adding a reticulation back into a network, as discussed in the proof of Theorem \ref{thm:allbinary}. (a) Two different edges are subdivided; (b) one edge is subdivided twice; (c) one edge is subdivided once, under which a multi-edge is placed.
}\label{fig:hangtypes}
\end{figure*}

\begin{theorem}
\label{thm:allbinary}
Let $\cT = \{T_1, \ldots, T_{m}\}$ be a set of binary trees on $\cX$. Then for a constant $r \geq 0$ it is 
possible to construct in polynomial time \emph{all} binary networks with reticulation number at most $r$ that represent 
$Cl(\cT)$ (if any exist).
\end{theorem}
\begin{proof}
Without loss of generality assume we wish to construct all such networks with reticulation number \emph{exactly} $r$. 
Let us suppose that at least one such network $N$ exists.  
By Lemma \ref{lem:treeseq} there is an ST-set tree sequence 
$\cS= (S_1, \ldots, S_{r})$ 
for $Cl(\cT)$ and this corresponds to a sequence of SBRs that, when removed, will prune $N$ into a tree. Let $|X|=n$. Now, note 
that by Observation \ref{obs:staysbinary} and Corollary \ref{cor:stedges} there are at most $O( n^{r} )$ ST-set tree 
sequences, which is polynomial in $n$ for fixed constant $r$. $\cS$ will be one of these, so we can find
$\cS$ in polynomial time. It remains to show how, assuming we have found $\cS$, we can reconstruct $N$.
First, note  that for a binary tree $T$ on $\cX$, $T$ is the \emph{unique} tree on $\cX$ that represents $Cl(T)$. The
clusters that still remain after removing $S_r$ can be represented by a tree, and in particular (by Observation \ref{obs:staysbinary})
by a unique binary tree. We call this tree $N_r$. We want to obtain $N_{r-1}$ by inserting (a tree corresponding to) $S_{r}$ into $N_{r}$. In particular,
we wish to introduce a new reticulation node into $N_r$ 
below which a tree $T$ (itself
binary and unique by Corollary \ref{cor:stedges}) that represents
$S_r$ will be attached.
We have two possibilities to do this: we subdivide two (not necessarily distinct) edges and use these as the tails of the new reticulation edges, see Figures \ref{fig:hangtypes}(a) and (b), or we subdivide one edge \emph{once} and use two \emph{identical} reticulation edges, see Figure \ref{fig:hangtypes}(c). Note that, if we only subdivide one edge once, we actually create a multi-edge and by our definition of phylogenetic network such edges are not allowed. However it might be necessary to create such multi-edges during \emph{intermediate} iterations to ensure that all phylogenetic networks, including ``redundant'' ones, are constructed.

Unfortunately we do not know in general which edge(s) of $N_{r}$ to
subdivide to create the new connection point(s). However, there are only polynomially many edges in $N_r$, so we simply try them all. Then we repeat the process, inserting $S_{r-1}$ into $N_{r-1}$ to obtain 
$N_{r-2}$, and so on until we have obtained $N_0$. Given that $r$ is a constant we can test in polynomial time whether $N_0$ represents
all the clusters in $\cC$ (see proof of Lemma \ref{lem:core}). 

Just as in the similar algorithm described in \cite{simplicityAlgorithmica} there are several slight technicalities that
should be noted.
Whenever some $S_i = \emptyset$ we use a ``dummy'' taxon (i.e. some taxon not in $\cX$) as the tree that we attach below a reticulation. The function of this is to ensure that subsequent
iterations can subdivide the edge leaving the reticulation i.e. it is a placeholder. (This will be necessary when the removal of a single SBR caused the disappearance of two or more reticulations). These 
dummy taxa can 
be removed just before $N_0$ is inspected to
check whether it represents $\cC$. Any $N_0$ that at this point still contains a dummy taxon whose parent is a reticulation, should be rejected, because it means at least one reticulation was not used. Secondly, when we construct the tree $N_r$ we 
actually add a ``dummy root'' which is simply a new node $\rho'$ and a single edge from $(\rho', \rho)$ where $\rho$ is the root of $N_r$. This deals with the situation when the removal
of some SBR caused the current root to disappear and a new root to take its place. At the end, $\rho'$ and the edge leaving it should be removed. Finally,
note that any multi-edges created by intermediate iterations of the algorithm should all have disappeared (i.e. have been subdivided by reticulation edges) by the time $N_0$ has been reached; for
this reason we reject any $N_0$ that still contains multi-edges. 

The network $N$ we would like to reconstruct will eventually be found as some $N_0$, and given that we made no assumptions about $N$ this shows that
the algorithm constructs all possible $N$.

\end{proof}

\begin{corollary}
\label{thm:allsimple}
Let $\cT = \{T_1, \ldots, T_{m}\}$ be a set of binary trees on $\cX$. Then for a constant $k \geq 0$ it is possible
to construct in polynomial time \emph{all} binary simple level-$\leq k$  networks that represent $Cl(\cT)$ (if any exist).
\end{corollary}
\begin{proof}
For each network produced by the algorithm described in Theorem \ref{thm:allbinary} we can easily check in polynomial time whether it is biconnected. 

\end{proof}

It is worth noting at this stage an important link with the rooted triplet literature. Recall the following proposition and 
lemma from \cite{twotrees}, an article in which the relationship between trees, clusters and triplets was discussed more broadly. Proposition
\ref{prop:clustrip} refers to not necessarily binary trees.

\begin{proposition}
\label{prop:clustrip}
(Van Iersel, Kelk \cite{twotrees}) For any set~$\cT$ of trees on the same set~$\cX$ of taxa, any phylogenetic 
network on~$\cX$ representing~$Cl(\cT)$ is consistent with~$Tr(\cT)$.
\end{proposition}

\begin{lemma}
\label{lem:binary2}
(Van Iersel, Kelk \cite{twotrees}) Let $N$ be a phylogenetic network on $\mathcal{X}$ and $\mathcal{T}$ a set of binary trees on 
$\mathcal{X}$. Then there exists a
binary phylogenetic network $N'$ on $\mathcal{X}$ such that (a) $N'$ has the same reticulation number and level as $N$, (b) if $N$
displays all trees in $\mathcal{T}$ then so too does $N'$, (c) if $N$ is consistent with $Tr(\mathcal{T})$ then so too is $N'$ and
(d) if $N$ represents $Cl(\mathcal{T})$ then so too does $N'$.
\end{lemma}

Suppose  that we have an Algorithm $A$ which, for each fixed $r \geq 0$, can construct in polynomial time every binary 
network consistent with $Tr(\cT)$ that has at most $r$ reticulations, where $\cT$ is a set of binary trees on $\cX$. Suppose $A'$
is an algorithm that examines in turn every network output by $A$ and rejects it if it does not represent $Cl(\cT)$ (such a 
filtering step can be done for each network in polynomial time for fixed $r$, see proof of Lemma \ref{lem:core}). 
If there exists a network $N$ on $\cX$ with at most 
$r$ reticulations that represents $Cl(\cT)$ then by Lemma \ref{lem:binary2} there exists a binary network $N'$ on $\cX$ with this 
property. Furthermore, in that case 
$N'$ will by Proposition \ref{prop:clustrip} be consistent with $Tr(\cT)$. The algorithm $A'$ is thus guaranteed to eventually find $N'$. The practical 
consequence
of this is that, simply by adding a filtering step, triplet software can in some cases easily be modified 
to work for clusters. Indeed, they can be used to determine 
whether a  network with $r$ reticulations  that represents $\cC$ exists and if so to construct all binary networks with this property.
As proof of
concept we have taken the triplet software \textsc{SIMPLISTIC} (based on the ideas described in \cite{simplicityAlgorithmica}), removed its biconnectedness-checking subroutine so that it generates all binary networks with 
up to $r$ reticulations (and not just all simple level-$\leq r$ binary networks), and added the cluster filtering step as
described above. This whole process took only one day of programming, and lead to the new software package \textsc{CLUSTISTIC}
which implements the result described in Theorem \ref{thm:allbinary}. This software  is available for download at  
\url{http://skelk.sdf-eu.org/clustistic}.



\section{Witnesses and a natural lower bound}
\label{sec:wit}

\noindent
Now, let $\cT = \{T_1, \ldots, T_{m}\}$ be a set of $m$ not necessarily binary trees on $\cX$.
Let $S$ be a non-trivial ST-set of $Cl(\cT)$. We know by Lemma \ref{lem:nonbinaryunseptree} that for each $T_i \in \cT$ there is an edge $e_i$ or node $v_i$ in $T_i$ that represents $S$. We define a \emph{witness} for $S$ in $T_i$ as follows. If $T_i$ contains an edge $e_i = (u_i,v_i)$ 
that represents $S$, then a witness is any leaf descendant $x_i \in \cX$ of $u_i$ that does not appear in $S$,  i.e. $x_i \in (\cX(u_i) \setminus  S)$. 
Otherwise, from Lemma \ref{lem:nonbinaryunseptree}
there exists a  
node $v_i$ in $T_i$ that represents $S$, and in that case a witness is any leaf descendant 
$x_i \in (\cX(v_i) \setminus  S)$.  
The only ST-sets with no witnesses are $\cX$ and the empty set, hence the restriction to non-trivial ST-sets.
As an example consider the four trees $\cT$ in Figure \ref{fig:4t}. Consider the 
ST-set $\{1\}$ of $Cl(\cT)$. In the top-left tree and bottom-left tree the only possible witness for this is taxon 5. In the top-right tree the only witness is taxon 3,
and in the bottom-right tree the only witness is taxon 2. 

Given a set of trees $\cT$ on $\cX$ and a non-trivial ST-set $S$ of $ \cX$,
 let $W \subseteq \cX$ be any subset 
of  taxa such that, for each tree $T_i \in \cT$, there exists $x \in W$ that is a witness for $S$ in $T_i$. 
We call such a set 
a \emph{witness set} of $S$ in $\cT$. Clearly there exist $W$ such that $|W| \leq m$. 
For example, for the set of trees in 
Figure \ref{fig:4t}, $\{2,3,5\}$ is a possible witness set for $\{1\}$. 

The two following simple observations are critical. 

\begin{observation}
\label{obs:wit}
Let $\cT$ be a set of trees  on $\cX$, $S$ a non-trivial ST-set  of $ \cX$ and 
$W$ a witness set  of $S$ in ${\cT}$. Then for each $C \in \cC(\cT)$ such that $S \subset C$, $W \cap C \neq \emptyset$.
\end{observation}
\begin{proof}
For each cluster $C \in \cC(\cT)$ there is at least one edge $e = (u,v)$ in some $T_i \in \cT$ such that $e$ represents $C$. Let $e_i$ $(v_i)$ be the
edge (node) in $T_i$ that represents $S$. Given that $S \subset C$, all leaf descendants of $v$ must also include all leaf descendants
of the tail of $e_i$ $(v_i)$. In particular, all possible witnesses for $S$ in $T_i$. 
\end{proof}

To understand the meaning of Observation \ref{obs:wit} it is helpful to again consider the example of ST-set $\{1\}$ in the context of Figure \ref{fig:4t}. We see that any cluster that is a strict superset of $\{1\}$ must contain at least one of the taxa from $\{2,3,5\}$.

\begin{observation}
\label{obs:wit2}
Let $\cT = \{T_1, \ldots, T_{m}\}$ be a set of trees on $\cX$ such that  $m \geq 2$ and 
let $S$ be a non-trivial ST-set of $\cC$. Let $W$ be a smallest-cardinality witness set of $S$ in $\cT$. If $|W| = m$ then 
for each $C \in \cC(\cT)$ such that $S \cap C = \emptyset$, $W \setminus C \neq \emptyset$.
\end{observation}
\begin{proof}
Note that it is not possible for a witness for $S$ in a tree $T_i \in \cT$ to also be a witness for $S$
in $T_j \neq T_i$, because then $|W| \leq m-1$. So each witness in $W$ comes from a different tree in $\cT$.
Suppose then that there is some $C\in \cC(\cT)$ such that $S \cap C = \emptyset$ and $W \setminus C = \emptyset$. Clearly, since $W\neq \emptyset$, $W \subseteq C$,
and $|C| \geq |W| \geq 2$. Furthermore some edge $e_i$ in some $T_i$ represents $C$. Combining the fact that $S \cap C = \emptyset$
and $W \subseteq C$ leads us to the conclusion that all elements of $W$ are possible witnesses for $S$ in $T_i$. But then some element $x$
of $W$ is a witness for $S$ both in $T_i$ and also in some $T_j \neq T_i$, contradiction.
\end{proof}

As mentioned in the proof of Observation \ref{obs:wit2}, a witness set $W$ (for a given ST-set $S$) with $m-1$ or fewer elements exists if and only if the possible witnesses for $S$ ranging
across the different $T_i$ are not all mutually disjoint.
If a witness set with $m-1$ or fewer elements does not exist then in the following results any witness set with $m$ elements will turn out to be sufficient, as a consequence of Observation \ref{obs:wit2}. This will become clear in due course. 

\subsection{A natural lower bound on $r(\cT)$ that is tight for clusters obtained from two (not necessarily binary) trees\label{sec:MSTlowerbound}}




\begin{lemma}
\label{lem:lowerbound2}
Given a set of clusters $\cC$ on $\cX$, there exists a maximal ST-set tree sequence $(S_1, S_2, ..., S_p)$ such
that ${p} \leq r(\cC)$.
\end{lemma}
%
%
\begin{proof}
The proof is equivalent to that of Lemma \ref{lem:treeseq} but for the fact that  no empty ST-set is inserted in the  maximal ST-set tree sequence.
\end{proof}

We define the \emph{maximal ST-set lower bound for $\cC$} (MST lower bound for short) as the  cardinality of the smallest 
 maximal ST-set tree sequence
for $\cC$. By Lemma \ref{lem:lowerbound2} this is a genuine lower bound on $r(\cC)$. In general it is however a rather weak lower-bound: consider the set 
$\cC^{i}$ of clusters
on $X^{i} = \{r, x_1, \ldots x_i\}$ defined by $\{ \{r,x_j\} | 1 \leq j \leq i \}$. The MST lower bound for this cluster set
is always 1, while $r(\cC^{i})$ rises linearly in $i$. However, 
as we shall see the tightness of the bound is to some extent correlated with the
number of trees which generate the clusters, with two trees being a special case. We first require some definitions and auxiliary lemmas.

\noindent

Consider a network $N$ on $\cX$. Let $\mathcal{X'} = \{\rho \} \cup \cX$ where $\rho \not \in \cX$ is some arbitrary symbol representing the 
root of $N$. Let $T$ be some tree on $\cX^{*}$ where $\cX^{*} \cap \cX' = \emptyset$ and let $H$ be a subset of $\cX'$. 
We can obtain a new network $N'$ on $(\cX^{*} \cup \cX)$
\emph{by hanging $T$ from $H$ in $N$}. Informally $N'$ is obtained by 
hanging the tree $T$ beneath a new reticulation which has $|H|$ incoming edges, where each such incoming edge begins ``just above'' an 
element of $H$. Formally the transformation is as follows. 
First we add a new edge $(r,r')$ to $T$, where $r'$ is the root of $T$ and $r$ is a new node.
For each $h \in H \setminus \{\rho\}$ we then subdivide the unique edge entering $h$; let $h_p$ be the new parent (with indegree and outdegree 
1) of $h$. 
For each $h \in H \setminus \{\rho\}$ we then add a new edge $(h_p,r)$. Finally, if $\rho \in H$ we also add an edge from the root of $N$ to $r$; we call this a 
\emph{root edge}. It is clear that $r(N') = r(N) + |H|-1$.

\begin{lemma}
\label{lem:hangback}
Let $\cT = \{T_1, \ldots, T_{m}\}$ be a set of $m$  trees on $\cX$ and let $\cC = Cl(\cT)$. Let $S$ be a maximum 
ST-set of $\cC$. Let $N$ be any network on $\cX \setminus S$ that represents $\cC \setminus S$. Then it is possible to extend $N$ to obtain a new network $N'$ 
that 
represents $\cC$ such that $r(N') \leq r(N) + (m-1)$.
\end{lemma}
\begin{proof}
Let $T_{S}$ be the unique tree on taxa set $S$ such that  $\cC(T_S)= \cC| S$.
Let $W$ be a minimum-cardinality witness set for $S$ in $\cT$.
Clearly $1 \leq |W| \leq m$. If $|W|=m$ then we let $N'$ be the network obtained by hanging $T_{S}$ from $W$ in $N$. If $|W| < m$ then
we let $N'$ be the network obtained by hanging $T_{S}$ from $W \cup \{ \rho \}$ in $N$. Clearly, $r(N') \leq r(N) + (m-1)$. It remains only to
show that $N'$ represents $\cC$. Consider any cluster $C \in \cC$. There are three cases to consider. (1) If $C \subseteq S$ then $N'$
clearly represents $C$ because $T_S$ already represented $\cC|S$. (2) If $S \subset C$ then consider $C' = C \setminus S$. Clearly $N$
represents $C'$. By Observation \ref{obs:wit} there exists 
some $w \in W \cap C$. Since $W \cap S= \emptyset$, this implies that there exists some $w \in W \cap C'$.

 To see that $N'$ represents $C$ consider
any tree displayed by $N$ that represents $C'$. We can extend this tree by ``switching on'' the new reticulation edge that begins
above $w$, i.e. the edge $(w_p,r)$, and ``switching off'' the remaining reticulation edges. (3) If $S \cap C = \emptyset$ then there
are two subcases. (3.1) If $|W| < m$ then we can ``switch on'' the root edge that enters the reticulation above $T_S$, i.e. the edge $(\rho,r)$, and ``switch off'' all 
other reticulation edges entering $T_S$. (3.2) If $|W|=m$ then by Observation \ref{obs:wit2} there exists $w \in W \setminus C$.
In $N'$ we can thus ``switch on'' the new reticulation edge that begins above $w$, and ``switch off'' the rest. 
\end{proof}

\begin{theorem}
\label{thm:tightbound}
Let $\cT = \{T_1, \ldots, T_{m}\}$ be a set of $m$ trees on $\cX$ and let $\cC = Cl(\cT)$. Let $p$ be the
MST lower bound for $\cC$. Then $r(\cC) \leq (m-1)p$. 
\end{theorem}
\begin{proof}
Given a tree $T$ and a node $u$ of $T$, we denote by $\cX(T)$ the label set of $T$ and by $T_u$  the subtree rooted at $u$.
From Lemma \ref{lem:lowerbound2} we already know that $p \leq r(\cC)$. Now, let $(S_1, S_2, ..., S_p)$ be a maximal ST-set tree sequence for
$\cC$. We will complete the proof by showing how to explicitly construct a network $N$ with reticulation number at most $(m-1)p$ that 
represents $\cC$. We define $\cC_{i}$, $1 \leq i \leq p$, as $\cC \setminus S_1 \setminus \ldots \setminus S_{i}$ and $\cC_{0}$ as $\cC$.
By Lemma
\ref{lem:nonbinaryunseptree} it can be seen that for each $i$, $\cC_{i} = Cl( \cT_{i} )$ where $\cT_{i}$ is a set of at most $m$
trees on $\cX \setminus S_1 \setminus \ldots \setminus S_{i}$ and where $\cT_{0} = \cT$. In particular, $\cT_{i+1}$ can be obtained from $\cT_{i}$ as follows. Given a tree $T_j$ in $\cT_{i}$, let (without loss of generality) $u_j$ be the node of $T_j$ such that $S_i$ is equal to the union of the clusters
represented by some not necessarily strict subset of its outgoing edges. Such a $u_j$ exists by Lemma \ref{lem:nonbinaryunseptree}. Let $Q = \{v_1, \ldots, v_k\}$ be the set of children of $u_j$ such that for each $v \in Q$, 
$\cX(T_{v})$ contains at least one element of $S_{i}$. The set of trees $\cT_{i+1}$ can be obtained from $\cT_{i}$
by computing, for each tree $T_j$ in $\cT_{i}$ the tree $T_j \setminus {T_{v_1}} \ldots \setminus {T_{v_k}}$ i.e. pruning away the subtrees corresponding to $S_i$ and tidying up the resulting tree.

Now, consider $\cC_{p}$. Let $N_{p}$ be the unique tree 
such that $\cC(N_p)= \cC_{p}$; $N_p$ will
be equal to the single tree in $\cT_{p}$. By Lemma \ref{lem:hangback} we can obtain a network $N_{p-1}$ with $(m-1)$ reticulations that represents 
$\cC_{p-1}$ by taking $\cT = \cT_{p-1}$, $N = N_{p}$ and $S = S_{p}$ in the proof of that lemma. We iterate this process for $p-2, p-3, \ldots 1$. This lasts at 
most $p$ iterations in total, and each iteration adds $(m-1)$ to the reticulation number, thus yielding a network $N_{0}$ that represents $\cC$ with 
reticulation number (at most) $(m-1)p$. 
\end{proof}

\begin{corollary}
\label{cor:tightbound}
Let $\cT = \{T_1, T_{2}\}$ be a set of two not necessarily binary trees on $\cX$ and let $\cC = Cl(\cT)$. Let $p$ be the
MST lower bound for $\cC$. Then $r(\cC) = p$. 
\end{corollary}

In \cite{twotrees} it is shown that it is NP-hard and APX-hard to compute $r(\cC)$ where $\cC$ is the set of clusters obtained from two binary
trees on $\cX$. The following corollary is thus immediate.

\begin{corollary}
\label{cor:boundishard}
The computation of the MST lower bound is NP-hard and APX-hard.
\end{corollary}


It is interesting to note that Lemma \ref{lem:lowerbound2} and Corollary \ref{cor:boundishard} have, in some sense, already appeared
in the phylogenetic network literature, albeit in the language of \emph{recombination networks}. Specifically, in \cite{twotrees} we highlight
that the phylogenetic network model described there (and also used here) is in a strong sense identical to the recombination network model under
the assumption of an all-0 root, the infinite sites model and multiple crossover recombination. The computational lower bound described in Algorithm 3 of \cite{myers2003}
is, taking this equivalence into account, essentially identical to the MST lower bound.  In
\cite{bafnabansal2006} it is shown that computing this bound is NP-hard, by reduction from \textsc{MAX-2-SAT}. The same authors also give an exponential-time dynamic programming algorithm 
for computing the bound because in \cite{myers2003} it was not explicitly indicated how this should be computed. 





\section{The optimality and non-optimality of \textsc{Cass}}
\label{sec:casseverything}

The {\cass} algorithm for constructing simple level-$k$ networks was presented in \cite{cass}. The algorithm
was designed to produce solutions of minimum level, not of minimum reticulation number. 
However, when the input is a 
separating set $\cC$ of clusters on $\cX$, minimizing the  level or the reticulation number is equivalent. Indeed, such cluster
sets have the property that any network that represents them is simple or can easily be made simple (see Lemma \ref{lem:simpleexists}) and a simple network contains exactly one non-trivial biconnected component.

The {\cass} algorithm can be used as a subroutine in a divide and conquer algorithm to construct general level-$k$ networks. We will call this more
general algorithm {\cassinden}. The basic idea of {\cassinden} is that it transforms each connected component of $IG(\cC)$ into a tangled set of clusters,
runs {\cass} separately on each of these tangled sets, and combines the resulting simple networks into a single final network $N$ (Recall that tangled cluster sets are separating).  The final network $N$ has reticulation number equal to the sum of the reticulation number of the
simple networks produced by {\cass}, and $N$  has level equal to the maximum level ranging over all the simple networks. For more details on the divide and conquer strategy, see \cite{cass} and \cite{HusonRuppScornavacca10}, Section 8.2.


In \cite{cass} the authors proved that if there is  a 
level-$\leq 2$ network that represents 
$\cC$, then {\cassinden} will find such a solution with minimal level. 
Here we clarify several other properties of the algorithm. 
On the negative side we show (using a special 
separating 
 cluster set) that
the {\cass} algorithm does not in general minimize level. 
On
the positive side we show that when the input set $\cC$ is equal to $Cl(\{T_1,T_2\})$ for any two (not necessarily 
binary) trees, {\cassinden} correctly minimizes level. In fact we show something even stronger: in this case {\cassinden} also correctly
constructs networks with minimum reticulation number, which in turn is exactly equal to the hybridization number of the two input trees i.e. the number
of reticulations required to display the trees themselves. We conclude with several open questions regarding {\cass}.

\subsection{Cass: the high-level idea}
\label{subsec:casshigh}


Let $N$ be a network that represents a set of clusters $\cC$ on $\cX$. Let $S$ be a non-trivial ST-set
with respect to $\cC$. We say that $S$ is \emph{under a cut-edge} if $N$ contains a cut-edge
$(u,v)$ such that the subnetwork rooted at $v$ is a tree that represents $\cC|S$.

The next two results formed (implicitly) the direct inspiration for \textsc{Cass}, which was designed to be a generalization of these results.

\begin{lemma}
\label{lem:moveST}
Let $N$ be a network that represents a set of clusters $\cC$ on $\cX$. Let $S$ be a non-trivial ST-set
with respect to $\cC$. Then there exists a network $N'$ such that $r(N') \leq r(N)$, $\ell(N') \leq \ell(N)$,
$S$ is under a cut-edge in $N'$ and for each ST-set $S'$ such that $S' \cap S = \emptyset$ and $S'$ is
under a cut-edge in $N$, $S'$ is also under a cut-edge in $N'$.
\end{lemma}
\begin{proof}
Deferred to the appendix.
\end{proof}

The following corollary follows from the fact that maximal ST-sets are disjoint:
\begin{corollary}
\label{cor:maxundercut}
Let $N$ be a network that represents a set of clusters $\cC$. There exists a network $N'$
such that $r(N') \leq r(N)$, $\ell(N') \leq \ell(N)$ and all maximal ST-sets (with respect to $\cC$)
are below cut-edges.
\end{corollary}




The pseudocode for {\cass} was originally given in \cite{cass}. That exposition is however rather dense and technical. See Section 8.5 of \cite{HusonRuppScornavacca10} for a clearer detailed description.
Here we only give the core idea of the algorithm. Let us assume 
without loss of generality that we want to know, given a 
separating set of clusters, whether a simple network solution exists with reticulation number \emph{exactly} $k$, for some constant $k$.

{\cass} tries to answer this by searching through the space of all
maximal ST-set tree sequences of length at most $k$, attempting to build a network with reticulation number $k$ from each one.
It looks first at shorter maximal ST-set tree sequences, padding those of length less
than $k$ with empty ST-sets to attain a sequence of length exactly $k$. (As in the proof of Theorem \ref{thm:allbinary}, this models the situation
when removing a single SBR causes the reticulation number to drop by more than 1).
If there are no maximal ST-set tree sequences of length at most $k$ then {\cass} will correctly report 
 that no solutions with reticulation number $k$ or lower exist. Hence {\cass} implicitly computes and incorporates the MST lower bound.

Assuming maximal ST-set tree sequences of length at most $k$ \emph{do} exist, {\cass} examines each one to determine whether it can be constructively turned into a real solution. Let $\cS= (S_1, \ldots, S_k)$  be a
(possibly padded) maximal ST-set tree sequence of length-$k$ for $\cC$.
As in earlier sections we define $\cC_{i}$, $1 \leq i \leq k$, as $\cC \setminus S_1 \setminus \ldots \setminus 
S_{i}$, and we let $\cC_{0} = \cC$. {\cass} does not however work with the set $\cC_{i}$. Instead it works with $\cC'_{i}$ which is obtained from each $\cC_{i}$ by 
``collapsing'' \emph{every} maximal ST-set $S$ with respect to $\cC_{i}$ into a single new ``meta-taxon''. This is an \emph{extremely} greedy step, and is in some
sense an attempt to generalize Corollary \ref{cor:maxundercut}. The informal motivation is this: we know from Corollary \ref{cor:maxundercut} that
there exists some network $N$ with a minimum number of reticulations such that all maximal ST-sets are under cut-edges. Suppose we guess an SBR $T$ on $\cX' \subset \cX$ (where $\cX'$ is a maximal ST-set) of $N$; 
we only have to make polynomially-many guesses because there are only polynomially-many maximal ST-sets. This gives a
new network $N'$ on $\cX \setminus \cX'$ where perhaps not all maximal ST-sets (with respect to $\cC \setminus \cX'$) are under cut-edges. In particular, some SBRs of $N'$ might correspond to non-maximal ST-sets, of which
there are potentially exponentially many, so how do we efficiently guess an SBR of $N'$? Fortunately, we can transform $N'$ (in the sense of Corollary \ref{cor:maxundercut}) to
obtain a new network $N''$ such that $r(N'') \leq r(N')$ and where all maximal ST-sets (with respect to $\cC \setminus \cX'$) \emph{are} under cut-edges of $N''$. Hence we know that we again only have to make polynomially-many
guesses to locate an SBR of $N''$. Furthermore, \textsc{Cass} assumes that these maximal ST-sets will always remain below cut-edges, so it collapses them into the aforementioned meta-taxa. We iterate this entire process $k$ times,
concluding the ``inward'' phase of \textsc{Cass}.

This assumption is important because it affects the ``outward'' phase of {\cass}, which begins immediately after completion of the inward phase. As in for example Theorem \ref{thm:allbinary} the general idea is 
to start with a tree that represents $\cC_{k}$ and then to work backwards, first trying all pairs of edges from which to ``hang back'' a SBR corresponding to $S_k$, then all pairs of edges (of the resulting network) from which to hang back an SBR corresponding to $S_{k-1}$, and so on, 
down to $S_1$. However, before hanging back each $S_i$ it first \emph{decollapses} the maximal ST-sets that were collapsed into meta-taxa during the corresponding iteration of the inward phase.

In Section  \ref{app:proofCassBreaks} of the appendix we will
explicitly and exhaustively walk through a specific execution of the \textsc{Cass} algorithm and this is helpful for clarifying exactly how the algorithm works. 



\subsection{The {\cass} algorithm is not always optimal}
\label{subsec:cassbad}
In this section we present a counter-example proving that the {\cass} algorithm does not always minimize level. 
\begin{figure}[h]
  \centering
  \includegraphics[scale=.14]{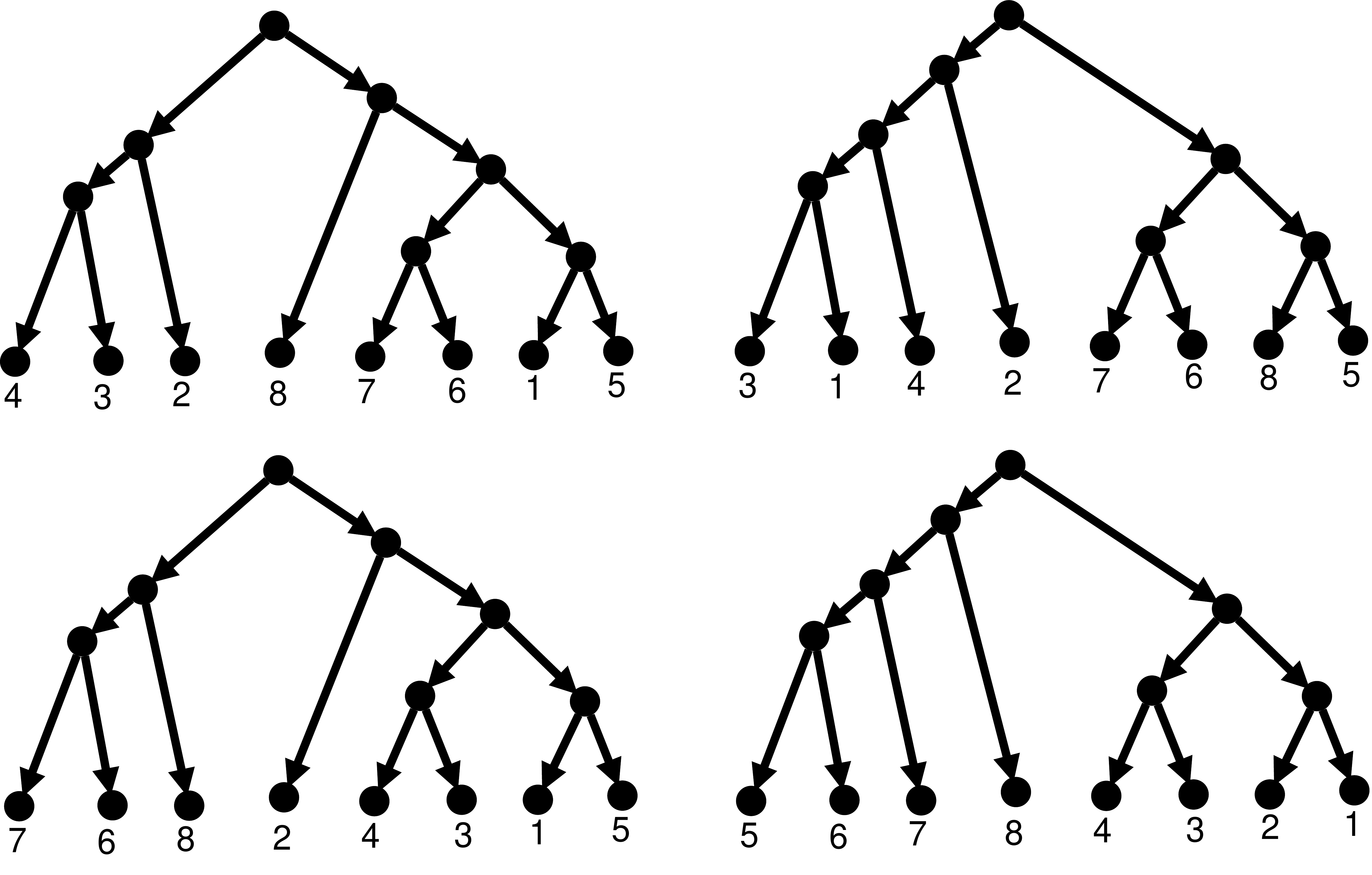}
  \caption{Let $\cT$ be the set of four trees shown here. The {\cass} algorithm returns a network $N$ that represents $Cl(\cT)$ where $r(N)=\ell(N)=4$. However, Figure \ref{fig:4tL3L4} shows that the true value of $r(Cl(\cT)) = \ell(Cl(\cT))$ is at most 3.}
  \label{fig:4t}
\end{figure}

\begin{figure*}
 \centering
\begin{tabular}{c}
\begin{tabular}{cc}
\includegraphics[scale=0.15]{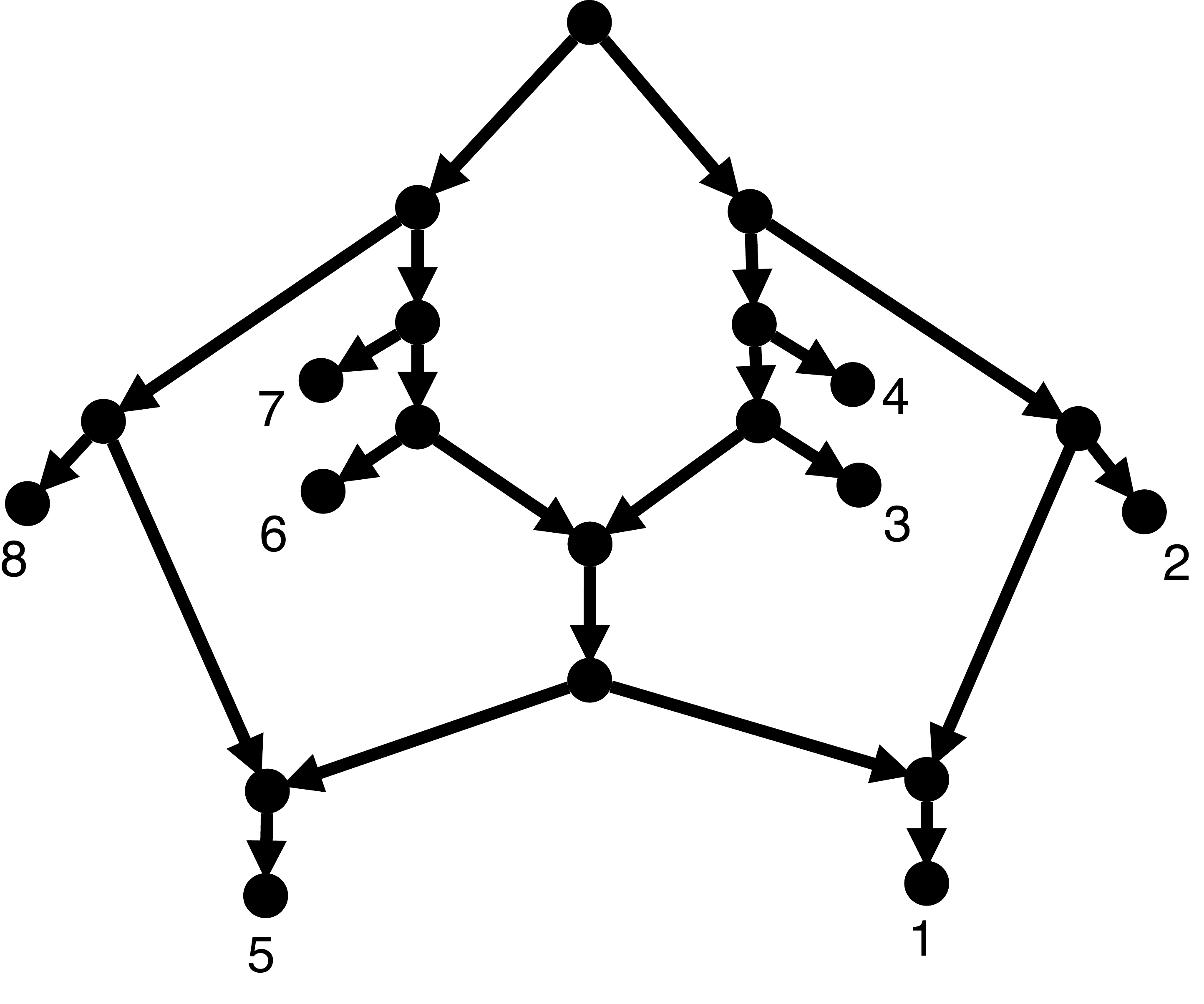} \hspace{1cm}
 &
\includegraphics[scale=0.23]{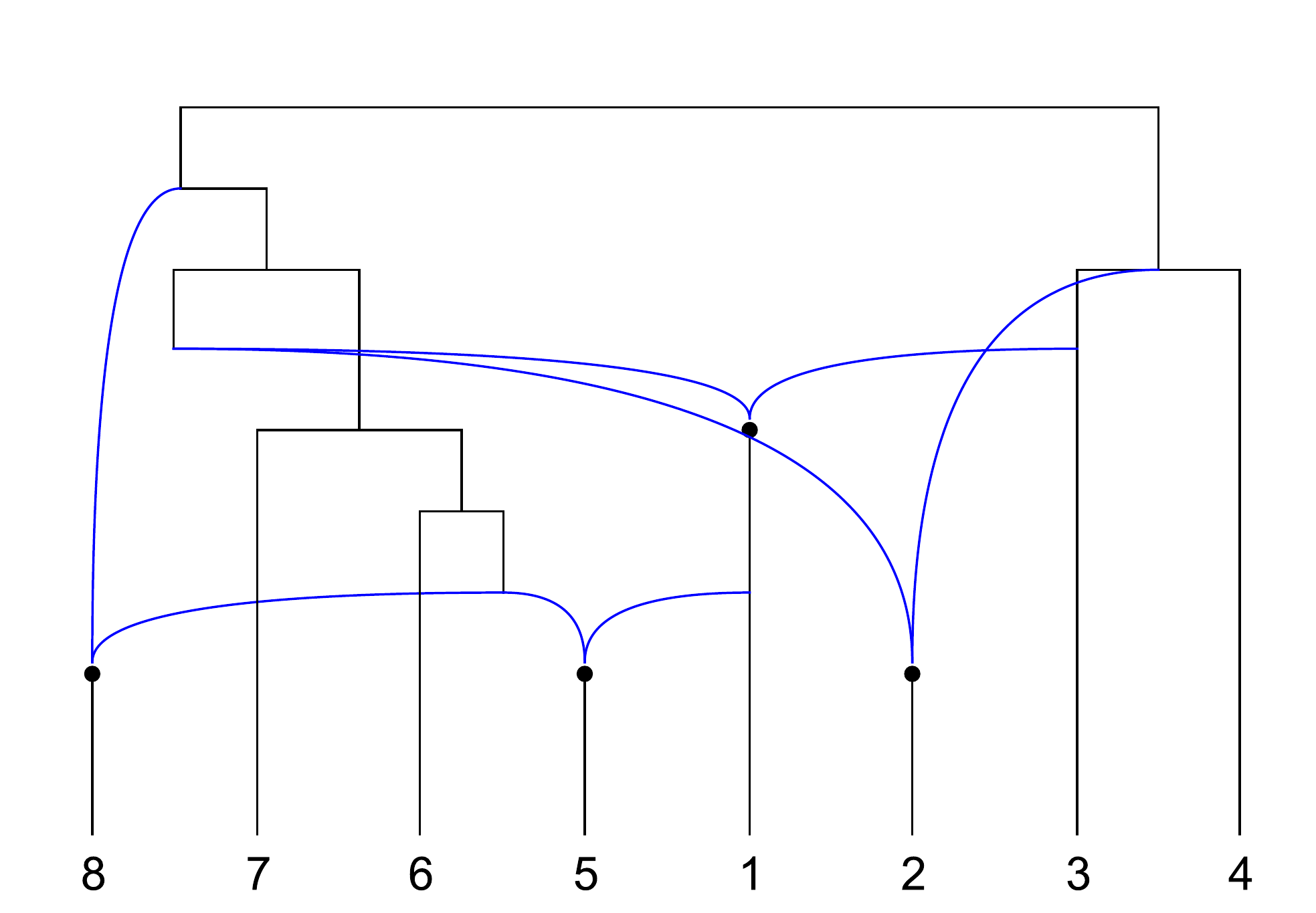}  
\\
(a) & (b) \\
\end{tabular}
\end{tabular}
\hfil
\caption{
(a) A simple level-3 network and (b) a simple level-4 network, both  representing $Cl(\cT)$ where $\cT$ is defined as described in Figure \ref{fig:4t}. The level-4 network was produced
by {\cassinden}.
}\label{fig:4tL3L4}
\end{figure*}

%

Consider the set $\cT_{4}$ of four binary trees shown in Figure \ref{fig:4t}. It is easy to verify that the set $\cC = 
Cl(\cT)$ is 
separating 
and every network that represents $\cC$ is thus simple or can easily be made simple. It is also easy to verify that the simple 
level-3 network in Figure \ref{fig:4tL3L4}(a) represents $\cC$; in the appendix we prove that this is optimal by showing that
any network that represents $\cC$ must have reticulation number 3 or higher. 

However, \cass~cannot find a level-3 network. To formally prove this 
we show in Section  \ref{app:proofCassBreaks}  of the appendix an exhaustive list of all possible 
executions of the algorithm with $k=3$. 
\textsc{Cass} returns  the simple level-4 network shown in Figure \ref{fig:4tL3L4}(b). To summarize, the problem with {\cass} seems to be 
that 
while 
the step of always collapsing at \emph{every} iteration \emph{all} maximal ST-sets and  treating them as meta-taxa (in the sense of
Corollary \ref{cor:maxundercut})
is a locally optimal move, it can force us 
to use too many reticulation edges when hanging (trees corresponding to) maximal ST-sets back in the outward phase. 
%

Note that this counter-example fits in a tradition of highly specific and complex counter-examples
in the phylogenetic network literature. In particular we note the very similar counter examples given initially in \cite{gusfielddecomp2007}, and (based on this) in 
\cite{husonetalgalled2009},
which showed that one cannot minimize reticulation number by optimizing independently over the connected components of the
incompatibility graph $IG(\cC)$. The relationship between these counter-examples - all of which are linked in some way or other to simple level-3 networks - seems to be that networks with 
minimum reticulation 
number have a very subtle internal structure that seems impervious to locally optimal and greedy strategies, but that these properties only start emerging for
level-3 and higher.


It is important to emphasize, however, that since we could not find \emph{any} non-synthetic dataset for which {\cass} does not find an optimal solution, we still have the feeling that {\cass} works quite well for real data.

\subsection{{\cass} is optimal for sets of clusters obtained from two trees}

Here we show that, despite the negative news in the previous section, {\cass} (and more generally {\cassinden}) correctly minimizes level in the case of clusters obtained from two not necessarily binary trees. Furthermore the algorithms also provably minimize reticulation number.

Note that the following theorem does not contradict the NP-hardness mentioned in Corollary \ref{cor:boundishard} because {\cass} only runs in polynomial time when it bounds its search to simple level-$\leq k$ networks, 
for a \emph{constant} $k$. 

\begin{theorem}
\label{thm:casscorrect}
Let $\cC = Cl(\cT)$ be a 
separating 
set of clusters where $\cT$ is a set of two not necessarily binary trees on $\cX$.  Then {\cass} constructs a simple network $N$ that represents $\cC$
such that $\ell(N) = r(N) = \ell(\cC)= r(\cC)$.
\end{theorem}
\begin{proof}
Let the MST lower bound for $\cC$ be $p$. Recall that, by Corollary \ref{cor:tightbound}, in this case $p=r(N)$. We know that there is a 
maximal ST-set tree sequence $ (S_1, \ldots, S_p)$. As explained in Section \ref{subsec:casshigh} {\cass}
will eventually find this maximal ST-set tree sequence. Now, Theorem \ref{thm:tightbound} essentially
works by invoking Lemma \ref{lem:hangback} $p$ times, and in the statement of Lemma \ref{lem:hangback} there
are absolutely no assumptions made about the structure of ``$N$'', other than that it represents a certain set of clusters. Hence
``$N$'' can just as well be one of the intermediate networks constructed by \textsc{Cass}. It remains only to show that {\cass}
can simulate the hanging-back construction described in the proof of Lemma \ref{lem:hangback}. This is definitely so, because
the proof of Lemma \ref{lem:hangback} requires the edges entering two witnesses (or the edge entering one witness and, to simulate the attachment
of a root edge, the edge connecting a dummy root to the real root) to be subdivided. {\cass} tries subdividing
all pairs of edges, including the edge between the dummy root and the real root, and hence will eventually subdivide the correct two edges.
\end{proof}

\noindent
 %
%

The proof of Theorem \ref{thm:casscorrect} not only shows that {\cass} is optimal for sets of clusters obtained from two not necessarily binary trees, but also that in this very special case the ``hanging back'' (i.e. outward) phase of {\cass} is in some sense
completely redundant. In particular: if we have already computed a maximal ST-set tree sequence, then we can easily compute a witness set for each of the maximal ST-sets in the sequence,
and these witness sets directly specify a sufficient set of edges to subdivide when hanging back the maximal ST-sets. So in the case of clusters coming from two trees {\cass} wastes rather a lot of time trying to hang back maximal ST-sets
from all possible pairs of edges, when in fact the information is already available to make this blind search unnecessary.\\
\\
Theorem \ref{thm:casscorrect} can actually be re-formulated and extended to general (i.e. not necessarily separating) sets of clusters $\cC$ obtained from two not necessarily binary trees. Indeed,  in Theorem \ref{th:cassworksnonbinary}, we will prove that, for such cluster sets, \cassinden
reconstructs a network $N$ such that 
%
$r(N) = r(\cC)$ and $\ell(N) = \ell(\cC)$.  However, we first need to prove Theorem \ref{thm:clusdecom} below, which is interesting in its own right, and  several auxiliary results.

\begin{observation}
\label{obs:sepstayssep}
Let $\cC$ be a set of clusters on $\cX$ and let $U \subseteq \cX$ be an 
unseparated set with respect to $\cC$. Then for
any $P \subseteq \cX$, $U \setminus P$ is unseparated with respect to $\cC 
| (\cX \setminus P)$.
\end{observation}
\begin{proof}
$U$ is unseparated with respect to $\cC$ so for each cluster $C \in \cC$ 
we have either that
$C \cap U = \emptyset$, $C \subseteq U$ or $U \subseteq C$. For the case 
$C \cap U = \emptyset$ it is clear
that $(C \setminus P) \cap (U \setminus P) = \emptyset$. For the case $C 
\subseteq U$ we have that
$(C \setminus P) \subseteq (U \setminus P)$, and for the case $U \subseteq 
C$ we have that $(U \setminus P) \subseteq
(C \setminus P)$. 
\end{proof}

\begin{observation}
\label{obs:ccunsep}
Let $\cC$ be a set of clusters on $\cX$ and let $U$ be the union of the 
set of clusters in a connected
component $K$ of $IG(\cC)$. Then $U$ is unseparated.
\end{observation}
\begin{proof}
Suppose $U$ is \emph{not} unseparated. Then there must exist some $C \in 
\cC$ such that $C \not \subseteq U$,
$U \not \subseteq C$ and $C \cap U \neq \emptyset$. Clearly $C$ cannot be 
incompatible with any cluster in the
connected component $K$, because
then $C$ would also be in the connected component $K$ and thus $C 
\subseteq U$. Hence every cluster in the connected
component $K$ is either disjoint from $C$, or a subset of it, and there is 
at least one of each type of cluster because
$U \setminus C \neq \emptyset$ and $U$ is the union of all the clusters in 
$K$. However, clusters
in $K$ that are disjoint from $C$, are always compatible with clusters in 
the connected
component that are contained inside $C$, so the connected component is not 
connected, contradiction. 
\end{proof}

\begin{observation}
\label{obs:subsetST}
Given a set of clusters $\cC$ on $\cX$, let $S$ be an ST-set with respect to $\cC$ and let $U$ be an unseparated set such that $U \subseteq S$. Then $U$ is also an ST-set  for $\cC$. 
\end{observation}
\begin{proof}
We only need to show that all pairs of clusters in $\cC|U$ are compatible. Clearly, for each $C \in \cC$ we have that $C|U \subseteq U$. Now, recall that, because $U$ is unseparated, for each $C \in \cC$ we have either $C \subseteq U$, $U \subseteq C$ or $C \cap U = \emptyset$.  Suppose by contradiction that for some $C_1 \neq C_2 \in \cC|U$, $C_1$ and $C_2$ are incompatible. But then $\emptyset \subset C_1, C_2 \subset U$. But in that case $C_1, C_2 \in \cC$ and $C_1, C_2 \subset S$, contradicting the fact that $S$ was an ST-set.
\end{proof}

For a set of clusters $\cC$ and an unseparated set $U$ with 
respect to $\cC$, $\cC_{U \rightarrow u}$
denotes the new cluster set obtained by replacing all elements of $U$ with a 
single new taxon $u$ (i.e. ``collapsing''
$U$ into a single taxon). 

\begin{figure*}
  \centering
  \includegraphics[scale=0.15]{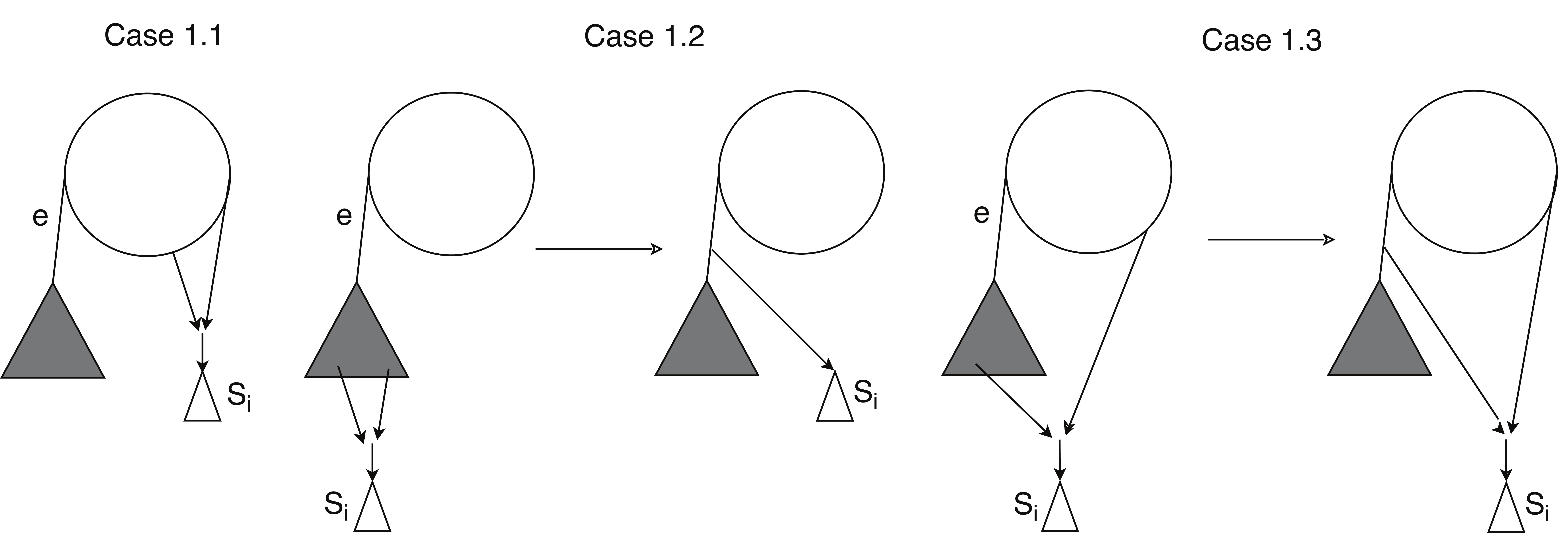}
  \caption{The three subcases covered by Case 1, which concerns the case 
when the ST-set $S_i$ that is 
being hung back, is disjoint from $U$. The parts of the network containng elements of $U$ are depicted in grey.}
  \label{fig:case1}
\end{figure*}

\begin{theorem}
\label{thm:clusdecom}
Let $\cT = \{T_1, T_2\}$ be two not necessarily binary trees on $\cX$, and 
let $\cC = Cl(\cT)$. Let $U \subseteq \cX$
be an unseparated set with respect to $\cC$. Then $r(\cC)= r(\cC|U) + 
r(\cC_{U \rightarrow u})$.
\end{theorem}
\begin{proof} Let $p = r(\cC)$. We know by Corollary \ref{cor:tightbound} that there 
exists a
maximal ST-set tree sequence $(S_1, \ldots, S_p)$. As usual we define $\cC_i$, $1 
\leq i \leq p$, as
$\cC \setminus S_1 \setminus ... \setminus S_i$, and we let $\cC_0 = \cC$. 
We let $\cX_i = \cup_{C \in \cC_i} C$
where $\cX_{0} = \cX$. (Note that $\cX_i=\cX \setminus S_1 \setminus ... \setminus S_i$). Since $U \setminus S_1 \setminus ... \setminus S_i=  \cX_i \cap U$,  by repeated 
application of Observation \ref{obs:sepstayssep},
$\cX_i \cap U$ is unseparated with respect to $\cC_i$ for $0 \leq i \leq 
p$. Now, recall (see proof of Lemma \ref{lem:hangback}) that for each $S_i$ we can identify
a set of two \emph{witnesses} (where perhaps one of the witnesses 
is the symbol $\rho$ representing the root). 

Here we show, for each $S_i$, how to hang back a tree representing 
$\cC_{i-1}|S_i$ from a network $N$ 
representing $\cC_i$ to obtain a network $N'$
representing $\cC_{i-1}$ where $r(N')=r(N)+1$ \emph{and} in $N'$ the taxa 
$\cX_{i-1} \cap U$ are \emph{exactly} the set of
taxa below some cut-edge. The witnesses of $S_i$ guide us
how to do this. If we repeat this $p$ times we will obtain a 
network with $p$ reticulations (and reticulation number $p$) that represents
$\cC$ and such that the taxa in $U$ are exactly the subset of taxa below 
some cut-edge. The theorem will then follow.
%

The first thing to do is to study the earliest point at which some 
elements of 
$U$ are added back into the network.
This is 
an important ``base case''.
Let us thus consider the largest value of $i$ such that $\cX_{i} \cap U 
\neq \emptyset$. Let $i' $ be equal to this
value. Now, suppose $i'  = p$.
We saw that by repeated 
application of Observation \ref{obs:sepstayssep} $\cX_p \cap U$ is 
unseparated with respect to $\cC_p$. Furthermore
we know that the clusters $\cC_{p}$ can be represented by a tree, so 
$\cX_p \cap U$ is actually an ST-set. Hence
we can assume without loss of generality (by Lemma \ref{lem:moveST})
that the tree that represents 
$\cC_{p}$, has a cut-edge such that
$\cX_{p} \cap U$ is exactly the set of taxa beneath it. Alternatively,
suppose $i'  < p$. In this case the first elements of $U$ that are
reintroduced into the network are a (not necessarily strict) subset
of $S_{i' +1}$, so we have $\cX_{i' } \cap U$ = $S_{i' +1} \cap U$. Given that 
$\cX_{i' } \cap U$ is unseparated w.r.t. $C_{i' }$, and $S_{i' +1} \cap U$
is a subset of an ST-set, it follows by Observation \ref{obs:subsetST} that $S_{i' +1} \cap U$ is also
an ST-set. We may thus assume without loss of generality that
$\cX_{i' } \cap U$ is exactly the set of taxa below a cut-edge (again
thanks to Lemma \ref{lem:moveST}).

Henceforth we may assume that the network 
$N$ that we want to hang (a tree corresponding to) $\cC_{i-1}|S_i$ back from, contains at least one 
taxon of $U$. We will make heavy use of this fact. Let $e$ be the cut-edge 
of $N$ which the elements of $\cX_{i} \cap U$ are below. Let $w_1, w_2$ be 
the two witnesses 
for $S_i$, where in some cases $w_2 = \rho$
(representing the root).\\

There are several cases to consider. The first case is when the tree that 
we are hanging back, is disjoint from $U$. See also Figure 
\ref{fig:case1}.\\
\\
\textbf{Case 1) $S_i \cap U = \emptyset$}\\
\\
\textbf{Subcase 1.1) $\{w_1, w_2\} \cap U = \emptyset$}. In this case we 
can simply hang back from $\{w_1, w_2\}$ because we are not adding any new 
elements of $U$ and we are not subdividing any edge reachable from $e$. 
Hence $e$ remains the cut-edge which all present elements of $U$ are 
below.\\
\\
\textbf{Subcase 1.2) $\{w_1, w_2\} \subseteq U$}. If we 
simply hang $S_i$ back from $\{w_1, w_2\}$ then we obtain a network that 
represents $\cC_{i-1}$. However, we see from this that every cluster $C$ 
in $\cC_{i-1}$ that is a strict superset of $S_i$, must contain at least 
one element of $U$. $S_i$ is disjoint from $U$ so by unseparation $C$ must 
also contain all other elements of $U$ in the network. Hence we 
do not actually need to put $S_i$ below a reticulation
at all: we can simply attach it to a single new cut-edge that subdivides 
$e$. (This case cannot actually occur because it saves a reticulation and 
hence implies that $r(\cC) < p$.)\\
\\
\textbf{Subcase 1.3) (wlog) $w_1 \in U$, $w_2 \not \in U$}. Suppose 
we simply hang back from $\{w_1, w_2\}$, obtaining a network $N'$ that 
represents $\cC_{i-1}$. Note that
after doing this any cluster that passes through the reticulation edge 
starting just above $w_1$, must (because of the unseparation of $U$)
contain all elements of $U$ that are in the network. So we can subdivide 
the cut-edge $e$ and move the tail of of that reticulation edge to the 
newly created vertex.\\
\begin{figure*}
  \centering
  \includegraphics[scale=0.15]{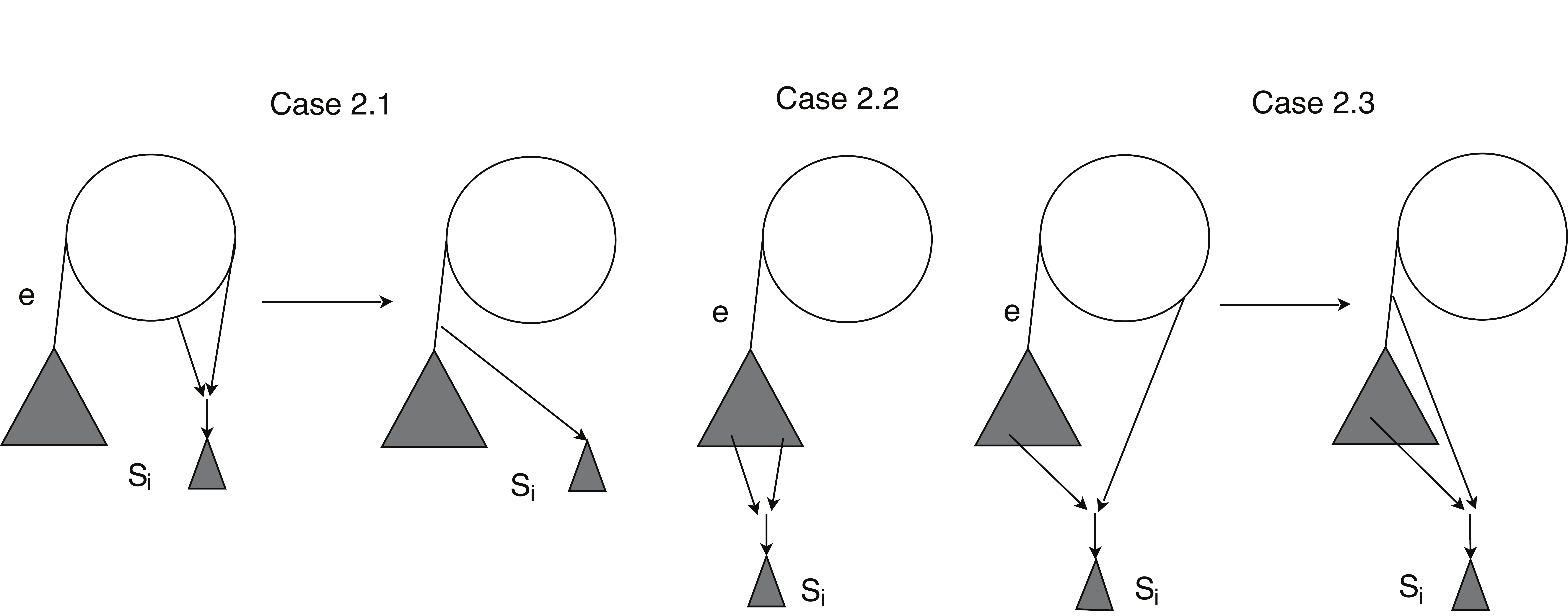}
  \caption{The three subcases covered by Case 2, which concerns the case 
when the ST-set $S_i$ that is 
being hung back, is a subset of $U$. The parts of the network containing elements of $U$ are depicted in grey.}
  \label{fig:case2}
\end{figure*}

The next case is when all elements of $S_i$ are in $U$, see also Figure \ref{fig:case2}.\\
\\
\textbf{Case 2) $S_i \subseteq U$}\\
\\
\textbf{Subcase 2.1) $\{w_1, w_2\} \cap U = \emptyset$}. In this case we 
don't need 
a reticulation at all: we can just hang $S_i$ from a single new cut-edge 
that subdivides $e$. (This case cannot actually happen because it saves a 
reticulation: see case 1.2).\\
\\
\textbf{Subcase 2.2) $\{w_1, w_2\} \subseteq U$}. This case is fine 
because if we hang back from $w_1$ and $w_2$ all
the elements of $U$ remain below the cut-edge $e$.\\
\\
\textbf{Subcase 2.3) (wlog) $w_1 \in U$, $w_2 \not \in U$}. Suppose 
we hang back from $w_1$ and $w_2$ to obtain a network that represents 
$\cC_{i-1}$. In this case we 
could move the reticulation edge that starts just above $w_2$ (or at the 
root, in the case that $w_2 = \rho$) to subdivide the cut-edge $e$.\\
\begin{figure}
  \centering
  \includegraphics[scale=0.18]{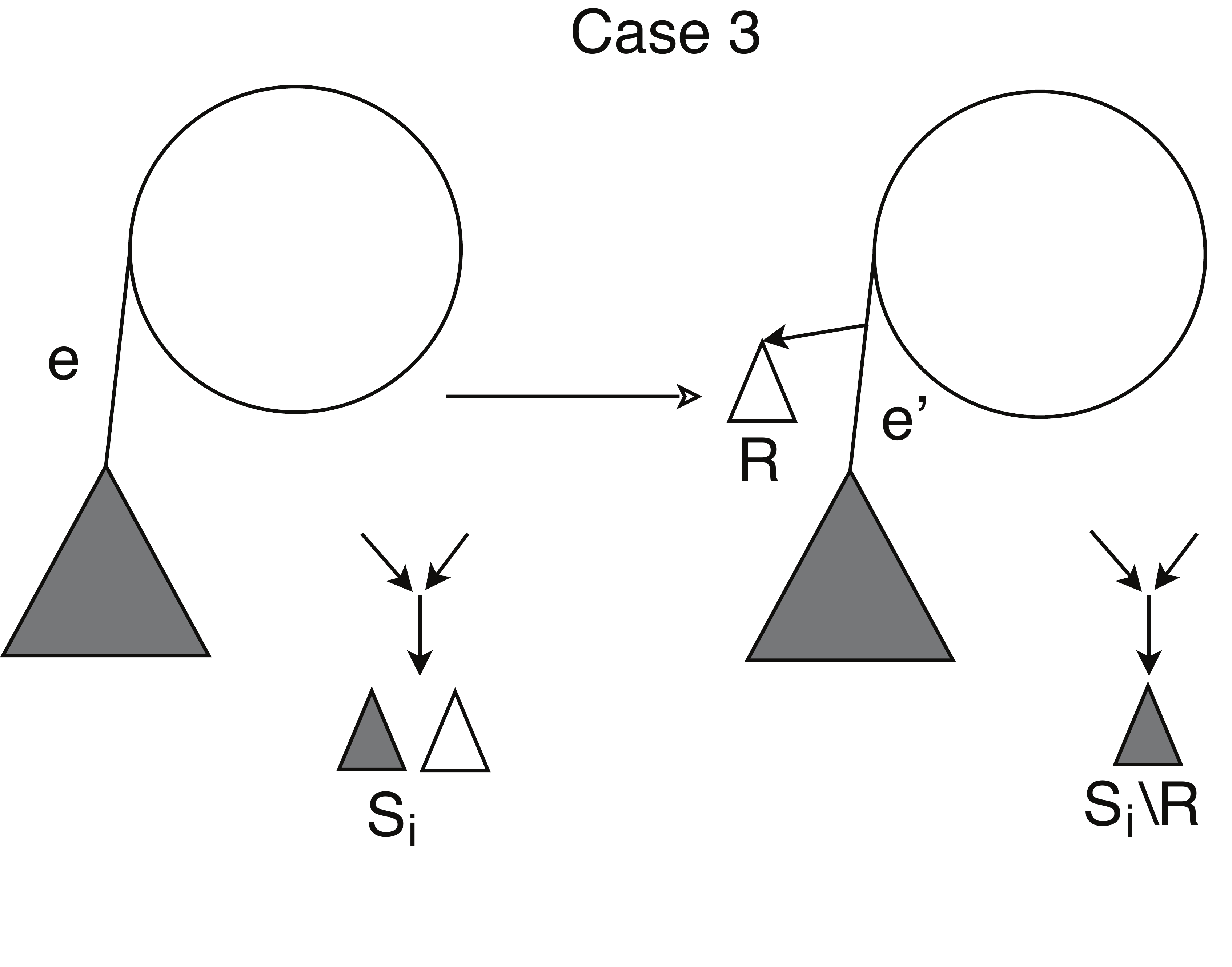}
  \caption{Case 3, which concerns the case 
when the ST-set $S_i$ that is 
being hung back, contains elements of $U$ \emph{and} elements not in $U$.  The parts of the network containing elements of $U$ are depicted in grey. Note that in this case we do not care where the reticulation edges connect to the rest of the network. Here $R = S_i \setminus U$.}
  \label{fig:case3}
\end{figure}

The final case is where $S$ contains at least one element of $U$ and at 
least one element not in $U$. See also Figure \ref{fig:case3}.\\
\\
\textbf{Case 3) $S_i \cap U \neq \emptyset$ and $S_i \setminus U \neq 
\emptyset$}\\
\\
In this case we will apply a transformation which brings us back into Case 
2. Let $R = S_i \setminus U$. Suppose we hang
$S_i$ back from its two witnesses $w_1$ and $w_2$, to obtain a network 
$N'$ that represents $\cC_{i-1}$. Consider any
cluster $C \in \cC_{i-1}$ such that $C \cap R \neq \emptyset$. Then either 
$C \subseteq R$ or
$\cX_{i-1} \cap U \subseteq C$. (The second condition holds because, if $C 
\not \subseteq R$, then it must contain some element in $S_i \cap U$, and 
hence all elements of $U$ in the network). So a cluster $C$
that is a \emph{strict} subset of $S_i$ is either a subset of $R$ or 
disjoint from $R$. Indeed, if $C$ contains one element of $R$ and one of the unseparated set $U$, $C$ must contain all elements of $U$, a contradiction since $C$ is a {strict} subset of $S_i$.
It follows that $R$ is unseparated so by Observation \ref{obs:subsetST} 
$R$ is an ST-set. 
Hence we may assume (without loss of generality) that $R$ is the taxa set 
of some subtree $T'$ below a cut-edge in the tree representing 
$\cC_{i-1}|S_i$ that we hung back. Now, we can prune $T'$ and regraft it 
back onto the network at a new vertex obtained by subdividing $e$. It can 
be verified that after this prune/regraft move the resulting network still 
represents $\cC_{i-1}$. It remains only to apply Case 2, taking $w_1, w_2$ 
as the witnesses and $S_i \setminus R$ as the ST-set that we want to hang 
back. 
\end{proof}

\begin{theorem}
\label{th:cassworksnonbinary}
Let $\cT=\{T_1,T_2\}$ be two not necessarily binary trees on  $\cX$, and let $\cC = Cl(\cT)$.
When given $\cC$ as input, {\cassinden} computes a level-$\ell(\cC)$ network 
$N$ with reticulation number $r(\cC)$ that represents $\cC$.
\end{theorem}

\begin{proof}
Observation \ref{obs:ccunsep} and  
Theorem \ref{thm:clusdecom} ensure that we can analyze each connected component  of the incompatibility graph $IG(\cC)$ separately, which (as mentioned) is exactly what {\cassinden} does.  (To see this it is helpful to note that any subset of $\cC$ can also, with the
possible exception of some superfluous singleton clusters, be expressed as the
set of clusters in two trees).

Let $K$ be a connected component and  denote by $\cC_K$ the set of clusters in $K$. Let $\cX_K$ be the set of taxa equal to the union of all clusters in $\cC_K$.  
Note that $\cC_K$ is not necessarily a tangled set. Indeed, while $IG(\cC_K)$ is connected, the second property of  a tangled set (every pair of taxa $x$ and $y$ in $\cX_K$ is separated by $\cC_K$, Section \ref{sec:prelim}) does not always hold. To ensure that the latter condition holds \cassinden~simply computes all the maximal ST-sets $\{S_1, ..., S_k\}$ for $\cC_K$ and for each  of them replaces all elements of $S_i$ with a single new taxon $s_i$ in $\cC_K$, to
obtain a new cluster set $\cC'_K$ that is tangled.

To see that the constructed network has minimum level, observe that (by Theorem \ref{thm:casscorrect}) {\cass} correctly
computes minimum level solutions for the $\cC'_{K}$ cluster sets mentioned above. In \cite{cass} it is proven that combining minimum-level
solutions for the various $\cC'_{K}$ yields a minimum-level solution for $\cC$. 

We now need to prove that the constructed network has reticulation number $r(\cC)$. Since maximal ST-sets are unseparated and all mutually disjoint (see Corollary \ref{cor:mostn}), it follows from Theorem \ref{thm:clusdecom}  that:\\
\[r(\cC_K) = \displaystyle\sum_{
\begin{subarray}{c}
 \mbox{S is a maximal} \\[1mm]
 \mbox{ST-set of $\cC_K$}  
\end{subarray}
}
r(\cC_K|S) + 
r(\cC'_K) .\]
Since $r(\cC_K|S)$  always equals zero when $S$  is an ST-set, $r(\cC_K) = r(\cC'_K)$. Moreover, since the sets $\cX_K$ are also unseparated (from Observation \ref{obs:ccunsep}) and all mutually disjoint, then from Theorem \ref{thm:clusdecom}  we have that:\\
\[r(\cC) = \displaystyle\sum_{
\begin{subarray}{c}
 \mbox{K is a connected} \\[1mm]
 \mbox{component of $IG(\cC)$}  
\end{subarray}
}
r(\cC|\cX_K) + 
r(\cC') ,\]
where $\cC'$ is obtained from $\cC$ by 
replacing all elements of $\cX_K$ with a single new taxon $x_K$. Obviously, $r(\cC')=0$. 
Since $r(\cC|\cX_K)=r(\cC_K) = r(\cC'_K)$, and $r(\cC'_K)= \ell(\cC'_K)$ because $\cC'_K$ is separating, 
 this concludes the proof that the constructed network has reticulation number
$r(\cC)$. 

\end{proof}

\subsection{{\cass} can be used to compute the hybridization number of two not necessarily binary trees\label{sec:cassHybrid}}


It is important to understand the relationship between the results presented in the previous section and the extensive literature on computing the hybridization distance of two \emph{trees} 
\cite{bordewich,bordewich2,sempbordfpt2007,quantifyingreticulation,linzsemple2009} i.e. the problem of displaying the trees themselves, and not just their clusters.

 For two binary trees $\cT = \{T_1, T_2\}$ on $\cX$ the hybridization distance $h(\cT)$ is defined as  the minimum
reticulation number $r_t(\cT)$ of any network that displays both the trees in $\cT$. In \cite{twotrees} we showed that $r_t(\cT) = r(Cl(\cT) )$. As discussed in \cite{twotrees}
that means that both positive and negative results for the computation of $r_t(\cT)$ transfer automatically to computation of $r(Cl(\cT))$ (for two binary trees). Negative results
are NP-hardness and APX-hardness; positive results include fixed parameter tractability
and running time improvements based on an increasingly
deep understanding of maximum acyclic agreement forests. 

Before proceeding there are some technical issues regarding the definition of hybridization number of a set $\cT$ of two \emph{not necessarily binary} trees on $\cX$. This can be attributed to the fact that
several different definitions have appeared in the literature:
\begin{enumerate}
\item $h^{+}(\cT)$ : the minimum reticulation number of $N$ ranging over all networks $N$ that display in a strict topological sense all the trees in $\cT$. This is exactly the
definition of \emph{display} given in Section \ref{sec:prelim}. In Figure 8 of \cite{twotrees} this strict definition was used.
\item $h^{0}(\cT)$: the minimum reticulation number of $N$ ranging over all networks $N$ that display some not necessarily binary refinement of each tree in $\cT$. Recall that a (binary) \emph{refinement} of a tree $T$ on $\cX$ is any (binary) tree $T'$ on $\cX$ such that $Cl(T) \subseteq Cl(T')$. (Note that a tree is generically considered to be a refinement of itself).
\item $h^{-}(\cT)$: the minimum reticulation number of $N$ ranging over all networks $N$ that display a binary refinement of each tree in $\cT$. This was the definition
used in \cite{linzsemple2009}. 
\end{enumerate}

Two of the definitions, $h^{0}$ and $h^{-}$, turn out to be equivalent. We clarify this, extend an equivalence result from \cite{twotrees} and thus show that {\cassinden} correctly computes the hybridization number of two not necessarily binary trees in the sense of $h^{0}$, equivalently $h^{-}$. In Figure 8 of \cite{twotrees} a set $\cT$ of two trees
is given, one of which is non-binary, such that $r(Cl(\cT)) < h^{+}(\cT)$. A result showing that {\cass} computes $h^{+}$ was thus already excluded.

\begin{observation}
\label{obs:treedefequiv}
$h^{0}(\cT) = h^{-}(\cT)$ for all sets $\cT$ of not necessarily binary trees on the same taxa set $\cX$. 
\end{observation}
\begin{proof}
The fact that $h^{0}(\cT) \leq h^{-}(\cT)$ follows immediately from the definitions. Suppose by way of contradiction that there exists a set of
not necessarily binary trees $\cT$ such that $h^{-}(\cT) > h^{0}(\cT)$. Let $N$ be any network with reticulation number $h^{0}(\cT)$ that
displays some refinement of each tree in $\cT$. Now, it is not too difficult to see (using for example the transformation described in Lemma 2
of \cite{twotrees}) that we can create a binary network $N'$ such that $r(N') = r(N)$ and such that $N'$ displays a binary refinement of
each of the trees in $\cT$, yielding a contradiction.
\end{proof}


Given a set of trees $\cT$, recall that we define $r_{tr}(\cT)$ to be the minimum reticulation number of any network that displays $Tr(\cT)$ (i.e. all the rooted triplets in the input trees). We have the following result:

\begin{theorem}
\label{theo:twotrees-1extended}
Let $\cT = \{T_1,T_2\}$ be two not necessarily binary trees on $\cX$. Then 
$h^{0}(\cT) = r_{tr}(\cT)$. 
\end{theorem}
\begin{proof}
Obviously, for any set $\cT$ of not necessarily binary trees on $\cX$, $r_{tr}(\cT) \leq h^{0}(\cT)$.  It remains to show that this inequality is always tight. Suppose that it is not always tight. 
Let then $\cT = \{T_1, T_2\}$ be 
two smallest 
(in terms of the size of $|\cX|=n$) trees
such that $h^{0}(\cT) > r_{tr}(\cT)$. 
Clearly $n>2$. Now, let $N_{tr}$ 
be any network that is consistent with
$Tr(\cT) = R$ such that $r(N_{tr}) = r_{tr}(R)$. If $r_{tr}(R) = 0$ we 
have a contradiction because this only occurs if
$T_1$ and $T_2$ have a common 
refinement, in which case $h^{0}(\cT) = 
0$. Hence $r_{tr}(R)>0$. This means that
$N_{tr}$ has at least one SBR with taxa set $S$. Now, we claim that, for 
each $i \in \{1,2\}$, $S$ is the set
of taxa reachable from an edge $e_i$ (in $T_i$) or as the set of taxa 
reachable from some subset of the children
of some node $u_i$ (in $T_i$). If this is not so then there exists some 
triplet $xy|z$ in $R$ such that
$x \not \in S$ and $y,z \in S$. However, $N_{tr}$  cannot  be 
consistent with $xy|z$ since $S$ is the taxa set of a SBR, which by definition sits below
a cut-edge, yielding a contradiction. 
(If $e_i$ exists assume without loss of generality that $u_i$ is its head).
Let $Q' = \{v_1, \ldots, v_k\}$ be the set of children of $u_i$ such that for each $v \in Q'$,  $\cX(T_{v})$ contains at least one element of $S$. Now let $T'_i$ be the tree $T_i \setminus {T_{v_1}} \ldots \setminus{T_{v_k}}$ and let $\cT'=\{T'_1, T'_2\}$.
 It is 
clear that $r_{tr}(\cT') \leq r_{tr}(R)-1$. By
the assumption of minimality on $|\cX|$ we have that $h^{0}(\cT') = 
r_{tr}(\cT')$. 
Now, let $N^{*}$ be a network with a minimum reticulation number that 
displays some refinement of each tree in $\cT'$. 
Observe that there
exists a binary tree $T_{S}$ on taxa set $S$ such that  $T_{S}$ is a binary 
refinement of both $T_1|S$ and $T_2|S$, otherwise
$R|S$ would not be an SBR (i.e. it would contain at least one 
reticulation). 

But we can obtain 
a network with $r(N^{*})+1$ reticulations
that displays some  refinement of $T_1$ and $T_2$ by hanging $T_{S}$ back below 
a single new reticulation in $N^{*}$
such that the tails of the two reticulation edges extend the embeddings of 
the refinements of $T'_1$ and $T'_2$ in $N^{*}$. Hence $h^{0}(\cT) \leq r_{t}(\cT') +1 \leq r_{t}(\cT) < h^{0}(\cT)$,
contradiction. 
\end{proof}

\noindent
The following lemma extends a result from \cite{twotrees}:

\begin{lemma}
\label{lemma:twotrees-cor2extended}
If $\cT$ consists of  two not necessarily binary phylogenetic trees on the same set of taxa, $r_{tr}(\cT) = h^{0}(\cT)=r(Cl(\cT))$ .
\end{lemma}
\begin{proof}
In \cite{twotrees} it is proven that for a not necessarily binary set of trees $\cT$ on the same set of taxa, $r_{tr}(\cT) \leq r(Cl(\cT))$.
Now, if a network displays some refinement of a not necessarily binary tree $T$, then it represents all the clusters in $Cl(T)$ (and
possibly more). Hence we also have that $r(Cl(\cT)) \leq h^{0}(\cT)$. Combining this with Theorem \ref{theo:twotrees-1extended}
gives the result.
\end{proof}

Combining all these results we finally obtain the following theorem.

\begin{theorem}
\label{theo:casshybridnumer}
Let $\cT=\{T_1,T_2\}$ be two not necessarily binary trees on $\cX$, and let $\cC = Cl(\cT)$.
When given $\cC$ as input, {\cassinden} computes a network 
$N$ such that $r(N)=h^{0}(\cT)=h^{-}(\cT)$.
\end{theorem}

\section{Conclusion and open problems}

The largest open problem emerging from this article is whether there exists a ``reasonable'' polynomial-time algorithm for
constructing phylogenetic networks of bounded reticulation number (or level) that represent a given set of clusters. The result in Section \ref{sec:theory}, which shows
a theoretical polynomial-time algorithm for constructing networks of bounded level, does not lend itself to a real-world implementation.
On the other hand we have seen that for clusters obtained from \emph{binary} trees a relatively simple and efficient algorithm can be used. At the moment
however there is no reasonable polynomial-time algorithm for general cluster sets i.e. those obtained from sets of potentially non-binary
trees. We have shown that {\cass}, which has a reasonable running time, is in general not optimal (although we had to explicitly engineer a highly
synthetic counter-example to determine this). {\cass} is, however, an \emph{extremely} greedy algorithm, in the sense that at every iteration
it assumes that \emph{all} maximal ST-sets are below cut-edges. Can we relax this assumption in some way to yield a slightly less greedy
version of {\cass} that is optimal? A less urgent problem, but nevertheless very interesting, is the question whether {\cass} is optimal for clusters obtained from exactly three
binary trees on $\cX$.


\appendix

\section{Appendix \label{AppendixA}} 

\subsection{Proofs deferred to the appendix \label{app:severalProof}}

For a tree-node $u$ with set of
children $Q$, we say that a $Q'$-\emph{refinement}, where $Q' \subset Q$ and $|Q'| \geq 2$, is the
network obtained by deleting all edges between $u$ and elements of $Q'$, adding a new node $u'$,
adding the edge $(u,u')$, and adding edges between $u'$ and each element of $Q'$. We say that a
tree-node can be refined if there exists some $Q'$-refinement of it. Note that if a network $N$
represents a set of clusters $\cC$, then the network $N'$ obtained by refining some tree-node of $N$
still represents $\cC$.\\
\\
\textbf{Observation \ref{obs:nocutfreedom}.} \emph{Let $\cC$ be a 
separating set of clusters on $\cX$. Let
$N$ be any network that represents $\cC$. Then each node of $N$ has at most one leaf child and for
each cut-edge $(u,v)$ in $N$, $|\cX(v)|=1$ or $\cX(v)=\cX$}.
\begin{proof} Suppose
$N$ contains a cut-edge $(u,v)$ such that $1 < |\cX(v)| < |\cX|$. Let $C$ be any cluster represented
by $N$. Then either $C \cap \cX(v) = \emptyset$, $\cX(v) \subseteq C$ or $C \subseteq \cX(v)$,
because $(u,v)$ is a cut-edge. Hence $\cX(v)$ is unseparated, contradiction. Now, suppose some
node of $N$ has two leaf-children $x \neq y \in \cX$. Then every non-singleton cluster that contains $x$, also
contains $y$, and vice-versa, meaning that $\{x,y\}$ is an unseparated set, contradiction.
\end{proof}

\noindent
\textbf{Lemma \ref{lem:simpleexists}. }\emph{Let $\cC$ be a  
separating set of clusters on $\cX$. Let $N$ be
any network that represents $\cC$. Then there exists a simple network $N^{*}$ with at most one
leaf-child per node such that $\ell(N^{*}) \leq \ell(N)$}.

 \begin{proof} Clearly
Observation \ref{obs:nocutfreedom} holds for $N$. We will transform $N$ to obtain $N^{*}$. We
repeatedly and arbitrarily apply any of the following four operations until any of them can be applied. (a) If $N$ contains a node $u$ that has indegree and outdegree both larger than 1 then
create a new node $u'$, add the edge $(u,u')$ and move the tails of all edges that leave $u$, to
$u'$. (b)  If $N$ contains a cut-edge $(u,v)$ such that $\cX(v) = \cX$, then keep all nodes and
edges reachable from $v$ by directed paths, take $v$ as a new root, and discard the rest of the
network. (c) If $N$ contains a cut-edge $(u,v)$ such that $\cX(v)=\{x\}$ where $x \in \cX$, but $v$
is not labelled by $x$, then delete all nodes and edges reachable from $v$ by directed paths (but
not $v$), and label $v$ with taxon $x$. (d) If some tree-node of $N$ can be refined to create a
cut-edge $(u,u')$ such that $u'$ is not a leaf, then do that.

Note that after applying each operation the resulting network still represents $\cC$. Furthermore,
and less obviously, the operations do not raise the level of the network. This is clear for the ``deletion''
operations (b) and (c) but for operations (a) and (d) the critical observation is that two
biconnected components of $N$ overlap on at most one node, and in particular are edge disjoint.
Let $N^{*}$ be network obtained after no more operations (a)-(d) can be applied. Clearly, $N^{*}$
represents $\cC$, and Observation \ref{obs:nocutfreedom} applies to it. If $N^{*}$ is simple we are
done. If $N^{*}$ is not simple then it must contain a cut-node $u$ which (when removed) disconnects
$N^{*}$ into two or more non-trivial components, because $N^{*}$ 
 does not contain
any cut-edges with this property by Observation  \ref{obs:nocutfreedom}. Let $P = \{ p_1, \ldots, p_i\}$  ($i \geq 0$) be the parent
nodes of $u$  and let $Q = \{q_1, \ldots, q_j\}$ $(j \geq 1)$ be the children of $u$. For a
node $v \neq u$ let $R_{u}(v)$ be the set of nodes connected to $v$ by an undirected path,
after deletion of node $u$. We know that deleting $u$ splits $N^{*}$ into at least two non-trivial
components i.e. components that contain at least one edge. Observe that each such component will be
the union of one or more of the $(i+j)$ $R_{u}(v)$ sets obtained by taking $v \in P \cup Q$.
Secondly, it is useful to note that $R_u(p_1)=R_u(p_2)=...=R_u(p_i)$. This follows because in a
phylogenetic network every node can be reached from the root by some directed path, and hence the
parents of $u$ can still reach each other with undirected paths (via the root if necessary) after
deletion of $u$. So at least one of the non-trivial components created by deletion of $u$ is equal
to $\cup_{v \in Q'} R_u(v)$ where $Q' \subseteq Q$. Consider the case that $Q' = Q$. Then we cannot
have that $i=0$, because this would mean that $u$ is not a cut-node. So $i \geq 1$. If $i=1$ then we
conclude that $(p_1, u)$ is a cut-edge in $N^{*}$ and thus that $u$ is a leaf, contradicting the
assumption that $j \geq 1$. If $i \geq 2$ then $j=1$ (because all reticulations in $N^{*}$ have
outdegree 1) and hence we can conclude that $(u,q_1)$ was a cut-edge in $N^{*}$. This means that in
$N^{*}$, $q_1$ is a leaf labelled by a single taxon and hence that $R_u(q_1)$ is a trivial component.
But this contradicts the fact that at least two non-trivial components are created by deletion of
$u$.

Consider finally the case that $Q' \subset Q$.  Suppose $|Q'|=1$ and let (wlog) $q_1$ be the only
element of $Q'$. Then the edge $(u,q_1)$ is actually a cut-edge in $N^{*}$, meaning that $q_1$ is a
leaf labelled by a single taxon and (again) that $R_u(q_1)$ is a trivial component, contradiction.
Hence $|Q'| \geq 2$ and $j \geq 3$. Note that we can construct a new network $N^{**}$ which still
represents $\cC$ by performing a $Q'$-refinement on $u$.  Furthermore, the edge $(u,u')$ created by
the refinement must be a cut-edge. But then we could have applied a type-(d) operation to $N^{*}$,
contradicting the assumption that no more applications of operations (a)-(d) were possible.

Hence we conclude that $N^{*}$ is a simple network with at most one leaf-child per node that
represents $\cC$ and, because the four operations (a)-(d) do not increase level, $\ell(N^{*}) \leq \ell(N)$.
\end{proof}

\noindent
\textbf{Lemma \ref{lem:transfBinary}.} \emph{Let $N$ be a phylogenetic network on $\cX$. Then we can
transform N into a binary phylogenetic network $N'$ such that $N'$ has the same reticulation number
and level as $N$ and all clusters represented by $N$ are also represented by $N'$.}

\begin{proof} The following transformation is similar to one described in Lemma 2 of
\cite{twotrees}. (a) For each reticulation node $u$ with outdegree 2 or higher,  introduce a new
node $u'$, add an edge $(u,u')$ and move the tails of edges that leave $u$, to $u'$. This ensures
that all reticulation nodes have outdegree 1. (b) For each reticulation node $u$ with indegree
$d$ ($d \geq3$), we can replace $u$ with a chain of $(d-1)$ reticulation nodes of indegree 2. For
example, if $u$ has 3 parents $p_1, p_2, p_3$ and child $q_1$ we delete $u$, add 2 new nodes
$u_1, u_2$ and new edges $(p_1, u_1), (p_2, u_1), (u_1,u_2), (p_3,u_2),  (u_2, q_1)$. To see that
steps (a) and (b) does not increase the level of the network recall that biconnected components
intersect on at most one node, and (as observed in the proof of Lemma \ref{lem:simpleexists}) a
reticulation node and all its parents are always in the same biconnected component.

(c) The only nodes we still need to consider are tree-nodes. Let $u$ be a tree-node with
outdegree $d$ ($d \geq 3$). If at least one child of $u$ is a leaf, but not all children of $u$ are
leaves, let $q_1$ be a leaf-child of $u$ and $q_2$ be a non-leaf child of $u$. In this case we
delete the edges $(u, q_1), (u, q_2)$, add a new node $u'$, and add new edges $(u,u'), (u',q_1),
(u',q_2)$.  We repeat this process until all tree-nodes $u$ with outdegree 3 or higher are such that
either all their children are leaves, or none of their children are leaves. For each tree-node $u$
with outdegree 3 or higher whose children are all leaves, let $\cX' = \cX(u)$, delete all children
$Q$ of $u$  and identify $u$ with the root of an arbitrary binary tree on taxa set $\cX'$. Now, all
remaining tree-nodes with outdegree 3 or higher will be such that none of their children are leaves.
The main subtlety here is to avoid a situation where, by refining a tree-node in the wrong way, we
merge two biconnected components and thus raise the level of the network. To avoid this we apply the
following operation as often as possible: (d) Let $u$ be any tree-node with outdegree 3 or higher
such that there exists a biconnected component $K$ which contains $u$ and at least two, but not all,
of the edges leaving $u$. Let $Q'$ be the children of $u$ that $K$ contains.  Perform a
$Q'$-refinement on $u$.

After no more applications of operation (d) are possible there is no danger of merging biconnected
components by refining tree-nodes so 
we can apply the last refinement step: (e) we replace each remaining tree-node $u$ with outdegree $d$
($d \geq 3$) by a chain of $(d-1)$ tree-nodes with outdegree 2. For example, if $u$ has 3 children
$q_1, q_2, q_3$ then we delete $u$, create two new nodes $u_1, u_2$, add edges $(u_1,q_1),
(u_1,u_2), (u_2,q_2), (u_2,q_3)$ and (if $u$ had a parent) add an edge from the former parent of $u$
to $u'$. This completes the transformation. \end{proof}


\noindent
\textbf{Lemma \ref{lem:stcompute}.} \emph{ Let $\mathcal{C}$ be a set of clusters on $\cX$ and let $S_1 \neq S_2$ be two ST-sets of $\mathcal{C}$. If $S_1 \cap S_2 \neq \emptyset$ then $S_1 \cup S_2$ is an ST-set.}
\begin{proof} 
If $S_1 \subseteq S_2$ or $S_2 \subseteq S_1$ then we are done, so we assume neither such condition holds. Now, suppose that some cluster $C \in \cC$ is incompatible with $S_1 \cup S_2$. We will derive a contradiction. Note that $C$ is compatible with  $S_1$ and $S_2$ since  
both $S_1$ and $S_2$ are ST-sets. Then we cannot have that $C \cap S_1 = \emptyset$ or 
$C \subseteq S_1$ because then $C$ would  be compatible with $S_1 \cup S_2$. It follows that 
$S_1 \subseteq C$.
With the same reasoning we obtain that 
$S_2 \subseteq C$. So $S_1 \cup S_2 \subseteq C$. Hence $C$ is compatible with $S_1 \cup S_2$, a contradiction. 
It remains to show is that any pair of clusters in $\cC | (S_1 \cup S_2)$ are compatible. Consider any cluster $C$ in $\cC | (S_1 \cup S_2)$. Obviously, $C \subseteq S_1 \cup S_2$ and there are three possibilities: (i) $C = S_1 \cup S_2$, (ii) $C \subseteq S_1 \cap S_2$, or (iii) $C \cap (S_1 \cap S_2) = \emptyset$. To see this, suppose there exists a cluster $C \in \cC | (S_1 \cup S_2)$ that violates all three possibilities. Then, without loss of generality, there is by violation of (ii) some $x \in S_1$ such that $x\not \in S_2$ and $x \in C$.
 There is also, by violation of (iii), some $x' \in S_1 \cap S_2$ such that $x' \in C$. This means that $S_1,S_2 \cap C \neq \emptyset$. 
Let $C'$ be a cluster in $\cC$  such that $C' | (S_1 \cup S_2) = C$.
 It cannot be that $C' \subseteq S_2$ (because of $x$) so $S_2 \subseteq C'$ (because $S_2$ is an ST-set in $\cC$
and $S_2 \cap C \neq \emptyset$).  Now, because (i) was violated it follows that there is some $x'' \in S_1$ such that $x'' \not \in C$ and thus $x'' \not \in C'$, 
so $S_1 \not \subseteq C'$. Now, since $S_2 \subseteq C'$ 
and we assumed that  $S_2 \not \subseteq S_1$, we have that  $C' \not \subseteq S_1$. But $C' \cap S_1 \neq \emptyset$, so $C'$ is incompatible with $S_1$ and  $S_1$ is not an ST-set, a contradiction. Hence (i)-(iii) are the only possibilities.

Now, suppose there are two clusters $C_1, C_2 \in \cC | (S_1 \cup S_2)$ which are incompatible. Let $C'_1$ and $C'_2$ be the clusters in $\cC$ such that $C'_1 | (S_1 \cup S_2) = C_1$ and $C'_2| (S_1 \cup S_2) = C_2$. 
Now, clearly type (i) clusters cannot cause incompatibilities. Furthermore, type (ii) and (iii) clusters are disjoint, so $C_1, C_2$ are both type (ii) clusters \emph{or} are both type (iii) clusters. Suppose they are both type (ii) clusters i.e. both are entirely contained inside $S_1 \cap S_2$. If $C'_1 = C_1$ and $C'_2 = C_2$ then we obtain a contradiction, since $S_1 \cap S_2 \subset S_1$. Specifically, we have that $C'_1, C'_2 \subset S_1$ and thus $S_1$ cannot be an ST-set. 
Hence, without loss of generality, we assume that $C'_1$ contains an element $x \not \in S_1 \cup S_2$. But then $C'_1 \not \subseteq S_1$, $S_1 \not \subseteq C'_1$ (because $S_1 \cap S_2 \subset S_1$) and $S_1 \cap C'_1 \neq \emptyset$. Again, we conclude that $S_1$ is not an ST-set. 
The final case is that both $C_1$ and $C_2$ are type (iii) clusters. Clearly either $C_1, C_2 \subseteq S_1 \setminus S_2$ or $C_1, C_2 \subseteq S_2 \setminus S_1$, otherwise  $C_1$ and $C_2$ would be compatible. 
Assume without loss of generality that $C_1, C_2 \subseteq S_1 \setminus S_2$. Since $S_1 \setminus S_2  \subset S_1$,
 if $C'_1 = C_1$ and $C'_2 = C_2$  then $S_1$ cannot be an ST-set, a contradiction. So assume $C'_1$ contains an element $x \not \in S_1 \cup S_2$. So $C'_1 \not \subseteq S_1$. Furthermore, $S_1 \not \subseteq C'_1$ because $S_1 \cap S_2 \neq \emptyset$. But $S_1 \cap C'_1 \neq \emptyset$, so $S_1$ is not an ST-set, again a contradiction.  \end{proof}


\noindent
\textbf{Lemma \ref{lem:polymax}.} \emph{The maximal ST-sets of a set of clusters $\cC$ on $\cX$ can be computed in polynomial time..}
\begin{proof}
Start with a set $\cS$ of $n$ singleton ST-sets. If there are two distinct ST-sets $S_1, S_2 \in \cS$ such that $S_1 \cup S_2$ is an ST-set, then remove $S_1$ and $S_2$ from $\cS$ and add $S_1 \cup S_2$ to $\cS$. Repeat this until it is no longer possible. This results in a set $\cS$ of ST-sets that partitions $\cX$. Let $\mathcal{M}$ be the set of maximal ST-sets of $\cC$. We will show that $\cS = \mathcal{M}$.

Observe that for each $S \in \cS$ and each $M \in \mathcal{M}$, either $M \cap S = \emptyset$ or $S \subseteq M$. Indeed, from Lemma \ref{lem:stcompute}, if some $S$ and $M$ were incompatible then $S \cup M$ would be an ST-set too, contradicting the maximality of $M$. So each $M \in \mathcal{M}$ is partitioned by one or more elements from $\cS$. Now, suppose some $S \in \cS$ is not maximal. Then there exists some $M \in \mathcal{M}$ such that $S \subset M$. Furthermore, there must also exist some $S' \neq S$ such that $S' \subset M$ i.e. $M$ is partitioned by at least two elements from $\cS$. Hence we can write $M = S_1 \cup \ldots \cup S_{k}$ where $k \geq 2$ and each $S_i \in \cS$. We denote by $\cS'$ the set of ST-sets of $\cS$ partitioning $M$. Now, consider the set of clusters $\cC | M$. 
All clusters in $\cC | M$ are mutually compatible, because $M$ is an ST-set.
Furthermore, every cluster in $\cC | M$ is compatible with every $S_i$ contained inside $M$
because the $S_i$ are ST-sets 
and $\cS'$ partitions $M$.


Now, for a cluster $C \in
\cS$, let $nb(C)$ be the number of  $S_i \in \cC|M$
that
it contains. Let $b = min \{ nb(C) |C \in \cC|M$ and $nb(C) \geq 2 \}$. Suppose
$b$ is not defined. 
Let  $S_i \neq S_j$ be {\em any} two ST-sets  of $\cS'$. We claim that $S_i \cup S_j$ is an ST-set.
First, we need to prove that $S_i\cup S_j$ is not separated by $\cC$. Let us suppose this is not true and there exists a cluster $C \in \cC$ such that 
 $C  \cap (S_i \cup S_j ) \neq \emptyset$, $C \not\subseteq S_i \cup S_j$  and  $S_i \cup S_j \not \subseteq C$.  It follows that $C \not\subseteq S_i$ and $C \not\subseteq S_j$. Moreover, since $S_i \cup S_j \not \subseteq C$, $C$ does not contain at least one among $S_i$ and  $S_j$, say $S_i$. Then $C  \cap S_i = \emptyset$, otherwise $S_i$ is not an ST-set. Then we have that $C  \cap S_j \neq \emptyset$. If $ S_j \not \subseteq C$ we have again a contradiction. So we must have $ S_j  \subset C$ and $C  \cap S_i = \emptyset$. But since $C$ cannot contain more than one ST-set of $\cS'$, $C$ separates $M$, a contradiction. It follows that $S_i\cup S_j$ is not separated by $\cC$. Moreover, 
since no cluster of $\cC|M$ contains more that one ST-set of $\cS$, each cluster in $\cC|M$ is entirely contained inside some single ST-set of  $\cS'$. It follows that $\cC| (S_i \cup S_j)$
is a compatible set of clusters 
because in this case $\cC| (S_i \cup S_j) = \cC|S_i \cup \cC|S_j$ and $S_i \cap S_j = \emptyset$.
Then $S_i \cup S_j$ is an ST-set, contradicting the termination condition for the algorithm. 

Now, suppose that $b$ is well defined. Let $C$ be any cluster in $\cC|M$
such that $nb(C) = b$. Let $S_i \neq S_j$ be any two ST-sets in $\cS'$
that are contained in $C$. We claim that $S_i \cup S_j$ is an ST-set.
We first show that $S_i \cup S_j$ is unseparated. Suppose this is
not true. Then there exists $C' \in \cC$ such that $C' \cap (S_i \cup S_j)
\neq \emptyset$, $C' \not \subseteq (S_i \cup S_j)$ and $(S_i \cup S_j)
\not \subseteq C'$. With the same argument as before we
have (without loss of generality) that $S_j \subset C'$ and $C' \cap S_i =
\emptyset$. This means that $M \not \subseteq C'$. Combined with the
fact that $C' \cap M \neq \emptyset$ and that $M$ is a maximal ST-set
of $\cC$, we have that $C' \subseteq M$ and thus $C' \in \cC|M$. $C'$ and $C$ are compatible
(because all clusters in $\cC|M$ are mutually compatible), so combining
that $C' \cap C \neq \emptyset$ and that $C \not \subseteq C'$, we have
that $C' \subset C$. Observe that $C'$ is partitioned by ST-sets from
$\cS'$. Now, we cannot have that $nb(C') \leq 1$ because
$S_j$ is a \emph{strict} subset of $C'$. So $2 \leq nb(C') < b$,
contradiction. 
It remains only to prove that all the clusters in $\cC | (S_i \cup S_j)$
are mutually compatible. $S_i \cup S_j$ is unseparated by $\cC$ so
for each $C \in \cC$ we have that $C \cap (S_i \cup S_j) = \emptyset$,
$S_i \cup S_j \subseteq C$ or $C \subseteq S_i \cup S_j$. For each
$C \in \cC$ such that $C \subseteq S_j \cup S_j$ we have either
that $C = S_i \cup S_j$, $C \subseteq S_i$ or $C \subseteq S_j$ (because
$C$ is compatible with both $S_i$ and $S_j$). Since $S_i \cap S_j = \emptyset$,  it is not possible
for $\cC | (S_i \cup S_j)$ to contain two incompatible clusters, and again we conclude
that $S_i \cup S_j$ is an ST-set.


\end{proof} 

\noindent
\textbf{Lemma \ref{lem:moveST}.} 
\emph{Let $N$ be a network that represents a set of clusters $\cC$. Let $S$ be a non-trivial ST-set
with respect to $\cC$. Then there exists a network $N'$ such that $r(N') \leq r(N)$, $\ell(N') \leq \ell(N)$,
$S$ is under a cut-edge in $N'$ and for each ST-set $S'$ such that $S' \cap S = \emptyset$ and $S'$ is
under a cut-edge in $N$, $S'$ is also under a cut-edge in $N'$.}
\begin{proof}
We obtain $N'$ from $N$ by the following transformation. Let $x$ be any element of $S$. Recall that leaves
in $N$ always have indegree-1 and outdegree-0, so let $(u,v)$ be the edge in $N$ such that $v$ is
labelled by $x$. This is trivially a cut-edge. (a) Delete in $N$ all taxa in $S$
(but not the leaves they label, we will deal with this in step (c)).
(b) Identify $v$ with the root of a tree $T_{S}$ on $S$ that represents $\cC|S$. (c) Tidy up redundant parts of the
network possibly created in step (a) by applying in an arbitrary order
any of the following steps until no more can be applied: deleting any nodes with outdegree-0 that are not labelled
by a taxon; suppressing any nodes with indegree-1 and outdegree-1; replacing any multi-edges with a single edge; deleting
any node with indegree-0 and outdegree-1.
This concludes the transformation. Let $N'$ be the resulting network. To see that $N'$ represents $\cC$ consider
any $C \in \cC$. If $C \subseteq S$ then clearly some edge in the tree we added in step (b) represents
$C$. Otherwise there are only two possibilities (because $S$ is unseparated): $S \cap C = \emptyset$ or $S \subset C$.
In either case we know that there exists some tree $T$ on $\cX$ which is displayed by $N$ and such that $T$
represents $C$. By applying steps (a)-(c) simultaneously to $T$ we obtain a new
tree $T'$ that is displayed by $N'$ and which represents $C$. The critical reason this works is that, if
$C \not \subseteq S$, $x \in C$ (where $x$ is the element we selected at the beginning of the proof) if and only if $S \subseteq C$.
The tidying-up operations in step (c) clearly do not raise the reticulation number or the level of the network.
Furthermore they leave untouched any taxa not in $S$. So any other ST-sets that were already under a cut-edge
in $N$, are also under a cut-edge in $N'$.
\end{proof}

\subsection{Proof that {\cass} does not construct a simple level-$\leq 
3$ network for the input $\cC$ described in Figure \ref{fig:4t}. \label{app:proofCassBreaks}}

Let us begin with the following observation.

\begin{observation}
\label{obs:atleast3}
If a phylogenetic network $N$ represents $\cC$, then $r(N) \geq 3$. 
\end{observation}
\begin{proof}
First, note that $r(N) \geq 2$. This follows
because the clusters $\{8,5\}, \{1,5\}, \{5,6\}$ are
all in $\cC$. A network with reticulation number $r$ can display at
most $2^{r}$ mutually non-isomorphic trees, and each tree displayed
by the solution $N$ can represent at most one of those clusters. Now, suppose 
some network $N$ exists that
represents $\cC$ such that $r(N)=2$. All SBRs of $N$ are single leaves, 
because otherwise $\cC$ would not be  
separating
. Now, observe that
at least one of the leaves $\{1,5,6,8\}$ must be an SBR. Suppose this is not
so, and let taxon $i \not \in \{1,5,6,8\}$ be a single leaf SBR. But then
removing leaf $i$ from $N$ and tidying up the resulting network would give us a network $N'$ with
$r(N') \leq 1$ such that $N'$ displays the three clusters described
at the beginning of the proof; contradiction. So suppose 1 is the SBR.
If we remove taxon 1 then the remaining network $N'$ has $r(N') \leq 1$ and
represents $\{8,5\}$, $\{5,6,7\}$ and $\{3,4,5\}$, contradiction. If
taxon 5 is the SBR then the same argument holds for $\{1,2\}$, $\{3,4,1\}$
and $\{7,6,1\}$. (In fact, taxon 5 is symmetrical to taxon 1). If taxon
6 is the SBR then clusters $\{5,8\}$, $\{5,7\}$, $\{1,3,4,5\}$ cause
the contradiction. If taxon 8 is the SBR then clusters $\{1,3,4\}$, $\{1,2\}$,
$\{1,5\}$ cause the contradiction. 
\end{proof}  

\noindent
To prove the overall result we thus exhaustively consider all maximal ST-set
tree sequences of length exactly 3. These sequences can either be
derived by hand, or generated computationally. Many maximal ST-set
sequences of length 3 do not remove all incompatibilities in the clusters
i.e. they are not tree sequences. We can thus automatically exclude these
sequences from our analysis (because {\cass} will immediately conclude that it
is not possible to construct a simple level-$\leq 3$ network that represents $\cC$ from such a sequence). Additionally, many of the \emph{tree} sequences
can be excluded by using an argument similar to that used
in Observation \ref{obs:atleast3}. In particular, if two maximal ST-sets
have already been removed, but the input still contains three mutually
incompatible clusters, then we discard this tree sequence. This is because
any network that represents the remaining clusters (i.e. those obtained after removing the first
two maximal ST-sets) must have reticulation number at least 2, and subsequently hanging back the first two maximal ST-sets below reticulations raises the reticulation
number of the network to at least 4.  In other words, all attempts of {\cass} to construct
a simple level-$\leq 3$ network representing $\cC$ using this maximal ST-set tree sequence are doomed to fail.

It turns out that after applying these two filtering rules\footnote{The automated case-analysis
can be downloaded from \url{http://skelk.sdf-eu.org/clustistic/caseanalysis.txt}} there only remains a very
small number of cases to consider. These are :

\begin{enumerate}
\item $(\{1\}, \{5\}, \emptyset)$
\item $(\{8\}, \{1\}, \{5\})$
\item $(\{1\}, \{8\}, \{5\})$
\item $(\{5\}, \{1\}, \emptyset)$
\item $(\{2\}, \{5\}, \{1\})$
\item $(\{5\}, \{2\}, \{1\})$
\end{enumerate}

Note that, because of symmetries in the cluster set, sequence 4 is entirely symmetrical to sequence 1, sequence 5 is entirely symmetrical to sequence 2 and  sequence 6 is entirely symmetrical to sequence 3. Hence we can actually restrict
our attention to only sequences 1-3.\\

\begin{figure}[h]
  \centering
  \includegraphics[scale=.27]{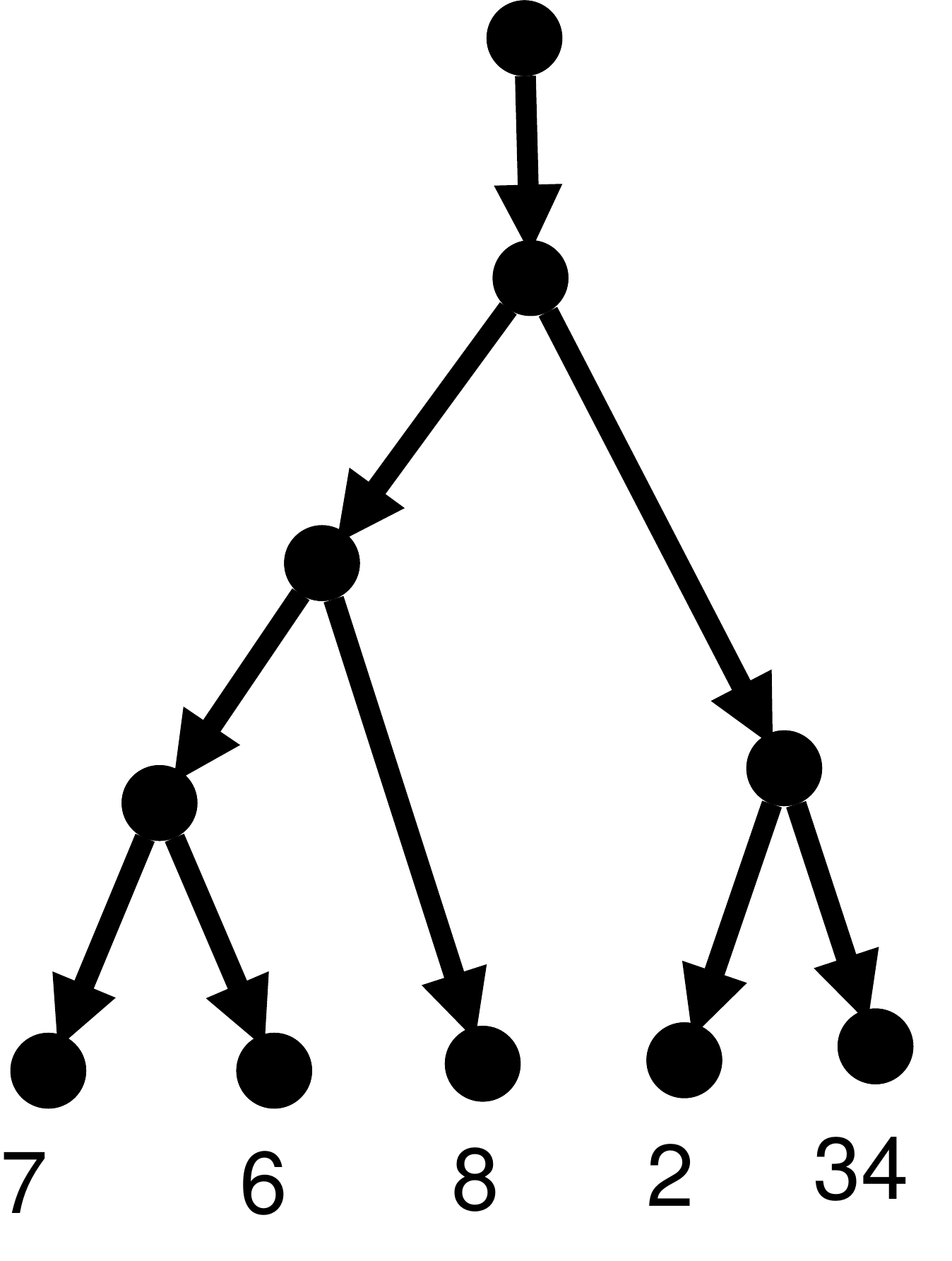}
  \caption{The ``base tree'' $T$  in the case $(\{1\}, \{5\}, \emptyset)$.}
  \label{fig:basetree}
\end{figure}

\textbf{Case} $(\{1\}, \{5\}, \emptyset$)\\
\\
We can see already that this is a potentially important
tree sequence because it is closely linked to the network shown in Figure \ref{fig:4tL3L4}(a). Specifically:
remove taxon 1, causing the reticulation number to drop by 1, and then remove taxon 5 which causes the
reticulation number to drop by 2 (because one of the two remaining reticulations is a child of the other). Here
we show why {\cass} nevertheless cannot reconstruct the network.

Before removing taxon 1 the maximal ST-sets are $\{i\}$, $1 \leq i \leq 8$. After removing taxon 1 the maximal 
ST-sets are $\{2\}$, $\{3,4\}$, $\{5\}$, $\{6\}$, $\{7\}$, $\{8\}$. $\{3,4\}$ is a (new) non-singleton maximal 
ST-set so the {\cass} algorithm collapses taxa 3 and 4 into a new meta-taxon, let us call this $34$ for simplicity. 
After collapsing the maximal ST-sets are thus $\{2\}$, $\{34\}$, $\{5\}$, $\{6\}$, $\{7\}$, $\{8\}$. Subsequently 
taxon 5 is removed and there are no more incompatible clusters in the input. We construct the unique tree 
$T$ that represents exactly the remaining clusters (and add a dummy root): see Figure \ref{fig:basetree}. For each
taxon $i$ in $T$ let $e_i$ be the edge in $T$ that feeds into it.

Note how meta-taxon 34 is still collapsed. The constructive phase of {\cass} will now proceed as follows: (a) try 
and hang back a dummy taxon (i.e. the empty set $\emptyset$) from $T$ below a reticulation; (b) try and hang taxon 
5 back below a second reticulation; (c) decollapse meta-taxon 34; (d) try and hang taxon 1 back below a third 
reticulation. Now, after decollapsing meta-taxon 34 we have a network $N'$ on taxa set $\{2,3,4,5,6,7,8\}$ with a 
cherry on taxa $\{3,4\}$. We argue that in steps (a) and (b) certain edges of $T$ will definitely
have been subdivided by the tails of reticulation edges. Note that $e_6$ 
has definitely been subdivided to make cluster $\{5,6\}$ possible. Similarly edge $e_8$ will have been
subdivided to make cluster $\{5,8\}$ possible. And edge $e_{34}$ will have been subdivided to make
cluster $\{5,34\}$ possible. Steps (a) and (b) can subdivide at most four edges of $T$, but observe
that if four edges are subdivided then one of the two reticulations in $N'$ wil not have been ``used'' i.e. there
will only hang a dummy taxon beneath it. Hence $N'$ cannot simultaneously represent $\{5,6\}$, $\{5,8\}$ and $\{5,3,4\}$, because a network with reticulation number at least 2 is required to resolve these incompatibilities.
Steps (c) and (d) cannot remedy this so we can exclude the case that four edges of $T$ have been subdivided. Hence we can safely conclude that edge $e_2$ of $T$ has not yet been subdivided in $N'$.

Consider step (d). We have to hang back taxon 1 beneath a single reticulation i.e. a 
reticulation with two incoming edges. Observe that at least one of those two reticulation edges subdivides (i.e. 
starts from) the edge entering taxon 3, otherwise the cluster $\{1,3\} \in \cC$ cannot be represented. 
Now, note also that $\{1,2\} \in \cC$. To obtain this cluster the edge $e_2$ has to be subdivided, because
otherwise every cluster that contains both 1 and 2 will also contain 3 and 4. But if both the two reticulation 
edges are used in this way, then any cluster that contains a taxon from $\{6,7,8\}$ and contains taxon 1, must also 
contain $\{2,3,4\}$. This is not satisfactory, however, because cluster $\{1,5,6,7\}$ is in $\cC$. In other words: 
at least three reticulation edges are needed to hang back taxon 1, yielding a contradiction. 

It will be clear from this example that the problem with {\cass} is that it collapses maximal ST-set
$\{3,4\}$, forcing the iteration in which taxon 1 is hung back to use at least 3 reticulation
edges instead of 2.

\begin{figure}[h]
  \centering
  \includegraphics[scale=.25]{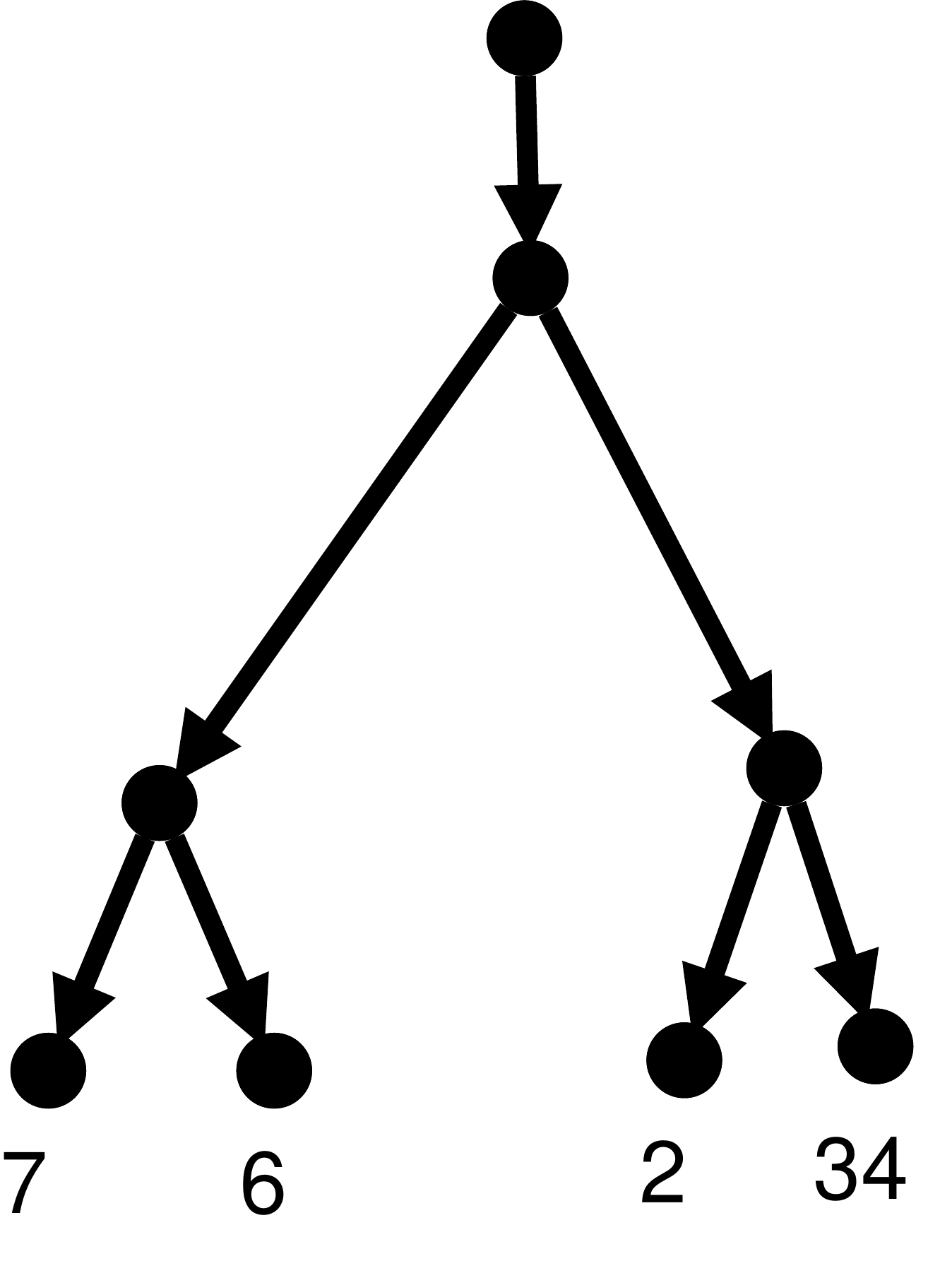}
  \caption{The ``base tree'' $T$ for the cases $(\{8\}, \{1\}, \{5\})$ and $(\{1\}, \{8\}, \{5\})$ .}
  \label{fig:basetree815}
\end{figure}

\noindent
\textbf{Case $(\{8\}, \{1\}, \{5\})$}\\
\\
Let $T$ be the tree obtained after removing these three maximal ST-sets. After removing the maximal ST-set $\{1\}$ taxa 3 and 4 will, as in the previous case, have been
collapsed into the meta-taxon 34. Hence $T$ will look as in Figure \ref{fig:basetree815}. As in the previous case let $e_i$ be the edge feeding into taxon $i$.
Observe that to hang back maximal ST-set $\{5\}$ it is essential to subdivide $e_6$ and $e_{34}$; this is the only way to ensure that the clusters $\{5,6\}$ and $\{5, 34\}$ are represented. Let $N'$ be the network obtained by doing this and subsequently expanding meta-taxon 34 into a cherry on taxa $\{3,4\}$. Now, we need to hang maximal ST-set $\{1\}$ back from $N'$ increasing the reticulation number exactly by one. 
Note that in any case the edge entering
taxon $3$ in $N'$ will have to be subdivided (to represent the cluster $\{1,3\}$), and also the edge entering taxon $2$ will need to be subdivided (to represent the cluster
$\{1,2\}$).  However, the cluster $\{1,5,6,7\}$ will then definitely not be represented. Hence
we conclude that actually at least three reticulation edges are required to hang back $\{1\}$, and thus that it is impossible to hang back $\{1\}$ increasing the reticulation number exactly by one. \\
\\
\noindent
\textbf{Case $(\{1\}, \{8\}, \{5\})$}\\
\\
The case has the same base tree $T$ as the case $(\{8\}, \{1\}, \{5\})$, and again in this case there is only one way to hang back maximal ST-set $\{5\}$ as an SBR, yielding
the same network $N'$. Now, we first need to hang back maximal ST-set $\{8\}$ in such a way that the reticulation number rises by exactly 1. Let $l$ be the child of the real root of $N'$ which lies on a directed path
from the real root to taxon 6. We say that an edge of $N'$ is a \emph{left edge} if there is a directed path from the real root that contains both the tail and head of the edge
and which also passes through $l$. Observe that when hanging back $\{8\}$ at least one left edge of $N'$ has to be subdivided by a reticulation edge, otherwise the cluster
$\{5,6,7,8\}$ will not be represented by the resulting network $N''$.  Suppose the second reticulation edge (when hanging back $\{8\}$) subdivided neither the edge entering taxon 2, nor the edge entering
taxon 3, in $N'$. But then, when subsequently hanging back $\{1\}$ from $N''$, both these edges will have to be subdivided (to ensure that the clusters $\{1,2\}$ and $\{1,3\}$ are represented), and such a network cannot possibly represent the cluster
$\{1,5,6,7\}$. Suppose then that, when hanging back $\{8\}$, the edge entering taxon 2 in $N'$ was subdivided by a reticulation edge $e$. But when hanging back $\{1\}$ the edge entering
taxon 3 in $N''$ will definitely have to be subdivided, and also the edge entering taxon 2 or (alternatively) $e$. But, again, the cluster $\{1,5,6,7\}$ is not represented by such a network. An essentially identical
argument applies if, when hanging back $\{8\}$, only the edge entering taxon 3 was subdivided. Hence hanging back $\{8\}$ and then $\{1\}$ causes the reticulation number to rise by at least 3, rendering
it impossible to construct a simple level-$\leq 3$ network.


\bibliographystyle{IEEEtran}
\bibliography{elusive_arxiv}

\end{document}